\RequirePackage{amsmath}
\documentclass[runningheads]{llncs}
\usepackage{bbding}
\usepackage[T1]{fontenc}
\usepackage{graphicx}
\usepackage{subcaption}
\usepackage{amssymb}
\usepackage{array}
\usepackage{enumitem}
\usepackage[hidelinks]{hyperref}
\usepackage{booktabs}
\newcolumntype{R}[1]{>{\raggedleft\let\newline\\\arraybackslash\hspace{0pt}}m{#1}}
\usepackage[table]{xcolor}
\definecolor{cBlue}{RGB}{56, 122, 223}
\definecolor{cRed}{RGB}{238, 78, 78}
\definecolor{cPurple}{RGB}{110, 70, 150}
\usepackage{pgfplots}
\usepackage{tikz}
\usetikzlibrary{fit, shapes, decorations.pathmorphing, arrows.meta, calc}
\tikzstyle{state}=[draw=black, text=black, circle, minimum size=0.8cm, inner sep=0pt]
\tikzstyle{smallstate}=[draw=black, text=black, circle, minimum size=16pt, inner sep=0pt]
\tikzstyle{bigstate}=[draw=black, text=black, circle, minimum size=1.1cm, inner sep=0pt]
\tikzstyle{class}=[rectangle, minimum width=0.8cm, minimum height=0.6cm, draw=black]
\tikzset{>={Classical TikZ Rightarrow[scale=2]}, fill fraction/.style n args={2}{path picture={\fill[#1] (path picture bounding box.south west) rectangle ($(path picture bounding box.north west)!#2!(path picture bounding box.north east)$);}}}
\pgfplotsset{compat=1.18}
\usepackage{listings}
\usepackage[ruled,linesnumbered]{algorithm2e}
\SetKwComment{Comment}{/* }{ */}
\spnewtheorem*{proofsketch}{Proof Sketch}{\itshape}{\rmfamily}
\spnewtheorem{claim-num}{Claim}{\itshape}{\rmfamily}

\newcommand{\nat}{\mathbb{N}}
\newcommand{\support}{\mathrm{support}}
\newcommand{\Dreal}{\mathcal{D}}
\newcommand{\Sreal}{\mathcal{S}}
\newcommand{\Creal}{\Omega}
\newcommand{\comment}[1]{\hspace{2em} [\mbox{#1}]}
\newcommand{\robust}{\mathord{\simeq}}

\newcommand{\vertex}{\node[circle, draw=black, inner sep=0pt, minimum size=12pt]}
\mathcode`<="4268
\mathcode`>="5269
\mathchardef\gr="213E
\mathchardef\ls="213C
\newcommand{\mvar}{V}
\newcommand{\pnext}{\textrm{Next}}
\newcommand{\qnext}{\textrm{next}}
\newcommand{\red}{\textit{tails}}
\newcommand{\imax}{\sqcup}
\newcommand{\sem}[1]{[\![#1]\!]}
\newcommand{\valu}{f}
\newcommand{\valus}{\mathcal{F}}
\def\set#1{\{#1\}}
\def\techReport{}

\usepackage{orcidlink}
\def\orcidID#1{\,\orcidlink{#1}}

\ifthenelse{\isundefined{\techReport}}{
\usepackage[firstpage]{draftwatermark}
\SetWatermarkAngle{0}
\SetWatermarkText{\raisebox{12.2cm}{
\hspace{-0.25cm}
\href{https://doi.org/10.5281/zenodo.15175186}{\includegraphics{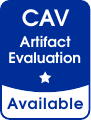}}
\hspace{8.4cm}
\includegraphics{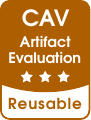}
}}
}{}

\begin{document}

\title{Robust Probabilistic Bisimilarity \\ for Labelled Markov Chains}
\titlerunning{Robust Probabilistic Bisimilarity}

\author{
Syyeda Zainab Fatmi\inst{1}\orcidID{0000-0001-7899-8665} \and
Stefan Kiefer\inst{1}\orcidID{0000-0003-4173-6877} \and
David Parker\inst{1}\orcidID{0000-0003-4137-8862} \and
Franck van Breugel\inst{2}\orcidID{0009-0002-7320-1527}}
\authorrunning{S.~Z.~Fatmi, S.~Kiefer, D.~Parker, F.~van Breugel}
\institute{University of Oxford, Oxford, UK \and York University, Toronto, Canada}

\maketitle

\begin{abstract}
Despite its prevalence, probabilistic bisimilarity suffers from a lack of robustness under minuscule perturbations of the transition probabilities. This can lead to discontinuities in the probabilistic bisimilarity distance function, undermining its reliability in practical applications where transition probabilities are often approximations derived from experimental data. Motivated by this limitation, we introduce the notion of robust probabilistic bisimilarity for labelled Markov chains, which ensures the continuity of the probabilistic bisimilarity distance function. We also propose an efficient algorithm for computing robust probabilistic bisimilarity and show that it performs well in practice, as evidenced by our experimental results.

\end{abstract}


\section{Introduction}

In the analysis and verification of probabilistic systems, one of the foundational concepts is identifying and merging system states that are behaviourally indistinguishable.
Kemeny and Snell \cite{KS60} introduced the notion of lumpability for Markov chains and it was adapted to the setting of labelled Markov chains by Larsen and Skou \cite{LS89}, known as probabilistic bisimulation.
State-of-the-art probabilistic verification tools~\cite{KNP11,HJK+22}
implement a variety of methods for minimizing the state space of the system by collapsing probabilistically bisimilar states.
This can significantly improve verification efficiency in some cases~\cite{KKZJ07}.

However, due to the sensitivity of behavioural equivalences to small changes in the transition probabilities, Giacalone et al. \cite{GJS90} proposed using behavioural pseudometrics to capture the behavioural similarity of states in a probabilistic system.  Instead of classifying states as either equivalent or inequivalent, the pseudometric maps each pair of states to a real value in the unit interval, thus also quantifying the behavioral difference between non-equivalent states.
Behavioural pseudometrics have been studied in the context of systems biology~\cite{ThorsleyKlavins2010}, games~\cite{CAMR10}, planning~\cite{ComaniciP11} and security~\cite{CaiG09}, among others.

\ifthenelse{\isundefined{\techReport}}{}{
\clearpage
}

In probabilistic verification, the most widely studied example of such a behavioural pseudometric is the \emph{probabilistic bisimilarity distance}.
It generalizes probabilistic bisimilarity quantitatively; in particular, the distance between two states is zero if and only if they are probabilistically bisimilar.
The probabilistic bisimilarity distance was introduced by Desharnais et al.~\cite{DGJP99}, based on a real-valued semantics for Larsen and Skou's probabilistic modal logic~\cite{LS89}.
A formula $\varphi$ of this logic maps any state $s$ to a number $\sem{\varphi}(s) \in [0,1]$.
The probabilistic bisimilarity distance between two states $s,t$ can be characterized as $\delta(s,t) = \sup_{\varphi} |\sem{\varphi}(s) - \sem{\varphi}(t)| \in [0,1]$, where $\varphi$ ranges over all formulas.
The lower the distance between two states, the more similar their behaviour.
As shown by Van Breugel and Worrell \cite{BreugelW01}, the probabilistic bisimilarity distance can also be characterized as a fixed point of a function (we use this definition in this paper).

However, as pointed out by Jaeger et al.\ \cite{JMLM14}, probabilistic bisimilarity distances are sometimes not continuous, leading to unexpected and abrupt changes in behaviour between two states when the transition probabilities are perturbed.  Since the probabilities of the labelled Markov chain are usually obtained experimentally and, therefore, are often an approximation \cite{TB17,EH20,MLAK17,OCHTC17,SLT21}, the lack of robustness of probabilistic bisimilarity is a serious drawback.  This inconsistency undermines the reliability of probabilistic bisimilarity as a measure of system equivalence and can be particularly problematic when used in practical applications where approximate models are prevalent.

\begin{example}\label{example:intro-1}
Consider Figure~\ref{fig:example-continuous} on page~\pageref{fig:example-continuous}.
When $\varepsilon = 0$, states $h_0$ and $h_1$ are probabilistically bisimilar; i.e., their distance $\delta_0(h_0,h_1)$ equals~$0$ (the subscript of~$\delta$ indicates~$\varepsilon$).
For $\varepsilon \gr 0$ we have $\delta_\varepsilon(h_0,h_1) \gr 0$; i.e., $h_0$ and $h_1$ are no longer bisimilar.
However, when $\varepsilon$ is small then $\delta_\varepsilon(h_0,h_1)$ is small.
In fact, one can show that $\delta_\varepsilon(h_0,h_1) \le 2 \varepsilon$, which implies that $\lim_{\varepsilon \to 0} \delta_\varepsilon(h_0,h_1) = 0$.
This means that in this example, the distance is continuous in~$\varepsilon$.
One may say that states $h_0, h_1$ are not only probabilistically bisimilar, but also robustly so, in that they remain ``almost'' bisimilar when the transition probabilities are perturbed.
Intuitively, states $h_0$ and $h_1$ behave similarly even for small positive~$\varepsilon$: both states carry a blue label and perform a geometrically distributed number of self-loops (about two in expectation) before transitioning to state~$t$.
\end{example}

\begin{example}\label{example:intro-2}
Consider Figure~\ref{fig:example2-discontinuous} on page~\pageref{fig:example2-discontinuous}.
When $\varepsilon = 0$, states $h_2$ and $h_3$ are probabilistically bisimilar; i.e., their distance $\delta_0(h_2,h_3)$ equals~$0$.
But for any $\varepsilon \gr 0$ we have $\delta_\varepsilon(h_2,h_3) = 1$; i.e., $h_2$ and $h_3$ behave ``maximally'' differently in terms of the probabilistic bisimilarity distance.
We have $\lim_{\varepsilon \to 0} \delta_\varepsilon(h_2,h_3) = 1$; so, in this example, the distance is discontinuous in~$\varepsilon$.
One may say that although states $h_2, h_3$ are probabilistically bisimilar, they are not robustly so, because upon perturbing the transition probabilities the behaviour of~$h_3$ changes completely.
For any positive $\varepsilon$, state~$h_2$ remains in a self-loop forever, whereas $h_3$ eventually reaches the (red-labelled) state~$t_3$ with probability~$1$.
Since reachability properties are at the heart of probabilistic model checking, it may be unsafe to merge states $h_2$ and $h_3$ if the transition probabilities are not known precisely.
\end{example}

In this paper, we address this issue by introducing the notion of \emph{robust probabilistic bisimilarity} for labelled Markov chains.
Robust probabilistic bisimilarity is a particular probabilistic bisimulation, implying that robust probabilistic bisimilarity is a subset of probabilistic bisimilarity.
Crucially, we show that our definition ensures the continuity of the probabilistic bisimilarity distance function.
This means that for any two states that are robustly probabilistically bisimilar, their probabilistic bisimilarity distance remains small even after small perturbations of any transition probabilities.
Note that it is easy to see that the distance from \cite{DLT08} is robust in this sense; on the other hand, states with very small distance in terms of \cite{DLT08} can have very different long-term behaviour, as in Example~\ref{example:intro-2}.

Secondly, we develop a polynomial-time algorithm for computing robust probabilistic bisimilarity.
It is suitable for large-scale verification tasks, opening the door to checking probabilistic models from the literature for (lack of) robustness of their probabilistic bisimilarity relation.
Thus, one can identify pairs of states that may be dangerous to merge if the transition probabilities are not known precisely.
We present experimental results on the applicability and efficiency of an implementation of our algorithm on models from the Quantitative Verification Benchmark Set (QVBS)~\cite{HKPQR19} and the examples included in the Java PathFinder extension jpf-probabilistic \cite{FCDWTB21}.

The rest of the paper is structured as follows. Section~\ref{section:premliminaries} introduces the model of interest, namely a labelled Markov chain, and probabilistic bisimilarity. In Section~\ref{section:distances}, we formally define probabilistic bisimilarity distances and further examine how the bisimilarity distance changes when the transition function is varied. Section~\ref{section:robust} describes robust probabilistic bisimilarity and demonstrates that it ensures the continuity of the bisimilarity distance function. In Section~\ref{section:algorithm}, we present a polynomial-time algorithm for computing robust probabilistic bisimilarity. Section~\ref{section:experiments} reports experimental results on the algorithm's implementation. Finally, Section~\ref{section:conclusion} concludes the paper and discusses directions for future research.
\ifthenelse{\isundefined{\techReport}}{%
The full version of this paper, found in \cite{arxiv}, includes omitted proofs and further details.
}{%
Omitted proofs can be found in the appendix. This paper is an extended version of \cite{full}.
}
\section{Labelled Markov Chains and Probabilistic Bisimilarity}
\label{section:premliminaries}

In this section, we present some fundamental concepts that underpin this paper.

Let $X$ be a nonempty finite set.  A function $\mu : X \to [0, 1]$ is a \emph{subprobability distribution} on $X$ if $\sum_{x \in X} \mu(x) \leq 1$.  We denote the set of subprobability distributions on $X$ by $\Sreal(X)$.  For $\mu \in \Sreal(X)$ and $A \subseteq X$, we often write $\mu(A)$ instead of $\sum_{x \in A} \mu(x)$.  For a distribution $\mu \in \Sreal(X)$ we define the \emph{support} of $\mu$ by $\support(\mu) = \{\, x \in X \mid \mu(x) \gr 0 \,\}$.  A subprobability distribution $\mu$ on $X$ is a \emph{probability distribution} if $\mu(X) = 1$.  We denote the set of probability distributions on $X$ by $\Dreal(X)$.

A \emph{Markov chain} is a pair $<S, \tau>$ consisting of a finite set $S$ of states and a transition probability function $\tau : S \to \Dreal(S)$.  A \emph{labelled Markov chain} is a tuple $<S, L, \tau, \ell>$ where $<S, \tau>$ is a Markov chain, $L$ is a finite set of labels and $\ell: S \to L$ is a labelling function.  A \emph{path} in a Markov chain $<S, \tau>$ is a sequence of states $s_0$, $s_1$, $s_2 \ldots$ such that $s_i \in S$ and $\tau(s_{i})(s_{i+1}) \gr 0$ for all $i \geq 0$.

For the remainder, we fix a labelled Markov chain $<S, L, \tau, \ell>$, and we will study perturbations of the transition probability function $\tau$.

For all $\mu$, $\nu \in \Dreal(X)$, the set $\Creal(\mu, \nu)$ of \emph{couplings} of $\mu$ and $\nu$ is defined by
\[
\Creal(\mu, \nu) = \{\, \omega \in \Dreal(X \times X) \mid \forall x \in X : \omega(x, X) = \mu(x) \wedge \omega(X, x) = \nu(x) \,\}.
\]
We write $\omega(x, X)$ for $\sum_{y \in X} \omega(x, y)$.

\begin{definition}
\label{definition:probabilistic-bisimilary}
An equivalence relation $R \subseteq S \times S$ is a \emph{probabilistic bisimulation} (or just bisimulation) if for all $(s, t) \in R$, $\ell(s) = \ell(t)$ and there exists $\omega \in \Creal(\tau(s), \tau(t))$ such that $\support(\omega) \subseteq R$.  States $s$ and $t$ are \emph{bisimilar}, denoted $s \sim t$, if $(s, t) \in R$ for some bisimulation~$R$.
\end{definition}

If $|\ell(S)| = 1$ then $\mathord{\sim} = S \times S$.  In the remainder, we assume that the labelled Markov chain contains states with different labels, that is, $|\ell(S)| \geq 2$.  Hence, we also have that $|S| \geq 2$.

Definition~\ref{definition:probabilistic-bisimilary} \cite[Definition~4.3]{JL91} differs from the standard definition \cite[Definition~6.3]{LS89} which defines a bisimulation as an equivalence relation $R \subseteq S \times S$ such that for all $(s, t) \in R$, $\ell(s) = \ell(t)$ and for all $R$-equivalence classes $C$, $\tau(s)(C) = \tau(t)(C)$, where $\tau(s)(C) = \sum_{t \in C} \tau(s)(t)$.  Nevertheless, an equivalence relation $R$ is a bisimulation by Definition~\ref{definition:probabilistic-bisimilary} if and only if it is a bisimulation as per the standard definition (see \cite[Theorem~4.6]{JL91}).

\section{Probabilistic Bisimilarity Distances}
\label{section:distances}

\begin{definition}
\label{definition:delta}
The \emph{probabilistic bisimilarity distance} (or just bisimilarity distance), $\delta_\tau : S \times S \to [0, 1]$, is the least fixed point of the function $\Delta_\tau : (S \times S \to [0, 1]) \to (S \times S \to [0, 1])$ defined by
\[
\Delta_\tau(d)(s, t) = \left \{
\begin{array}{ll}
1 & \hspace{1cm} \mbox{if $\ell(s) \neq \ell(t)$}\\
\displaystyle \inf_{\omega \in \Creal(\tau(s), \tau(t))} \sum_{u, v \in S} \omega(u, v) \; d(u, v) & \hspace{1cm} \mbox{otherwise.}
\end{array}
\right .
\]
\end{definition}

Intuitively, the smaller the distance between two states, the more similar they behave.

\begin{theorem}[{\cite[Theorem~4.10]{DGJP04}}]
\label{theorem:distance-zero}
For all $s$, $t \in S$, $s \sim t$ if and only if $\delta_\tau(s, t) = 0$.
\end{theorem}

Quantitative $\mu$-calculus \cite{K83,AM01,MM05} is an expressive modal logic that uses fixed point operators to define properties of transition systems.  It supports the concise representation of a wide range of properties, including reachability, safety, and the probability of satisfying a general $\omega$-regular specification.
We use the syntax described in \cite{CAMR10}, except that we use the operator $\qnext$ instead of $\textrm{pre}_1$ and $\textrm{pre}_2$.  Let $q\mu$ denote the set of quantitative $\mu$-calculus formulae, then a formula $\varphi \in q\mu$ maps states to a numerical value within $[0,1]$, that is, $\sem{\varphi} : S \to [0, 1]$.
The bisimilarity distances can be characterized in terms of the quantitative $\mu$-calculus \cite[Equation~2.3]{CAMR10} as $\delta_\tau(s, t) = \sup_{\varphi \in q\mu} |\sem{\varphi}(s) - \sem{\varphi}(t)|$.

\subsection{Examples}

\begin{figure}
\begin{center}
\begin{subfigure}[t]{\textwidth}
    \begin{tikzpicture}[scale=0.9, every node/.style={scale=0.9}]
    \node[state, fill=cBlue!40] (1) at (0,2.7) {$h_0$};
    \node[state, fill=cRed!40] (2) at (2.5,1.5) {$t$};
    \draw[->] (1) edge node[above] {$\frac{1}{2}$} (2);
    \draw[->, loop left] (1) edge node[left] {$\frac{1}{2}$} (1);
    \draw[->, loop right] (2) edge node[right] {$1$} (2);
    \node[state, fill=cBlue!40] (1) at (0,0.3) {$h_1$};
    \draw[->] (1) edge node[below, xshift=5] {$\frac{1}{2} + \varepsilon$} (2);
    \draw[->, loop left] (1) edge node[left] {$\frac{1}{2} - \varepsilon$} (1);

    \begin{axis}
    [xlabel={$\varepsilon$}, ylabel={$\delta_{\tau_\varepsilon} (h_0, h_1)$}, width=6cm, at={(0.5\linewidth,0)}]
    \addplot[domain=0:0.5,color=cPurple,line width=1.2pt] {x/(1/2 + x)};
    \end{axis}
    \end{tikzpicture}
    \caption{Repeated tosses of a fair coin (top) and a biased coin (bottom) until each lands on tails.}
    \label{fig:example-continuous}
\end{subfigure}
\bigskip

\begin{subfigure}[t]{\textwidth}
    \begin{tikzpicture}[scale=0.9, every node/.style={scale=0.9}]
    \node[state, fill=cBlue!40] (1) at (0,2.7) {$h_2$}; 
    \node[state, fill=cRed!40] (2) at (2.5,1.5) {$t_3$};
    \draw[->, loop left] (1) edge node[left] {$1$} (1);
    \draw[->, loop right] (2) edge node[right] {$1$} (2);
    \node[state, fill=cBlue!40] (1) at (0,0.3) {$h_3$}; 
    \draw[->] (1) edge node[above] {$\varepsilon$} (2);
    \draw[->, loop left] (1) edge node[left] {$1 - \varepsilon$} (1);

    \begin{axis}
    [xlabel={$\varepsilon$}, ylabel={$\delta_{\tau_\varepsilon} (h_2, h_3)$}, width=6cm, at={(0.5\linewidth,0)}]
    \addplot[color=cPurple, mark=*, mark options={scale=1}] coordinates {(0,0)};
    \addplot[color=cPurple, mark=o, mark options={scale=1.2}] coordinates {(0,1)};
    \addplot[domain=0.01:0.5,color=cPurple,line width=1.2pt] {1};
    \end{axis}
    \end{tikzpicture}
    \caption{Single toss of a rigged coin (top) and repeated tosses of an extremely biased coin until it lands on tails (bottom).}
    \label{fig:example2-discontinuous}
\end{subfigure}
\bigskip

\begin{subfigure}[t]{\textwidth}
    \begin{tikzpicture}[scale=0.9, every node/.style={scale=0.9}]
    \node[state, fill=cBlue!40] (1) at (0,2.7) {$h_4$}; 
    \node[state, fill=cRed!40] (2) at (2.5,2.7) {$t_4$};
    \draw[->] (1) edge[bend left=20] node[above] {$\frac{1}{2}$} (2);
    \draw[->] (2) edge[bend left=20] node[below] {$\frac{1}{2}$} (1);
    \draw[->, loop left] (1) edge node[left] {$\frac{1}{2}$} (1);
    \draw[->, loop right] (2) edge node[right] {$\frac{1}{2}$} (2);

    \node[state, fill=cBlue!40] (1) at (0,0.3) {$h_5$}; 
    \node[state, fill=cRed!40] (2) at (2.5,0.3) {$t_5$};
    \draw[->] (1) edge[bend left=20] node[above] {$\frac{1}{2} + \varepsilon$} (2);
    \draw[->] (2) edge[bend left=20] node[below] {$\frac{1}{2}  - \varepsilon$} (1);
    \draw[->, loop left] (1) edge node[left] {$\frac{1}{2} - \varepsilon$} (1);
    \draw[->, loop right] (2) edge node[right] {$\frac{1}{2} + \varepsilon$} (2);

    \begin{axis}
    [xlabel={$\varepsilon$}, ylabel={$\delta_{\tau_\varepsilon} (h_4, h_5)$}, width=6cm, at={(0.5\linewidth,0)}]
    \addplot[color=cPurple, mark=*, mark options={scale=1}] coordinates {(0,0)};
    \addplot[color=cPurple, mark=o, mark options={scale=1.2}] coordinates {(0,1)};
    \addplot[domain=0.01:0.5,color=cPurple,line width=1.2pt] {1};
    \end{axis}
    \end{tikzpicture}
    \caption{Repeated tosses of a fair coin (top) and a biased coin (bottom).}
    \label{fig:example-discontinuous}
\end{subfigure}
\caption{Various examples featuring fair and biased coins. States labeled with \emph{heads} are shown in blue, while states labeled with \emph{tails} are shown in red.}
\label{figure:examples}
\end{center}
\end{figure}
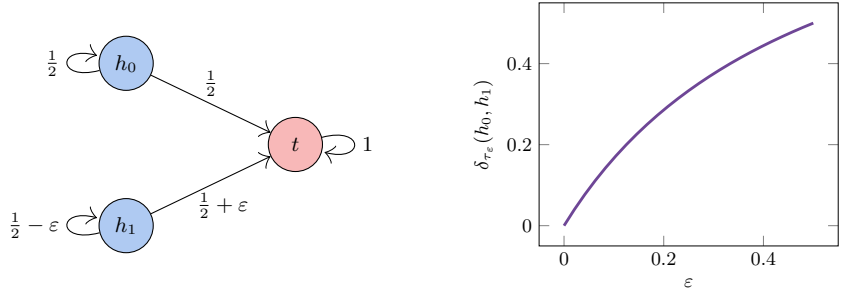
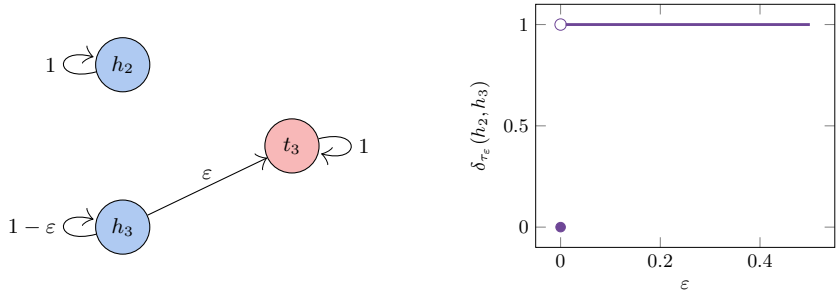
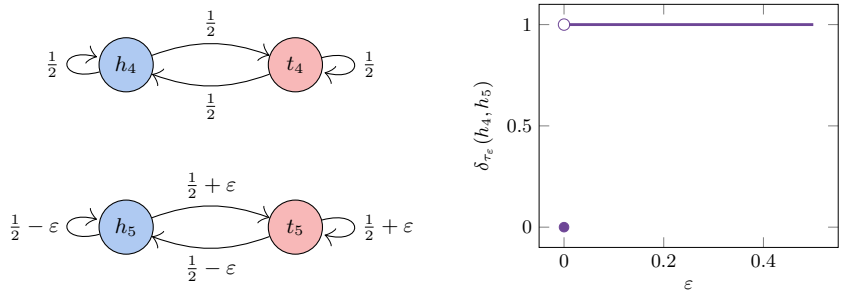

We now investigate how the bisimilarity distance changes when the transition function varies. In the following, let $\varepsilon \in [0, \frac{1}{2}]$.  Define $\tau_\varepsilon$ as shown in Figure~\ref{figure:examples}.  For example, $\tau_{1/6}(h_5)(t_5) = \frac{2}{3}$.  Then $\tau_\_ : [0, \frac{1}{2}] \to (S \to \Dreal(S))$ and $\delta_{\tau_\_} : [0, \frac{1}{2}] \to (S \times S \to [0,1])$.

\begin{example}
\label{example:continuous}
Consider Figure~\ref{fig:example-continuous}.  As $\varepsilon$ increases, $h_1$ becomes more biased and the distance between $h_0$ and $h_1$ increases proportionally.  One can show that $\delta_{\tau_\varepsilon}(h_0, h_1) = \frac{\varepsilon}{0.5 + \varepsilon} \leq 2\varepsilon$.  Note that if $\varepsilon$ is small then the distance is also small and $\lim_{\varepsilon \to 0} \delta_{\tau_\varepsilon}(h_0, h_1) = 0$.

The formula $\varphi = \mu\mvar.\, \qnext(\red \vee \mvar) \ominus 0.5$ distinguishes the states $h_0$ and $h_1$ the most, that is $\delta_{\tau_\varepsilon}(h_0, h_1) = \frac{\varepsilon}{0.5 + \varepsilon} = |\sem{\varphi}(h_0) - \sem{\varphi}(h_1)|$.  The quantifier~$\mu$ denotes the least fixed point of the recursive formula involving the variable $\mvar$.  Intuitively, a state satisfies $\mvar$ if the next state is $\red$ or satisfies $\mvar$ with probability greater than a half.  More precisely, considering state~$h_1$, $\sem{\varphi} (h_1)$ is the expected value of $\max(\sem{\varphi}(s), \sem{\red}(s)) - \frac{1}{2}$, where $s$ denotes the random successor state of $h_1$. Then, $\sem{\varphi} (h_1)$ evaluates to $\sum_{n=0}^\infty \varepsilon (\frac{1}{2} - \varepsilon)^n = \frac{\varepsilon}{0.5 + \varepsilon}$.  Each summand in the series represents the probability of reaching state $t$ in $n+1$ steps, starting from state $h_1$, with $0.5$ subtracted at each step.  On the other hand, $\sem{\varphi} (h_0) = 0$.
\end{example}

\begin{example}
\label{example:discontinuous}
In Figure~\ref{fig:example2-discontinuous}, when $\varepsilon = 0$, the states $h_2$ and $h_3$ are bisimilar with $\delta_{\tau_0}(h_2, h_3) = 0$.  However, if $\varepsilon \gr 0$, then $\delta_{\tau_\varepsilon}(h_2, h_3) = 1$.  This difference is evident when considering the probability of eventually reaching a state labelled with \emph{tails} when starting in $h_2$ compared to $h_3$.  In the first Markov chain, $\sem{\lozenge \textit{tails}} = 0$, while in the second Markov chain, $\sem{\lozenge \textit{tails}} = 1$.  This property can be expressed as the quantitative $\mu$-calculus formula $\mu\mvar.\, \qnext(\red \vee \mvar)$.  This example was also presented in \cite{JMLM14}.
\end{example}

\begin{example}
The first Markov chain in Figure~\ref{fig:example-discontinuous} represents fair coin flips, while the second Markov chain represents potentially biased coin flips.  When $\varepsilon = 0$, the states $h_4$ and $h_5$ are bisimilar with $\delta_{\tau_0}(h_4, h_5) = 0$.  However, if $\varepsilon \gr 0$, one can show that $\delta_{\tau_\varepsilon}(h_4, h_5) = 1$.  Intuitively, this is because small differences in probabilities can compound and lead to qualitative differences in the long-run behaviour.

Let us illustrate this.  Assume that a point is awarded each time the coin lands on \emph{tails} and a point is deducted each time it lands on \emph{heads}.  Let us examine the limit behaviour of the Markov chains.  Observe that the Markov chains behave like a random walk on the integer number line, $\mathbb{Z}$, starting at $0$.  At each step, the first Markov chain goes up by one with probability $\frac{1}{2}$ and down by one with probability $\frac{1}{2}$.  On the other hand, at each step, the second Markov chain goes up by one with probability $\frac{1}{2} + \varepsilon$ and down by one with probability $\frac{1}{2} - \varepsilon$.  Let $Y_1, Y_2, Y_3, \ldots$ be the sequence of independent random variables, where $Y_i$ denotes the $i^{\mathrm{th}}$ step taken by the random walk, with $Y_i = 1$ for a step up and $Y_i = -1$ for a step down. Define $S_n = \sum_{i = 1}^{n} Y_i$.  In the first Markov chain, $P(\liminf_{n \to \infty} S_n = -\infty) = 1$ and $P(\limsup_{n \to \infty} S_n = \infty) = 1$, by the Hewitt-Savage zero-one law \cite[Example~5.19]{E11}.  In contrast, in the second Markov chain, we have $P(\lim_{n \to \infty} S_n = \infty) = 1$, by the law of large numbers \cite{S64}.  Thus, in the first Markov chain, with equal chances of gaining or losing points at each step, the random walk almost surely oscillates infinitely.  In contrast, in the second Markov chain, the upward bias introduced by $\varepsilon \gr 0$ guarantees that the total number of points will eventually diverge to $+\infty$.

We see that small changes in the transition probabilities can lead to significant changes in the behaviour and, thus, in the distances between states.  This example is similar to the one presented in \cite{JMLM14}.
\end{example}

In the remainder, we conservatively assume that the transition function can be varied arbitrarily, that is, changes to the transition function are not restricted to specific transitions with constrained variables as in the examples.  Also, different from \cite{JMLM14}, the changes might ``add transitions.''  Therefore, we are interested in the continuity of the function $\delta_{\_}(s, t) : (S \to \Dreal(S)) \to [0, 1]$.  See \ifthenelse{\isundefined{\techReport}}{\cite[Appendix~A]{arxiv}}{Appendix~\ref{appendix:metric-topology}} for the metric on distributions $S \to \Dreal(S)$ used here.

The bisimilarity distance function $\delta_{\_}(s, t)$ is \emph{lower semi-continuous} at $\tau$ if for any sequence $(\tau_n)_n$ converging to $\tau$ we have $\liminf_n \delta_{\tau_n}(s, t) \geq \delta_\tau(s, t)$ and \emph{upper semi-continuous} at $\tau$ if we have $\limsup_n \delta_{\tau_n}(s, t) \leq \delta_\tau(s, t)$.  Lastly, $\delta_{\_}(s, t)$ is \emph{continuous} at $\tau$ if it is both upper semi-continuous and lower semi-continuous at $\tau$.

The examples in Figure~\ref{figure:examples} suggest that the bisimilarity distance function $\delta_{\_}$ is lower semi-continuous at $\tau$.  Indeed the following proposition shows that this holds in general, even allowing for arbitrary modifications of $\tau$.

\begin{proposition}
\label{proposition:lower-semi-continuous}
For all $s$, $t \in S$, the function $\delta_{\_}(s, t) : (S \to \Dreal(S)) \to [0, 1]$ is lower semi-continuous at $\tau$, that is, if $(\tau_n)_n$ converges to $\tau$ then $\liminf_n \delta_{\tau_n}(s, t)$ $\geq \delta_\tau(s, t)$.
\end{proposition}

In Figure~\ref{fig:example-discontinuous}, the bisimilarity distance function is not upper semi-continuous.  Specifically, $\limsup_{\epsilon \to 0} \delta_{\tau_\epsilon} (h_4, h_5) = 1$, while $\delta_{\tau_0} (h_4, h_5) = 0$.  As a result, small perturbations of $\tau$ cause a jump in the distance from $0$ to $1$.  The main goal of this paper is to characterize and  identify the continuity of the bisimilarity distance function for bisimilar pairs of states.

The following subsets of $S \times S$ play a key role in the subsequent discussion.
\begin{definition}
\label{definition:sets}
The sets $S^2_\Delta$, $S^2_{0,\tau}$, $S^2_1$, $S^2_{?,\tau}$, and $S^2_{0?}$ are defined by
\begin{align*}
S^2_\Delta = & \{\, (s, s) \mid s \in S \,\}\\
S^2_{0,\tau} = & \{\, (s, t) \in S \times S \mid s \not= t \wedge s \sim t \,\}\\
S^2_1 = & \{\, (s, t) \in S \times S \mid \ell(s) \not= \ell(t) \,\}\\
S^2_{?,\tau} = & (S \times S) \setminus (S^2_\Delta \cup S^2_{0,\tau} \cup S^2_1)\\
S^2_{0?} = & S^2_{0,\tau} \cup S^2_{?,\tau}
\end{align*}
\end{definition}
The first four sets form a partition of $S \times S$.  Observe that the sets $S^2_{0,\tau}$ and $S^2_{?,\tau}$ depend on $\tau$ and may, therefore, change when we perturb $\tau$, whereas the sets $S^2_\Delta$ and $S^2_1$ stay the same.  Note that $S^2_{0?} = (S \times S) \setminus (S^2_\Delta \cup S^2_1)$.  Hence, this set also stays the same if we perturb $\tau$.  Furthermore, note that $\mathord{\sim} = S^2_\Delta \cup S^2_{0,\tau}$ and for all $(s, t) \in S^2_1$, we have $\delta_\tau(s,t) = 1$.

\begin{definition}
Let $\tau : S \to \Dreal(S)$.  The set $\mathcal{P}_\tau$ of \emph{policies} for $\tau$ is defined by
\[
\mathcal{P}_\tau
=
\left \{\, P : S \times S \to \Dreal(S \times S) \, \middle \vert 
\begin{array}{l}
\forall (s, t) \in S^2_\Delta \cup S^2_{0?} : P(s, t) \in \Creal(\tau(s), \tau(t))\\
\forall (s, t) \in S^2_1 : \support(P(s, t)) = \{ (s, t) \}
\end{array}
\right \}.
\]
\end{definition}
Note that a policy $P \in \mathcal{P}_\tau$ induces a Markov chain $<S \times S, P>$.  The subscript $\tau$ is omitted when clear from the context.  The following proposition characterizes $\delta_\tau$ in terms of policies.

\begin{proposition}
\label{proposition:reach-s1}
For all $s$, $t \in S$, $\displaystyle \delta_\tau(s, t) = \min_{P \in \mathcal{P}} \gamma_P$, where $\gamma_P$ is the probability with which $(s, t)$ reaches $S_1^2$ in $<S \times S, P>$.
\end{proposition}
\begin{proofsketch}
The proof follows from {\cite[Theorem~10.15]{BK08}} and {\cite[Theorem~8]{CBW12}}.
\qed \end{proofsketch}

\begin{example}
\label{example:distance}
Consider the labelled Markov chain in Figure~\ref{fig:example-continuous} when $\varepsilon = \frac{1}{8}$.  Then the probability with which $(h_0, h_1)$ reaches $S_1^2$ for any policy $P \in \mathcal{P}$ is $\geq \frac{1}{5}$.  Any policy $P$ such that $P(h_0, h_1) = \set{(h_0, h_1) \mapsto \frac{3}{8}, (h_0, t) \mapsto \frac{1}{8}, (t, t) \mapsto \frac{1}{2}}$ achieves the minimum probability of $\frac{1}{5}$. The Markov chain induced by such a policy $P$ is illustrated in Figure~\ref{figure:distance}.  Thus, $\delta_{\tau_\varepsilon}(h_0, h_1) = \frac{1}{5}$.
\end{example}

\begin{figure}
\begin{center}
\begin{tikzpicture}[scale=0.9]
    \node[bigstate, fill=cRed!40, fill fraction={cBlue!40}{0.5}] (1) at (1,3) {$h_1 ~~ t~$};
    \node[bigstate, fill=cRed!40, fill fraction={cBlue!40}{0.5}] (2) at (3,3) {$h_0 ~~ t~$};
    \node[bigstate, fill=cBlue!40, fill fraction={cRed!40}{0.5}] (3) at (5,3) {$~t ~~ h_0$};
    \node[bigstate, fill=cBlue!40, fill fraction={cRed!40}{0.5}] (4) at (7,3) {$~t ~~ h_1$};
    \node[bigstate, fill=cBlue!40] (5) at (3,1) {$h_0 ~ h_1$};
    \node[bigstate, fill=cBlue!40] (6) at (5,1) {$h_1 ~ h_0$};
    \node[bigstate, fill=cBlue!40] (7) at (2,-1) {$h_0 ~ h_0$};
    \node[bigstate, fill=cRed!40] (8) at (4,-1) {$t ~~ t$};
    \node[bigstate, fill=cBlue!40] (9) at (6,-1) {$h_1 ~ h_1$};
    \draw (1.south) -- (1.north) ;
    \draw (2.south) -- (2.north) ;
    \draw (3.south) -- (3.north) ;
    \draw (4.south) -- (4.north) ;
    \draw (5.south) -- (5.north) ;
    \draw (6.south) -- (6.north) ;
    \draw (7.south) -- (7.north) ;
    \draw (8.south) -- (8.north) ;
    \draw (9.south) -- (9.north) ;
    \draw[->, loop above] (1) edge node[above] {$1$} (1);
    \draw[->, loop above] (2) edge node[above] {$1$} (2);
    \draw[->, loop above] (3) edge node[above] {$1$} (3);
    \draw[->, loop above] (4) edge node[above] {$1$} (4);
    \draw[->, loop below] (7) edge node[below] {$\frac{1}{2}$} (7);
    \draw[->, loop below] (8) edge node[below] {$1$} (8);
    \draw[->, loop below] (9) edge node[below] {$\frac{3}{8}$} (9);
    \draw[->] (7) edge node[below] {$\frac{1}{2}$} (8);
    \draw[->] (9) edge node[below] {$\frac{5}{8}$} (8);
    \draw[->] (5) edge node[left] {$\frac{1}{8}$} (2);
    \draw[->] (5) edge node[left] {$\frac{1}{2}$} (8);
    \draw[->, loop left] (5) edge node[left] {$\frac{3}{8}$} (5);
    \draw[->] (6) edge node[right] {$\frac{1}{8}$} (3);
    \draw[->] (6) edge node[right] {$\frac{1}{2}$} (8);
    \draw[->, loop right] (6) edge node[right] {$\frac{3}{8}$} (6);
\end{tikzpicture}
\end{center}
\caption{The Markov chain $<S \times S, P>$ induced by the policy $P$ such that $(h_0, h_1)$ reaches $S_1^2$ with probability $\frac{1}{5}$.}
\label{figure:distance}
\end{figure}
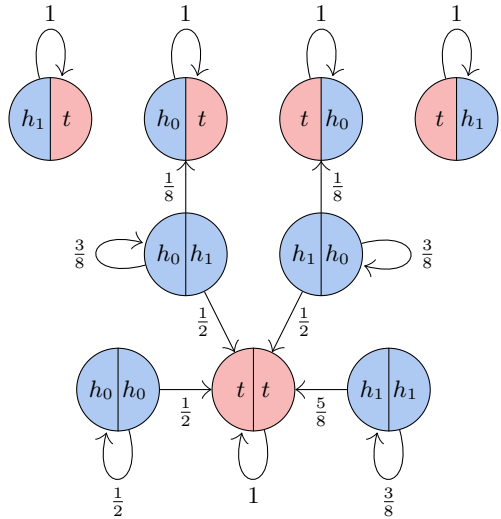


\section{Robust Probabilistic Bisimilarity}
\label{section:robust}

We aim to define a notion of robust bisimilarity which is a bisimulation that is robust against perturbations of the transition function $\tau$.  As we will see in Theorem~\ref{theorem:continuous} below, the following definition fulfills this requirement.

\begin{definition}
\label{definition:robust-probabilistic-bisimilarity}
\emph{Robust probabilistic bisimilarity} (or just robust bisimilarity), denoted $\robust$, is defined for $s$, $t \in S$ as $s \simeq t$ if there exists a policy $P \in \mathcal{P}$ such that $(s, t)$ reaches $S^2_\Delta$ with probability $1$ in $<S \times S, P>$.
\end{definition}

\begin{lemma}
\label{lemma:rpb-is-pb}
Robust bisimilarity, $\simeq$, is a bisimulation.
\end{lemma}
\begin{proofsketch}
Clearly, $\robust$ is reflexive and symmetric.  We prove in \ifthenelse{\isundefined{\techReport}}{\cite[Appendix]{arxiv}}{the Appendix} that $\robust$ is transitive as well and, therefore, an equivalence relation.

Let $s$, $t \in S$ such that $s \simeq t$.  Let $P \in \mathcal{P}$ be the policy such that $(s, t)$ reaches $S^2_\Delta$ with probability $1$ in $<S \times S, P>$.  Then, it follows from the definition of $\mathcal{P}$ that $(s, t) \not\in S_1^2$.  Thus, $\ell(s) = \ell(t)$.

Let $\omega = P(s, t)$, $u$, $v \in S$ and $(u, v) \in \support(\omega)$.  Hence, $\omega(u, v) \gr 0$ and $(u, v)$ is reachable from $(s, t)$.  Therefore, $(u, v)$ must reach $S^2_\Delta$ with probability $1$ in $<S \times S, P>$.  Consequently, $u \simeq v$. As a result, $\support(\omega) \subseteq \mathord{\simeq}$.
\qed \end{proofsketch}
Therefore, $\robust \subseteq \mathord{\sim}$ and, by Theorem~\ref{theorem:distance-zero}, for any $s$, $t \in S$ such that $s \simeq t$ we have $\delta_\tau(s, t) = 0$.

\begin{example}
\label{example:robust}
In Figure~\ref{fig:example-continuous}, when $\varepsilon = 0$, then $h_0 \simeq h_1$, since there exists a policy $P \in \mathcal{P}$ such that $(h_0, h_1)$ reaches $(t, t) \in S^2_\Delta$ with probability $1$ in $<S \times S, P>$.  Indeed, take $P(h_0, h_1) = \set{(h_0, h_1) \mapsto \frac{1}{2}, (t, t) \mapsto \frac{1}{2}}$ as shown in Figure~\ref{figure:robust}. Hence, $h_0 \simeq h_1$.  Note, however, that $h_2 \not\simeq h_3$ and $h_4 \not\simeq h_5$.
\end{example}

\begin{figure}[ht]
\begin{center}
\begin{tikzpicture}[scale=0.9]
    \node[bigstate, fill=cRed!40, fill fraction={cBlue!40}{0.5}] (1) at (1,3) {$h_1 ~~ t~$};
    \node[bigstate, fill=cRed!40, fill fraction={cBlue!40}{0.5}] (2) at (3,3) {$h_0 ~~ t~$};
    \node[bigstate, fill=cBlue!40, fill fraction={cRed!40}{0.5}] (3) at (5,3) {$~t ~~ h_0$};
    \node[bigstate, fill=cBlue!40, fill fraction={cRed!40}{0.5}] (4) at (7,3) {$~t ~~ h_1$};
    \node[bigstate, fill=cBlue!40] (5) at (3,1) {$h_0 ~ h_1$};
    \node[bigstate, fill=cBlue!40] (6) at (5,1) {$h_1 ~ h_0$};
    \node[bigstate, fill=cBlue!40] (7) at (2,-1) {$h_0 ~ h_0$};
    \node[bigstate, fill=cRed!40] (8) at (4,-1) {$t ~~ t$};
    \node[bigstate, fill=cBlue!40] (9) at (6,-1) {$h_1 ~ h_1$};
    \draw (1.south) -- (1.north) ;
    \draw (2.south) -- (2.north) ;
    \draw (3.south) -- (3.north) ;
    \draw (4.south) -- (4.north) ;
    \draw (5.south) -- (5.north) ;
    \draw (6.south) -- (6.north) ;
    \draw (7.south) -- (7.north) ;
    \draw (8.south) -- (8.north) ;
    \draw (9.south) -- (9.north) ;
    \draw[->, loop above] (1) edge node[above] {$1$} (1);
    \draw[->, loop above] (2) edge node[above] {$1$} (2);
    \draw[->, loop above] (3) edge node[above] {$1$} (3);
    \draw[->, loop above] (4) edge node[above] {$1$} (4);
    \draw[->, loop below] (7) edge node[below] {$\frac{1}{2}$} (7);
    \draw[->, loop below] (8) edge node[below] {$1$} (8);
    \draw[->, loop below] (9) edge node[below] {$\frac{3}{8}$} (9);
    \draw[->] (7) edge node[below] {$\frac{1}{2}$} (8);
    \draw[->] (9) edge node[below] {$\frac{5}{8}$} (8);
    \draw[->] (5) edge node[left] {$\frac{1}{2}$} (8);
    \draw[->, loop left] (5) edge node[left] {$\frac{1}{2}$} (5);
    \draw[->] (6) edge node[right] {$\frac{1}{2}$} (8);
    \draw[->, loop right] (6) edge node[right] {$\frac{1}{2}$} (6);
\end{tikzpicture}
\end{center}
\caption{The Markov chain $<S \times S, P>$ induced by the policy $P$ such that $(h_0, h_1)$ reaches $S^2_\Delta$ with probability $1$.}
\label{figure:robust}
\end{figure}
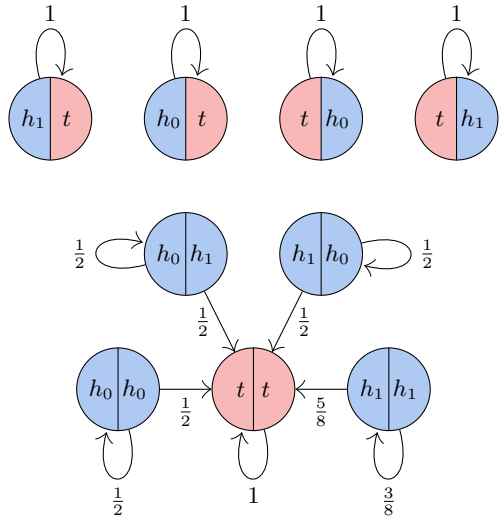

The following theorem provides the rationale behind the term robust bisimilarity. It establishes that for all robust bisimilar pairs of states, small perturbations of $\tau$ result in a correspondingly small change in the distance between them.

\begin{theorem}
\label{theorem:continuous}
For all $s$, $t \in S$, if $s \simeq t$ then the function $\delta_{\_}(s, t) : (S \to \Dreal(S)) \to [0, 1]$ is continuous at $\tau$, that is, for any sequence $(\tau_n)_n$ converging to $\tau$ we have $\lim_n \delta_{\tau_n}(s, t) = 0$.
\end{theorem}
\begin{proofsketch}
To build some intuition behind this theorem, we first outline the underlying idea.  Let $P \in \mathcal{P}$ be the policy such that $(s, t)$ reaches $S^2_\Delta$ with probability $1$ in $<S \times S, P>$.  Then, for some horizon $k$, the probability of $(s, t)$ reaching $S^2_\Delta$ within $k$ steps is almost one, say $1 - x$, where $x \gr 0$ is a small value.  When the transition function $\tau$ is perturbed by a small $\varepsilon$, there is a policy $P'$ such that the transitions in $<S \times S, P'>$ differ from those in $<S \times S, P>$ only by a small $\varepsilon' \gr 0$. Therefore, $(s, t)$ still reaches $S^2_\Delta$ with high probability in $<S \times S, P'>$.

To argue the last point in slightly more detail, observe that if $\varepsilon \gr 0$ is small enough then the probability, say~$p$, of any individual path of length at most~$k$ from $(s,t)$ to $S^2_\Delta$ remains at least $(1-x) \cdot p$ after the perturbation.
It follows that the probability of \emph{all} paths of length at most~$k$ from $(s,t)$ to $S^2_\Delta$ remains at least $(1-x) \cdot (1-x) \ge 1 - 2 x$ after the perturbation.

In \ifthenelse{\isundefined{\techReport}}{\cite[Appendix]{arxiv}}{the Appendix}, we provide a different, formal proof using matrix norms.  There we construct a graph consisting of the closed communication classes of $<S \times S, P>$ that are reachable from $(s, t)$.  Let $P_n \in \mathcal{P}_{\tau_n}$.  We then show that for all closed communication classes $C$ reachable from $(s, t)$ and for all pairs $(u, v) \in C$, it holds that $\lim_n \gamma_{P_n}(u, v) = \gamma_P(u, v) = 0$, by induction on the length of a longest path from $C$ in the above-mentioned graph.

By Proposition~\ref{proposition:reach-s1}, we have $\limsup_n \delta_{\tau_n}(s, t) \leq \limsup_n \gamma_{P_n}(s, t)$.  Using the above results, we conclude that $\limsup_n \delta_{\tau_n}(s, t) \leq \gamma_P (s, t) = 0$.
\qed \end{proofsketch}

Towards an algorithm for computing $\simeq$, let us develop another characterization of robust bisimilarity.
Given a policy $P \in \mathcal{P}$, we say that a set $R \subseteq S \times S$ \emph{supports} a path $(u_1, v_1) \hdots (u_n, v_n)$ in $<S \times S, P>$ if for all $1 \leq i \leq n$ we have $(u_i, v_i) \in R$ and $\support(P(u_i, v_i)) \subseteq R$.

\begin{definition}\
\label{definition:robust-probabilistic-bisimulation}
A \emph{robust bisimulation} is a bisimulation $R \subseteq S \times S$ such that for all $(s, t) \in R$, there exists a policy $P \in \mathcal{P}$ such that $R$ supports a path from $(s, t)$ to $S^2_\Delta$ in $<S \times S, P>$.
\end{definition}

\begin{proposition}
\label{proposition:robust-probabilistic-bisimulation}
Robust bisimilarity, $\mathord{\simeq}$ is a robust bisimulation.
\end{proposition}
\begin{proofsketch}
By Lemma~\ref{lemma:rpb-is-pb}, $\robust$ is a bisimulation.  Let $P \in \mathcal{P}$ be the policy such that $(s, t)$ reaches $S^2_\Delta$ with probability $1$ in $<S \times S, P>$.
Observe that for all $(u, v)$ reachable from $(s, t)$, $(u, v)$ must reach $S^2_\Delta$ with probability $1$ in $<S \times S, P>$.  Consequently, $\support(P(s, t)) \subseteq \mathord{\simeq}$.  In fact, $\mathord{\simeq}$ supports a path from $(s, t)$ to $S^2_\Delta$ in $<S \times S, P>$, and we can conclude that $\robust$ is a robust bisimulation.
\qed \end{proofsketch}

\begin{proposition}
\label{proposition:robust-bisimilarity}
For any robust bisimulation $R \subseteq S \times S$, we have $R \subseteq \mathord{\simeq}$.
\end{proposition}
\begin{proofsketch}
We construct a policy $P \in \mathcal{P}$ such that for every $(s, t) \in R$, $R$ supports a path from $(s, t)$ to $S^2_\Delta$ in $<S \times S, P>$ and for all $(s, t) \in S^2_\Delta$, $\support(P(s, t)) \subseteq S^2_\Delta$. Note that $P$ is designed to simultaneously ensure that all pairs in $R$ have an $R$-supported path to $S^2_\Delta$ in $<S \times S, P>$.  It follows from a standard result in Markov chain theory that all $(s, t) \in R$ reach $S^2_{\Delta}$ with probability $1$ in $<S \times S, P>$.
\qed \end{proofsketch}

It follows from Propositions~\ref{proposition:robust-probabilistic-bisimulation} and \ref{proposition:robust-bisimilarity} that $\mathord{\simeq}$, that is, robust bisimilarity, is the greatest robust bisimulation.  This is analogous to ordinary bisimulation, where bisimilarity is the greatest bisimulation.

\section{Algorithm}
\label{section:algorithm}

In this section, we present an efficient algorithm to compute robust bisimilarity; see Algorithm~\ref{algorithm:robust-bisimilarity}.  The algorithm relies on the following properties of any robust bisimulation $R$:
\begin{enumerate}
    \item for every $(s, t) \in R$ there exists a policy $P$ such that $R$ supports a path from $(s, t)$ to $S^2_\Delta$ in $<S \times S, P>$,
    \item $R$ is an equivalence relation, and 
    \item $R$ is a bisimulation.
\end{enumerate}
Robust bisimilarity is the greatest relation with these properties.  More formally, we define a function, $\mathrm{Refine}$, such that robust bisimilarity is the greatest fixed point of $\mathrm{Refine}$.

\begin{algorithm}[ht]
\caption{Computing robust bisimilarity for labelled Markov chains}
\label{algorithm:robust-bisimilarity}
\DontPrintSemicolon
\SetAlgoNoLine
\KwIn{A labelled Markov chain with a finite set $S$ of states and a transition probability function $\tau : S \to \Dreal(S)$, and the set of pairs of bisimilar states $\mathord{\sim} = S^2_{0,\tau} \cup S^2_{\Delta}$}
\KwOut{The set of pairs of robustly bisimilar states $R = \mathord{\simeq}$}
$R \gets \mathord{\sim}$\;
\Repeat{$R = R_{\mathrm{old}}$}{
  $R_{\mathrm{old}} \gets R$\;
  $R \gets \mathrm{Refine}(R)$ \Comment*{see Algorithm~\ref{algorithm:refine}}
}
\Return $R$
\end{algorithm}

For any $L, U$ with $L \subseteq U \subseteq S \times S$, write $[L, U] = \{\, R \subseteq S \times S \mid L \subseteq R \subseteq U \,\}$ and $[L, U]_\mathcal{B} = \{\, R \in [L, U] \mid R$ is a bisimulation $\}$.

\begin{itemize}
    \item The function $\mathrm{Filter} : [S^2_{\Delta}, \mathord{\sim}]_\mathcal{B} \to [S^2_{\Delta}, \mathord{\sim}]$ is defined as\\
    $\mathrm{Filter}(R) = \{\, (s, t) \in R \mid \exists P \in \mathcal{P} \mbox{ such that } R \mbox{ supports a path from } (s, t)$
    $\phantom{\mathrm{Filter}(R) = \{\, (s, t) \in R \mid\quad} \mbox{ to } S^2_{\Delta} \mbox{ in } <S \times S, P> \,\}.$
    \item The function $\mathrm{Prune} : [S^2_{\Delta}, \mathord{\sim}] \to [S^2_{\Delta}, \mathord{\sim}]$ is defined as\\
    $\begin{aligned}
    \mathrm{Prune}(R) = \{\, (s, t) \in R \mid \forall (t, u) \in R : (s, u) \in R \wedge \forall (u, s) \in R : (u, t) \in R \,\}.
    \end{aligned}$
    \item The function $\mathrm{Bisim} : [S^2_{\Delta}, \mathord{\sim}] \to [S^2_{\Delta}, \mathord{\sim}]_\mathcal{B}$ is defined as\\
    $\mathrm{Bisim}(R)$ is the largest bisimulation $R'$ with $R' \subseteq R$.\\
    Given an equivalence relation $R$, $\mathrm{Bisim}(R)$ can be computed in polynomial time (see \ifthenelse{\isundefined{\techReport}}{\cite[Proposition~24]{arxiv}}{Proposition~\ref{proposition:stabilise} in the Appendix}).
    \item Lastly, the function $\mathrm{Refine} : [S^2_{\Delta}, \mathord{\sim}]_\mathcal{B} \to [S^2_{\Delta}, \mathord{\sim}]_{\mathcal{B}}$ is defined as\\
    $\mathrm{Refine}(R) = \mathrm{Bisim}(\mathrm{Prune}(\mathrm{Filter}(R)))$.
\end{itemize}

\begin{proposition}
\label{proposition:monotonicity}
$\mathrm{Bisim}$ and $\mathrm{Filter}$ are monotone with respect to $\subseteq$. However, $\mathrm{Prune}$ is not.
\end{proposition}
\begin{proofsketch}
A counterexample for $\mathrm{Prune}$ is as follows.  Let $S = \set{s, t, u}$, $A = \{(s, s)$, $(t, t)$, $(u, u)$, $(s, t)$, $(t, s)\}$ and $B = \{(s, s)$, $(t, t)$, $(u, u)$, $(s, t)$, $(t, s)$, $(t, u)$, $(u, t)\}$.  $A$ and $B$ are symmetric and reflexive and, thus, can be visualized as an undirected graph as shown in Figure~\ref{figure:monotone}.  Observe that $A \subseteq B$, however, $\mathrm{Prune}(A) = A \not\subseteq \mathrm{Prune}(B) = \{(s, s), (t, t), (u, u)\}$.
\qed \end{proofsketch}

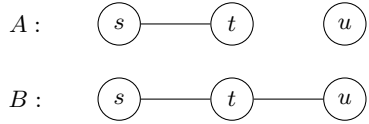
\begin{figure}[ht]
\begin{center}
\begin{tikzpicture}
    \node (0) at (0.25, 3) {$A:$};
    \node[smallstate] (1) at (1.5,3) {$s$};
    \node[smallstate] (2) at (3,3) {$t$};
    \node[smallstate] (3) at (4.5,3) {$u$};
    \draw[-] (1) to (2);

    \node (0) at (0.25, 2) {$B:$};
    \node[smallstate] (1) at (1.5,2) {$s$};
    \node[smallstate] (2) at (3,2) {$t$};
    \node[smallstate] (3) at (4.5,2) {$u$};
    \draw[-] (1) to (2);
    \draw[-] (2) to (3);
\end{tikzpicture}
\end{center}
\caption{Graph visualization of the relations $A$ and $B$ used in the proof of Proposition~\ref{proposition:monotonicity}.}
\label{figure:monotone}
\end{figure}

Note that Algorithm~\ref{algorithm:robust-bisimilarity} is not a typical fixed point iteration, since we do not know whether $\mathrm{Refine}$ is monotone.

\begin{algorithm}[ht]
\caption{Refine}
\label{algorithm:refine}
\DontPrintSemicolon
\SetAlgoNoLine
\KwIn{A set $R \in [S^2_{\Delta}, \mathord{\sim}]_\mathcal{B}$}
\KwOut{$\mathrm{Refine}(R)$}
$R \gets \mathrm{Filter}(R)$ \Comment*{see Algorithm~\ref{algorithm:filter}}
$R \gets \mathrm{Prune}(R)$ \Comment*{see Algorithm~\ref{algorithm:prune}}
$R \gets \mathrm{Bisim}(R)$\;
\Return $R$
\end{algorithm}

\begin{proposition}
\label{proposition:fp-refine-robust}
Any relation $R \subseteq S \times S$ is a robust bisimulation if and only if it is a fixed point of $\mathrm{Refine}$.
\end{proposition}
\begin{proofsketch}
Let $R \subseteq S \times S$.  Assume that $R$ is a robust bisimulation. By definition, $\mathrm{Refine}(R) \subseteq R$.  Since $R$ is a robust bisimulation, $R \subseteq \mathrm{Refine}(R)$.

Assume that $R$ is a fixed point of $\mathrm{Refine}$, then $R$ is a bisimulation and for every $(s, t) \in R$ there exists a policy $P$ such that $R$ supports a path from $(s, t)$ to $S^2_\Delta$ in $<S \times S, P>$.  Therefore, $R$ is a robust bisimulation.
\qed \end{proofsketch}

It follows from Propositions~\ref{proposition:robust-bisimilarity} and \ref{proposition:fp-refine-robust} that every fixed point of $\mathrm{Refine}$ is a subset of $\simeq$.  Furthermore, by Propositions~\ref{proposition:robust-probabilistic-bisimulation} and \ref{proposition:fp-refine-robust}, $\simeq$ is a fixed point of $\mathrm{Refine}$.  Therefore, $\simeq$ is the greatest fixed point of $\mathrm{Refine}$.

Let $Q \subseteq S \times S$ and $s$, $t$, $u$, $v \in S$.  We use the following notation below:  $\mathrm{Post}((s, t)) = \support(\tau(s)) \times \support(\tau(t))$ and $\mathrm{Pre}(Q) = \{\, (s, t) \in S \times S \mid \mathrm{Post}((s, t)) \cap Q \neq \varnothing \,\}$.

\begin{algorithm}[ht]
\caption{$\mathrm{Filter}$}
\label{algorithm:filter}
\DontPrintSemicolon
\SetAlgoNoLine
\KwIn{A set $R \in [S^2_{\Delta}, \mathord{\sim}]_\mathcal{B}$}
\KwOut{$\mathrm{Filter}(R)$}
$Q \gets S^2_{\Delta}$\;
\textcolor{gray}{$n \gets 0$\;}
\Repeat{$Q = Q_{\mathrm{old}}$}{
  $Q_{\mathrm{old}} \gets Q$\;
  \ForEach{$(s, t) \in (R \cap \mathrm{Pre}(Q_{\mathrm{old}})) \setminus Q_{\mathrm{old}}$}{
    $Q \gets Q \cup \{(s, t)\}$\;
  }
  \textcolor{gray}{$n \gets n + 1$\;}
}
\Return $Q$
\end{algorithm}

\begin{algorithm}[ht]
\caption{Prune}
\label{algorithm:prune}
\DontPrintSemicolon
\SetAlgoNoLine
\KwIn{A set $Q \in [S^2_{\Delta}, \mathord{\sim}]$}
\KwOut{$\mathrm{Prune}(Q)$}
$E \gets Q$\;
\ForEach{$(s, t) \in Q$}{
  \ForEach{$u \in S : (t, u) \in Q$}{
    \If{$(s, u) \not\in Q$}{
      $E \gets E \setminus \{(s, t), (t, u)\}$
    }
  }
}
\Return $E$
\end{algorithm}

\begin{proposition}
\label{proposition:transitive}
Given $R \in [\mathord{\simeq}, \mathord{\sim}]_\mathcal{B}$, for all $(s, t)$, $(t, u) \in \mathrm{Filter}(R)$, if $s \simeq t$ or $t \simeq u$ then $(s, u) \in \mathrm{Filter}(R)$.
\end{proposition}
\begin{proofsketch}
We show that if $t \simeq u$ then $(s, u) \in \mathrm{Filter}(R)$.  The case $s \simeq t$ is similar.  Write $s_1 = s$ and $t_1 = t$ and $u_1 = u$.

The idea behind the proof is that since $\mathrm{Filter}(R) \subseteq R$, we have $(s, t)$, $(t, u) \in R$.  Since $R$ is an equivalence relation, $(s, u) \in R$.  We define a policy $P \in \mathcal{P}$ such that for all $(s, t) \in R \cap S^2_{0?}$, $\support(P(s, t)) = \mathrm{Post}((s, t)) \cap R$.  We then show that since $(s, t) \in \mathrm{Filter}(R)$, there exists a path $(s_1, t_1), \ldots, (s_n, t_n)$ in $<S \times S, P>$, where $s_n = t_n$.

Assume that $(t, u) \in \mathord{\simeq}$.  Recall that $\simeq$ is a bisimulation.  Since $t_1, \ldots, t_n$ is a path in the original Markov chain $<S, \tau>$, there is also a path $u_1, \ldots, u_n$ in $<S, \tau>$ such that $(t_i, u_i) \in \mathord{\simeq}$ for all $1 \leq i \leq n$.  Since $\mathord{\simeq} \subseteq R$, there exists a path $(t_1, u_1), \ldots, (t_n, u_n)$ in $<S \times S, P>$.  Note that $(s_i, u_i) \in R$ for all $1 \leq i \leq n$.  Hence, there exists a path $(s_1,u_1), \ldots, (s_n,u_n) = (t_n,u_n)$ in $<S \times S, P>$.  See Figure~\ref{figure:transitive}.

Since $(t_n, u_n) \in \mathord{\simeq}$, we know that $(t_n, u_n)$ reaches $S^2_{\Delta}$ with probability $1$.  Therefore, there is a path $(t_n, u_n), \ldots, (t_m, u_m)$, with $t_m = u_m$ in $<S \times S, P>$ and $(t_i, u_i) \in \mathord{\simeq}$ for all $n \leq i \leq m$.  Thus, there exists paths $(s_1,u_1), \ldots, (s_n,u_n)$ and $(t_n,u_n), \ldots, (t_m,u_m)$ in $<S \times S, P>$, with $(s_n,u_n) = (t_n,u_n)$.  By the definition of $P$, $R$ supports the same path in $<S \times S, P>$. Hence, $(s,u) \in \mathrm{Filter}(R)$.
\qed \end{proofsketch}

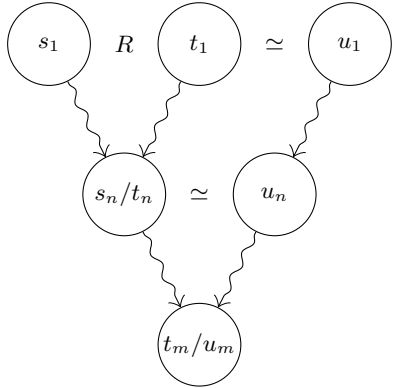
\begin{figure}[ht]
\begin{center}
\begin{tikzpicture}[decoration={snake, pre length=1pt, post length=2pt, amplitude=1.3pt}]
    \node (0) at (2, 3) {$R$};
    \node (0) at (4, 3) {$\simeq$};
    \node[bigstate] (1) at (1,3) {$s_1$};
    \node[bigstate] (2) at (3,3) {$t_1$};
    \node[bigstate] (3) at (5,3) {$u_1$};
    \node[bigstate] (4) at (2,1) {$s_n / t_n$};
    \node (0) at (3, 1) {$\simeq$};
    \node[bigstate] (5) at (4,1) {$u_n$};
    \node[bigstate] (6) at (3,-1) {$t_m / u_m$};

    \draw[->,decorate] (1) to (4);
    \draw[->,decorate] (2) to (4);
    \draw[->,decorate] (3) to (5);
    \draw[->,decorate] (4) to (6);
    \draw[->,decorate] (5) to (6);
\end{tikzpicture}
\end{center}
\caption{Illustration of the proof of Proposition~\ref{proposition:transitive}.}
\label{figure:transitive}
\end{figure}

Proposition~\ref{proposition:transitive} allows us to prove the following proposition.

\begin{proposition}
\label{proposition:robust-subset-R}
$R \in [\mathord{\simeq}, \mathord{\sim}]_\mathcal{B}$ is a loop invariant of Algorithm~\ref{algorithm:robust-bisimilarity}.
\end{proposition}
\begin{proofsketch}
$R$ is initialized to $\mathord{\sim}$, so the loop invariant holds before the loop.

Assume that the loop invariant holds before an iteration of the loop.  Since $\mathord{\simeq} \subseteq R$, $\mathrm{Filter}$ is monotone and $\mathord{\simeq}$ is a fixed point of $\mathrm{Refine}$, we have that $\mathord{\simeq}$ is a subset of $\mathrm{Filter}(R)$.

If $s \simeq t$, then $(s, t) \in \mathrm{Filter}(R)$.  Then, by Proposition~\ref{proposition:transitive}, for all $(t, u) \in \mathrm{Filter}(R)$ we have $(s, u) \in \mathrm{Filter}(R)$ and for all $(u, s) \in \mathrm{Filter}(R)$ we have $(u, t) \in \mathrm{Filter}(R)$.  Hence, $(s, t) \in \mathrm{Prune}(\mathrm{Filter}(R))$, and we have that $\mathord{\simeq}$ is a subset of $\mathrm{Prune}(\mathrm{Filter}(R))$.

$\mathrm{Bisim}$ is monotone, therefore, $\mathord{\simeq}$ is a subset of $\mathrm{Bisim}(\mathrm{Prune}(\mathrm{Filter}(R)))$ and $\mathrm{Refine}(R)$.  By the definition of $\mathrm{Bisim}$, $\mathrm{Refine}(R) \in [\mathord{\simeq}, \mathord{\sim}]_\mathcal{B}$.  Thus, the loop invariant is maintained in each iteration of the loop.
\qed \end{proofsketch}

Using the loop invariant established in Proposition~\ref{proposition:robust-subset-R}, we can now prove the correctness of Algorithm~\ref{algorithm:robust-bisimilarity}.

\begin{theorem}
\label{theorem:refine}
Algorithm~\ref{algorithm:robust-bisimilarity} computes the set $\simeq$.
\end{theorem}
\begin{proofsketch}
It is immediate from the definitions of $\mathrm{Bisim}$, $\mathrm{Filter}$ and $\mathrm{Prune}$ that $\mathrm{Refine}(R) \subseteq R$ holds for all $R \subseteq S \times S$.  By Proposition~\ref{proposition:robust-subset-R}, $\mathord{\simeq} \subseteq R$, thus, it computes a fixed point of $\mathrm{Refine}$ greater than or equal to $\mathord{\simeq}$.  Since $\mathord{\simeq}$ is the greatest fixed point of $\mathrm{Refine}$, we can conclude that Algorithm~\ref{algorithm:robust-bisimilarity} computes $\mathord{\simeq}$.
\qed \end{proofsketch}

In \ifthenelse{\isundefined{\techReport}}{\cite[Proposition~27]{arxiv}}{Proposition~\ref{proposition-runtime}}, we show that Algorithm~\ref{algorithm:robust-bisimilarity} has a time complexity of $\mathcal{O}(n^6)$, where $n = |S|$.  The computational bottleneck is the function $\mathrm{Filter}$.
\section{Experiments}
\label{section:experiments}

To evaluate the efficiency and usefulness of our robust bisimilarity algorithm,
we implemented it in the widely used probabilistic model checker PRISM~\cite{KNP11},
an open-source tool providing quantitative verification and analysis of several types of probabilistic models, including labelled Markov chains.

\subsection{Implementation}
PRISM's implementation of the traditional (i.e., non-robust) bisimilarity algorithm, $\mathrm{Bisim}$, is a standard partition-refinement approach
which uses the signature-based method of Derisavi \cite{D07}.
The initial partition is based on the labelling of the states.  Let $\Pi$ be the current partition and $E_\Pi$ be the set of equivalence classes in $\Pi$.  Then the new partition is computed as $\{\, (s, t) \in \Pi \mid \forall B \in E_\Pi : \tau(s)(B) = \tau(t)(B) \,\}$.

We implemented Algorithm~\ref{algorithm:robust-bisimilarity} in Java as part of PRISM's explicit-state model checking engine.  Each state and equivalence class (referred to as a block) is represented by an integer ID.  The current partition of the state space is tracked by an array that is indexed by state IDs and contains the corresponding block IDs.  To store the list of successors for each state, we use a map.  $\mathrm{Bisim}$ is run on the input Markov chain to obtain the set of bisimilar states.

The function $\mathrm{Filter}$ first constructs $R$ from the current partition and initializes $Q$ to $S^2_{\Delta}$.  In our approach, $R$ is implemented as an array indexed by block IDs, with each block containing a list of states.  Conversely, $Q$ is implemented as an array indexed by state IDs, with each state storing the set of states related to it.  Predecessors of $Q$ in $R$ are added to $Q$ until a fixed point is reached.  A pair of states $(s, t) \in R \setminus Q$ is a predecessor of $Q$ if they have some successors that are related in $Q$.  Specifically, there must exist successors $s'$ and $t'$ of $s$ and $t$, respectively, such that $t' \in Q[s']$ and vice versa.

$\mathrm{Prune}$ constructs a new partition of the state space by grouping states in the same block if they have the same neighbourhood in $Q$, that is, they are related to the same states.  In other words, $s$ and $t$ are placed in the same block if $Q[s] = Q[t]$ holds.  $\mathrm{Bisim}$ is then called with the current partition passed as the initial partition.  This process continues until no further refinement is possible, resulting in the set of robustly bisimilar states.  Finally, the minimized Markov chain is constructed.

\subsection{Experimental Setup}
We evaluated our algorithm by applying it to all (discrete-time) labelled Markov chains from the
%
%
Quantitative Verification Benchmark Set (QVBS) \cite{HKPQR19}, a comprehensive collection of probabilistic models which is designed as a benchmark suite for quantitative verification and analysis tools and is the foundation of the Quantitative Verification Competition (QComp), which compares the performance, versatility, and usability of such tools.

For an additional source of models, we also use jpf-probabilistic~\cite{FCDWTB21}.
Java PathFinder (JPF) \cite{VHBPL03} is the most popular model checker for Java code, and the JPF extension jpf-probabilistic provides Java implementations of sixty randomized algorithms \cite{FCDWTB21}. As shown in \cite{FCDWTB21}, JPF, extended by jp-probabilistic and jpf-label, can be used in tandem with PRISM to check properties of these algorithms and supplement JPF’s qualitative results with quantitative information.  A description of the subset of these algorithms utilized in our study is provided in \ifthenelse{\isundefined{\techReport}}{\cite[Appendix~J]{arxiv}}{Appendix~\ref{appendix:experiments}}.

In order to explore both the benefits and the efficiency of our algorithm, we run both the robust and traditional bisimilarity algorithms on all models.
For the latter, we use PRISM's existing implementation, in order to provide a comparable implementation.
Our experiments were run on a MacBook with an M1 chip and 16GB memory,
and with the Java virtual machine limited to 8GB.


\begin{table}[ht!]
\caption{The results of the benchmarks for which the minimized models differ.}
\rowcolors{1}{gray!10}{white}
\centering
\begin{tabular}{ m{0.16\textwidth} m{0.08\textwidth} m{0.1\textwidth} R{0.09\textwidth} R{0.12\textwidth} R{0.1\textwidth} R{0.14\textwidth} R{0.12\textwidth} }
  \toprule
  \rowcolor{white}
  \multicolumn{4}{c}{Benchmark} & \multicolumn{2}{r}{Bisimilarity\;} & \multicolumn{2}{r}{Robust Bisimilarity} \\
  \midrule
  Name (prop.) & \multicolumn{2}{c}{Parameters} & States & Min & Time & Min & Time \\ 
  \midrule
  brp (p1) & N=32 & MAX=2 & 1349 & 646 & 0.036 & 901 & 0.054 \\
  & & MAX=3 & 1766 & 871 & 0.043 & 1127 & 0.062 \\
  & & MAX=4 & 2183 & 1096 & 0.051 & 1353 & 0.068 \\
  & & MAX=5 & 2600 & 1321 & 0.058 & 1579 & 0.075 \\
  & N=64 & MAX=2 & 2693 & 1286 & 0.080 & 1797 & 0.105 \\
  & & MAX=3 & 3526 & 1735 & 0.084 & 2247 & 0.119 \\
  & & MAX=4 & 4359 & 2184 & 0.103 & 2697 & 0.130 \\
  & & MAX=5 & 5192 & 2633 & 0.132 & 3147 & 0.167 \\
  brp (p4) & N=32 & MAX=2 & 1349 & 10 & 0.012 & 711 & 3.690 \\
  & & MAX=3 & 1766 & 12 & 0.013 & 937 & 6.291 \\
  & & MAX=4 & 2183 & 14 & 0.018 & 1163 & 9.331 \\
  & & MAX=5 & 2600 & 16 & 0.021 & 1389 & 13.952 \\
  & N=64 & MAX=2 & 2693 & 10 & 0.017 & 1415 & 27.299 \\
  & & MAX=3 & 3526 & 12 & 0.015 & 1865 & 45.031 \\
  & & MAX=4 & 4359 & 14 & 0.016 & 2315 & 69.949 \\
  & & MAX=5 & 5192 & 16 & 0.018 & 2765 & 102.941 \\
  \midrule
  crowds & CS=5 & TR=3 & 1198 & 41 & 0.018 & 505 & 0.231 \\
  (positive) & & TR=4 & 3515 & 61 & 0.021 & 1484 & 1.304 \\
  & & TR=5 & 8653 & 81 & 0.038 & 3659 & 7.575 \\
  & & TR=6 & 18817 & 101 & 0.071 & 7969 & 34.765 \\
  & CS=10 & TR=3 & 6563 & 41 & 0.024 & 2320 & 8.296 \\
  & & TR=4 & 30070 & 61 & 0.078 & 10524 & 196.233 \\
  & & TR=5 & 111294 & 81 & 0.190 & 38770 & 2946.840 \\
  \midrule
  oscillators & T=6 & N=3 & 57 & 28 & 0.007 & 38 & 0.009 \\
  (power) & T=8 & N=6 & 1717 & 1254 & 0.037 & 1255 & 0.037 \\
  & & N=8 & 6436 & 5148 & 0.100 & 5149 & 0.122 \\
  oscillators & T=6 & N=3 & 57 & 28 & 0.007 & 38 & 0.008 \\
  (time) & T=8 & N=6 & 1717 & 1254 & 0.032 & 1255 & 0.036 \\
  & & N=8 & 6436 & 5148 & 0.111 & 5149 & 0.115 \\
  \midrule
  set isolation & U=13 & ST=3 & 8196 & 19 & 0.029 & 27 & 21.885 \\
  \multicolumn{2}{l}{(good sample)} & ST=4 & 8196 & 20 & 0.029 & 26 & 24.325 \\
  & & ST=5 & 8196 & 21 & 0.032 & 25 & 24.330 \\
  & & ST=6 & 8196 & 22 & 0.031 & 24 & 25.162 \\
  \bottomrule
\end{tabular}
\label{table:results}
\end{table}

\begin{table}[htp]
\caption{Models with the maximum state space per benchmark.}
\rowcolors{1}{gray!10}{white}
\centering
\begin{tabular}{ m{0.23\textwidth} m{0.17\textwidth} R{0.13\textwidth} R{0.25\textwidth} R{0.17\textwidth} }
  \toprule
  \rowcolor{white}
  \multicolumn{3}{c}{Benchmark} & \multicolumn{2}{r}{Robust Bisimilarity \;\;} \\
  \midrule
  Name & Property & States & Min States & Time \\ 
  \midrule
  crowds & positive & 111294 & 38770 & 2946.84 \\
  egl & messages & 115710 & 131 & 153.01 \\
  herman & steps & 32768 & 612 & 25.29 \\
  oscillators & power & 24311 & 17877 & 0.42 \\
  \bottomrule
\end{tabular}
\label{table:largest}
\end{table}

\begin{table}[htp]
\caption{Summary of all benchmarks with the change due to robust bisimilarity.}
\rowcolors{1}{gray!10}{white}
\centering
\begin{tabular}{ m{0.27\textwidth} m{0.17\textwidth} R{0.11\textwidth} R{0.23\textwidth} R{0.17\textwidth} }
  \toprule
  \rowcolor{white}
  \multicolumn{3}{c}{Benchmark} & \multicolumn{2}{r}{Average \% Increase \;} \\
  \midrule
  Name & Property & Instances & States & Time \\ 
  \midrule
  brp & p1 & 12 & 27.93 & 28.95 \\
  & p2 & 12 & 7.76 & 80.54 \\
  & p4 & 12 & 9193.43 & 142193.57 \\
  crowds & positive & 7 & 12306.72 & 273258.24 \\
  egl & messages & 6 & - & 20693.83 \\
  erd{\"o}s-r\'enyi model & connected & 18 & - & 799.07\\
  fair biased coin & heads & 9 & - & 0.00 \\
  has majority element & incorrect & 24 & - & 16.08 \\
  herman & steps & 7 & - & 297.99 \\
  leader-sync & elected \& time & 18 & - & 518.98 \\
  haddad-monmege & target & 3 & - & 0.00 \\
  oscillators & power \& time & 14 & 5.12 & 11.73 \\
  pollards factorization & input & 8 & - & 0.00 \\
  queens & success & 6 & - & 1193.93 \\  
  set isolation & good sample & 4 & 25.06 & 79035.95 \\
  \midrule
  \rowcolor{white}
  \textbf{Total} & & 160 & 1231.68 & 25589.22 \\
  \bottomrule
\end{tabular}
\label{table:aggregate}
\end{table}

\begin{table}[htp]
\caption{Models for which robust bisimilarity results in an \texttt{OutOfMemoryError}.}
\rowcolors{1}{gray!10}{white}
\centering
\begin{tabular}{ m{0.13\textwidth} m{0.15\textwidth} m{0.1\textwidth} m{0.1\textwidth} R{0.11\textwidth} R{0.22\textwidth} R{0.12\textwidth} }
  \toprule
  \rowcolor{white}
  \multicolumn{5}{c}{Benchmark} & \multicolumn{2}{c}{\hspace{1.5cm} Bisimilarity} \\
  \midrule
  Name & Property & \multicolumn{2}{c}{Parameters} & States & Min States & Time \\ 
  \midrule
  crowds & positive & CS=10 & TR=6 & 352535 & 101 & 0.57 \\
  egl & messages & N=5 & L=8 & 156670 & 171 & 0.87 \\
  & unfair & N=5 & L=2 & 33790 & 229 & 0.10 \\
  & & & L=4 & 74750 & 469 & 0.26 \\
  & & & L=6 & 115710 & 709 & 0.40 \\
  & & & L=8 & 156670 & 949 & 0.70 \\
  nand & reliable & N=20 & K=1 & 78332 & 39982 & 0.81 \\
  & & & K=2 & 154942 & 102012 & 1.89 \\
  & & & K=3 & 231552 & 164042 & 3.67 \\
  queens & success & & N=10 & 23492 & 527 & 0.08 \\
  \bottomrule
\end{tabular}
\label{table:too-large}
\end{table}

\subsection{Results}
Table~\ref{table:results} shows results for all benchmarks where the minimized models obtained by traditional bisimilarity and robust bisimilarity differ.
These are of particular interest because they are instances where our algorithm identifies that a model minimized in traditional fashion may not be robust.
In fact, in all benchmarks we have checked, we have observed that the distance between pairs of states that are not robustly bisimilar is discontinuous.  This leads us to the conjecture that for bisimilar states, robust bisimilarity is also a necessary condition for continuity.
The property used for each benchmark dictates the labelling used for the model.
In the table, \emph{Min} denotes the number of states in the minimized model and \emph{Time} denotes the amount of time taken (in seconds) to compute bisimilarity.

The results are promising, since robust bisimilarity, although (unsurprisingly) slower than traditional bisimilarity, 
remains practical across a wide range of standard benchmarks.
Table~\ref{table:largest} displays some of the largest models per benchmark along with the time required to compute robust bisimilarity.  The longest time recorded is about 50 minutes for the \emph{crowds} benchmark.
This may be due to the fact that the \emph{crowds} benchmark has many non-robustly bisimilar pairs of states, which we believe makes the benchmark harder. Other benchmarks, e.g. \emph{brp (p4)}, point in a similar direction.

The complete set of experiments includes 170 models, of which 160 are aggregated in Table~\ref{table:aggregate}. This table presents the average percentage increase in both the state space of the minimized model and the computation time for robust bisimilarity compared to traditional bisimilarity.  The \emph{crowds} benchmark exhibits the largest average percentage increase for both metrics.  The reported values may seem large, however, it is important to note that the traditional bisimilarity algorithm required a maximum of $2.14$ seconds per model in this table.

Furthermore, robust bisimilarity was successfully computed in less than a minute for 152 models (over 89\%).
Of the total set of models, the remaining 10 (approximately~6\%), listed in Table~\ref{table:too-large}, could only be minimized using traditional bisimilarity, as the robust bisimilarity computation ran out of available memory before completion.  This issue occurred with all instances of the \emph{nand} benchmark and half of the instances of the \emph{egl} benchmark.

Ultimately, robust bisimilarity proves feasible for large models, despite needing more resources than traditional bisimilarity. Furthermore, it offers a more reliable method of determining equivalence, which can be particularly beneficial in mission-critical applications, which require a higher level of precision.
\section{Conclusions and Future Work}
\label{section:conclusion}

To address the lack of robustness of probabilistic bisimilarity, we have introduced the concept of robust bisimilarity for labelled Markov chains.
Robust bisimilarity ensures that the distance function remains continuous even under perturbations of transition probabilities.
Additionally, we have presented a computationally efficient algorithm, with experimental results demonstrating that robust bisimilarity is plausible for large-scale verification tasks.

Our work opens new avenues for future exploration.
First, a logical characterization of robust bisimilarity could provide deeper insights.
Second, while we have established in Theorem~\ref{theorem:continuous} that robust bisimilarity is a sufficient condition for continuity, we conjecture that for bisimilar states, robust bisimilarity is in fact also a necessary condition for continuity.
We also aim to define continuity for non-bisimilar state pairs, to complete the theoretical characterization of robustness.
Thirdly, in \cite{JMLM14} it was shown that when the distances are discounted (i.e., differences that manifest themselves later count less), the distance function becomes continuous.  This raises the question: can we identify the properties for which the discontinuity is relevant? The examples suggest that these are long-term properties.
Finally, we plan to investigate specific types of perturbations of the transition probabilities, such as those that do not introduce new transitions, preserving the graph structure, as seen in Figures~\ref{fig:example-continuous} and \ref{fig:example-discontinuous}, unlike the perturbation shown in Figure~\ref{fig:example2-discontinuous} which adds a new transition.

\begin{credits}
\subsubsection{\ackname}
This research was supported in part by the Clarendon Fund and by the Natural Sciences and Engineering Research Council of Canada.

\subsubsection{\discintname}
The authors have no competing interests to declare that are relevant to the content of this article.
\end{credits}

\ifthenelse{\isundefined{\techReport}}{
}{
\clearpage
\appendix
\section*{\appendixname}

\section{Metric Topology}
\label{appendix:metric-topology}

\begin{definition}
The function $d_{E} : [0, 1] \times [0, 1] \to [0, 1]$ is defined by
\[
d_E(r, s) = | r - s |.
\]
\end{definition}

In the remainder, let $(Y, d)$ be a 1-bounded metric space.  According to, for example, \cite[Theorem~4.1.17]{E89}, a set $A \subseteq Y$ is \emph{compact} if and only if it is \emph{sequentially compact}, that is, each sequence in $A$ has a converging subsequence.  Given a sequence $(y_n)_n$, we denote a subsequence by $(y_{f(n)})_n$, where the function $f : \nat \to \nat$ is strictly increasing.

\begin{proposition}[{\cite[Example~3.1.25]{E89}}]
\label{proposition:unit-interval-compact}
$([0, 1], d_E)$ is compact.
\end{proposition}

\begin{definition}
Given a metric $d$ on $Y$, the function $d_F : (X \to Y) \times (X \to Y) \to [0, 1]$ is defined by
\[
d_F(f, g) = \max_{x \in X} d(f(x), g(x)).
\]
\end{definition}

\begin{proposition}[{\cite[Theorem~3.2.4]{E89}}]
\label{proposition:finite-product-compact}
If $(Y, d)$ is compact then $(X \to Y, d_F)$ is compact.
\end{proposition}

\begin{definition}
\label{definition:tv-distance}
The function $d_{TV} : \Dreal(X) \times \Dreal(X) \to [0, 1]$ is defined by
\[
d_{TV}(\mu, \nu) = \max_{x \in X} | \mu(x) - \nu(x) |.
\]
\end{definition}

\begin{proposition}
The function $f : Y \to [0, 1]$ is continuous at $y \in Y$ if and only if $f$ is lower semi-continuous at $y$ and upper semi-continuous at $y$.
\end{proposition}
\begin{proof}
Let $f : Y \to [0, 1]$ and $y \in Y$.  We prove two implications.  Assume that $f$ is lower and upper semi-continuous at $y$.  Let $(y_n)_n$ be a sequence in $Y$ converging to $y$.  Since
\begin{align*}
f(y) 
& \leq \lim \inf_n f(y_n)
\comment{$f$ is lower semi-continuous at $y$}\\
& \leq \lim \sup_n f(y_n)\\
& \leq f(y)
\comment{$f$ is upper semi-continuous at $y$}
\end{align*}
we can conclude that $\lim \inf_n f(y_n) = \lim \sup_n f(y_n) = f(y)$.  Hence, $\lim_n f(y_n)$ $= f(y)$.

Assume that $f$ is continuous at $y$.  Let $(y_n)_n$ be a sequence in $Y$ converging to $y$.  Since the sequence $(f(y_n))_n$ converges, we have that $\lim \inf_n f(y_n) = \lim \sup_n f(y_n) = \lim_n f(y_n) = f(y)$.  Hence, $f$ is lower and upper semi-continuous at $y$.
\qed \end{proof}

\section{Couplings}

\begin{lemma}
\label{lemma:couplings}
For all $\kappa$, $\lambda$, $\mu$, $\nu \in \Dreal(X)$ and $\omega \in \Creal(\lambda, \mu)$, there exist $\pi \in \Creal(\kappa, \mu)$ and $\rho \in \Creal(\lambda, \nu)$ such that $d_{TV}(\omega, \pi) \leq d_{TV}(\kappa, \lambda)$ and $d_{TV}(\omega, \rho) \leq d_{TV}(\mu, \nu)$.
\end{lemma}
\begin{proof}
We only prove the existence of coupling $\pi$, as $\rho$ can be dealt with similarly.	Let $\kappa$, $\lambda$, $\mu \in \Dreal(X)$ and $\omega \in \Creal(\lambda, \mu)$.
	
The proof is structured as follows.  We first construct a coupling $\sigma$ of $\kappa$ and $\lambda$.  Next, we compose this coupling $\sigma$ with the coupling $\omega$ of $\lambda$ and $\mu$, obtaining a coupling $\pi$ of $\kappa$ and $\mu$.  Finally, we show that this coupling $\pi$ satisfies the desired property.

\begin{center}
\begin{tikzpicture}[xscale=1.8]
\node (muu) at (0.5, 2) {$\kappa(x)$};
\node (muv) at (0.5, 0) {$\kappa(y)$};
\node at (0.5, 1) {$\vdots$};
\vertex (v0) at (1, 0) {$y$};
\vertex (u0) at (1, 2) {$x$};
\vertex (v1) at (3, 0) {$y$};
\vertex (u1) at (3, 2) {$x$};
\node (nuu) at (3.5, 2) {$\lambda(x)$};
\node (nuv) at (3.5, 0) {$\lambda(y)$};
\node at (3.5, 1) {$\vdots$};
\node at (1, 1) {$\vdots$};
\node at (3, 1) {$\vdots$};	
	
\vertex (v2) at (4, 0) {$y$};
\vertex (u2) at (4, 2) {$x$};
\vertex (v3) at (6, 0) {$y$};
\vertex (u3) at (6, 2) {$x$};
\node (lambdau) at (6.5, 2) {$\mu(x)$};
\node (lambdav) at (6.5, 0) {$\mu(y)$};
\node at (6.5, 1) {$\vdots$};
\node at (4, 1) {$\vdots$};
\node at (6, 1) {$\vdots$};
		
\draw[-] (muu) to (u0);
\draw[-] (muv) to (v0);
\draw[-] (nuu) to (u1);
\draw[-] (nuv) to (v1);
\draw[-] (nuu) to (u2);
\draw[-] (nuv) to (v2);
\draw[-] (lambdau) to (u3);
\draw[-] (lambdav) to (v3);
\draw[-] (u0) to node[midway,above] {$\scriptstyle \sigma(x, x)$} (u1);
\draw[-] (u0) to node[pos=0.3,above] {\, $\scriptstyle \sigma(x, y)$} (v1);
\draw[-] (v0) to node[pos=0.3,below] {\, $\scriptstyle \sigma(y, x)$} (u1);
\draw[-] (v0) to node[midway,below] {$\scriptstyle \sigma(y, y)$} (v1);

\draw[-] (u2) to node[midway,above] {$\scriptstyle \omega(x, x)$} (u3);
\draw[-] (u2) to node[pos=0.3,above] {\, $\scriptstyle \omega(x, y)$} (v3);
\draw[-] (v2) to node[pos=0.3,below] {\, $\scriptstyle \omega(y, x)$} (u3);
\draw[-] (v2) to node[midway,below] {$\scriptstyle \omega(y, y)$} (v3);
\end{tikzpicture}
\end{center}

We first construct $\sigma \in \Creal(\kappa, \lambda)$.  Set $\sigma(x, x) = \min \{ \kappa(x), \lambda(x) \}$ for all $x \in X$.  It remains to consider $\kappa'$, $\lambda'$ with
\[
\kappa'(x) = \left \{
\begin{array}{ll}
0 & \hspace{1cm} \mbox{if $\lambda(x) \geq \kappa(x)$}\\
\kappa(x) - \lambda(x) & \hspace{1cm} \mbox{otherwise}		
\end{array}
\right .
\]
and 
\[
\lambda'(x) = \left \{
\begin{array}{ll}
0 & \hspace{1cm} \mbox{if $\lambda(x) \leq \kappa(x)$}\\
\lambda(x) - \kappa(x) & \hspace{1cm} \mbox{otherwise.}	
\end{array}
\right .
\]
Note that $\lambda'(X) = \kappa'(X)$.  By means of the North-West corner method, we construct a $\sigma' \in \Creal(\kappa', \lambda')$.  Hence, for all $x \in X$, $\sigma'(x, X) = \kappa'(x)$ and $\sigma'(X, x) = \lambda'(x)$.  Since for all $x \in X$, $\kappa'(x) = 0$ or $\lambda'(x) = 0$, we can conclude that $\sigma'(x, x) = 0$.  We complete $\sigma$ by setting $\sigma(x, y) = \sigma'(x, y)$ for all $x$, $y \in X$ with $x \not= y$.  It remains to show that $\sigma \in \Creal(\kappa, \lambda)$.  For all $x \in X$,
\begin{align*}
\sigma(x, X) 
& = \sigma(x, x) + \sigma(x, X \setminus \{ x \})\\
& = \min \{  \kappa(x), \lambda(x) \} + \sigma'(x, X)\\
& = \min \{  \kappa(x), \lambda(x) \} + \kappa'(x)\\
& = \kappa(x)
\end{align*}
and
\begin{align*}
\sigma(X, x) 
& = \sigma(x, x) + \sigma(X \setminus \{ x \}, x)\\
& = \min \{  \kappa(x), \lambda(x) \} + \sigma'(X, x)\\
& = \min \{  \kappa(x), \lambda(x) \} + \lambda'(x)\\
& = \lambda(x).
\end{align*}

Let $\omega \in \Creal(\lambda, \mu)$.  We define
\[
\pi(x, y) = \sum_{z \in X \wedge \lambda(z) \not= 0} \frac{\sigma(x, z) \, \omega(z, y)}{\lambda(z)}.
\]
To conclude that $\pi \in \Creal(\kappa, \mu)$, we observe that for all $x \in X$,
\begin{align*}
&\pi(x, X)\\
& = \sum_{y \in X} \sum_{z \in X \wedge \lambda(z) \not= 0} \frac{\sigma(x, z) \, \omega(z, y)}{\lambda(z)}\\
& = \sum_{z \in X \wedge \lambda(z) \not= 0} \frac{\sigma(x, z)}{\lambda(z)} \sum_{y \in X} \omega(z, y)\\
& = \sum_{z \in X \wedge \lambda(z) \not= 0} \frac{\sigma(x, z)}{\lambda(z)} \lambda(z)
\comment{$\omega(z, X) = \lambda(z)$ since $\omega \in \Creal(\lambda, \mu)$}\\
& = \sum_{z \in X \wedge \lambda(z) \not= 0} \sigma(x, z)\\
& = \sigma(x, X)
\comment{if $\lambda(z) = 0$ then $\sigma(x, z) \leq \sigma(X, z) = \lambda(z) = 0$ since $\sigma \in \Creal(\kappa, \lambda)$}\\
& = \kappa(x)
\comment{$\sigma \in \Creal(\kappa, \lambda)$}
\end{align*}
and
\begin{align*}
& \pi(X, x)\\
& = \sum_{y \in X} \sum_{z \in X \wedge \lambda(z) \not= 0} \frac{\sigma(y, z) \, \omega(z, x)}{\lambda(z)}\\
& = \sum_{z \in X \wedge \lambda(z) \not= 0} \frac{\omega(z, x)}{\lambda(z)} \sum_{y \in X} \sigma(y, z)\\
& = \sum_{z \in X \wedge \lambda(z) \not= 0} \frac{\omega(z, x)}{\lambda(z)} \lambda(z)
\comment{$\sigma(X, z) = \lambda(z)$ since $\sigma \in \Creal(\kappa, \lambda)$}\\
& = \sum_{z \in X \wedge \lambda(z) \not= 0} \omega(z, x)\\
& = \omega(X, x)
\comment{if $\lambda(z) = 0$ then $\omega(z, x) \leq \omega(z, X) = \lambda(z) = 0$ as $\omega \in \Creal(\lambda, \mu)$}\\
& = \mu(x)
\comment{$\omega \in \Creal(\lambda, \mu)$}
\end{align*}

It remains to show that $d_{TV}(\omega, \pi) \leq d_{TV}(\kappa, \lambda)$.  Let $x$, $y \in X$.  It suffices to prove that $| \omega(x, y) - \pi(x, y) | \leq d_{TV}(\kappa, \lambda)$.  We distinguish the following cases.
\begin{itemize}
\item
Assume that $\lambda(x) = 0$.  Then $\omega(x, y) \leq \omega(x, X) = \lambda(x) = 0$ since $\omega \in \Creal(\lambda, \mu)$ and, hence, $\omega(x, y) = 0$.  Furthermore, $\pi(x, y) \leq \pi(x, X) = \kappa(x)$ since $\pi \in \Creal(\kappa, \mu)$.  Hence, $| \omega(x, y) - \pi(x, y) | \leq \kappa(x) = \kappa(x) - \lambda(x) \leq d_{TV}(\kappa, \lambda)$.
\item
Assume that $\kappa(x) = 0$.  Then $\pi(x, y) \leq \pi(x, X) = \kappa(x) = 0$ since $\pi \in \Creal(\kappa, \mu)$ and, hence, $\pi(x, y) = 0$.  Furthermore, $\omega(x, y) \leq \omega(x, X) = \lambda(x)$ since $\omega \in \Creal(\lambda, \mu)$.  Hence, $| \omega(x, y) - \pi(x, y) | \leq \lambda(x) = \lambda(x) - \kappa(x) \leq d_{TV}(\kappa, \lambda)$.
\item 
Assume that $0 \ls \lambda(x) \leq \kappa(x)$.  Then
\begin{align*}
\pi(x, y)
& = \sum_{z \in X \wedge \lambda(z) \not= 0} \frac{\sigma(x, z) \, \omega(z, y)}{\lambda(z)}\\
& = \omega(x, y) + \sum_{z \in X \wedge \lambda(z) \not= 0 \wedge z \not= x} \frac{\sigma(x, z) \, \omega(z, y)}{\lambda(z)}
\comment{$\sigma(x, x) = \lambda(x) \gr 0$}
\end{align*}
and 
\begin{align*}
& \sum_{z \in X \wedge \lambda(z) \not= 0 \wedge z \not= x} \frac{\sigma(x, z) \, \omega(z, y)}{\lambda(z)}\\
& \leq \sum_{z \in X \wedge \lambda(z) \not= 0 \wedge z \not= x} \sigma(x, z)
\comment{$\omega(z, y) \leq \omega(z, X) = \lambda(z)$ since $\omega \in \Creal(\lambda, \mu)$}\\
& = \sigma(x, X \setminus \{ x \})\\
& \quad \comment{if $\lambda(z) = 0$ then $\sigma(x, z) \leq \sigma(X, z) = \lambda(z) = 0$ since $\sigma \in \Creal(\kappa, \lambda)$}\\
& = \sigma'(x, X)\\
& = \kappa'(x)
\comment{$\sigma' \in \Creal(\kappa', \lambda')$}\\
& = \kappa(x) - \lambda(x)\\
& \leq d_{TV}(\kappa, \lambda).
\end{align*}
Hence, $| \omega(x, y) - \pi(x, y) | \leq d_{TV}(\kappa, \lambda)$.
\item 
Assume that $0 \ls \kappa(x) \leq \lambda(x)$.  Then
\begin{align*}
& \pi(x, y)\\
& = \sum_{z \in X \wedge \lambda(z) \not= 0} \frac{\sigma(x, z) \, \omega(z, y)}{\lambda(z)}\\
& = \frac{\sigma(x, x) \, \omega(x, y)}{\lambda(x)}
\comment{$\lambda(x) \gr 0$ and $\sigma(x, X \setminus \{ x \}) = 0$ since $\kappa(x) \leq \lambda(x)$}\\
& = \frac{\kappa(x) \, \omega(x, y)}{\lambda(x)}
\comment{$\sigma(x, x) = \kappa(x)$ since $\kappa(x) \leq \lambda(x)$}\\
& \leq \omega(x, y)
\comment{$\kappa(x) \leq \lambda(x)$}
\end{align*}	
Hence,
\begin{align*}
& \omega(x, y) - \pi(x, y)\\
& = \left (1 - \frac{\kappa(x)}{\lambda(x)} \right )\, \omega(x, y)\\
& = (\lambda(x) - \kappa(x))\, \frac{\omega(x, y)}{\lambda(x)}\\
& \leq \lambda(x) - \kappa(x)
\comment{$\omega(x, y) \leq \omega(x, X) = \lambda(x)$ since $\omega \in \Creal(\lambda, \mu)$}\\
& \leq d_{TV}(\kappa, \lambda).
\end{align*}
Because $\pi(x, y) \leq \omega(x, y)$ and $\omega(x, y) - \pi(x, y) \leq d_{TV}(\kappa, \lambda)$, we have that $| \omega(x, y) - \pi(x, y) | \leq d_{TV}(\kappa, \lambda)$.
\end{itemize}
\qed \end{proof}

\begin{corollary}
\label{corollary:coupling-set-nonexpansive}
For all $\kappa$, $\lambda$, $\mu$, $\nu \in \Dreal(X)$ and $\omega \in \Creal(\lambda, \mu)$, there exist $\pi \in \Creal(\kappa, \nu)$ such that $d_{TV}(\omega, \pi) \leq d_{TV}(\kappa, \lambda) + d_{TV}(\mu, \nu)$.
\end{corollary}
\begin{proof}
Let $\kappa$, $\lambda$, $\mu$, $\nu \in \Dreal(X)$ and $\omega \in \Creal(\lambda, \mu)$.  By Lemma~\ref{lemma:couplings}, there exists $\rho \in \Creal(\lambda, \nu)$ such that $d_{TV}(\omega, \rho) \leq d_{TV}(\mu, \nu)$ and, again using Lemma~\ref{lemma:couplings}, there exists $\pi \in \Creal(\kappa, \nu)$ such that $d_{TV}(\rho, \pi) \leq d_{TV}(\kappa, \lambda)$.  Therefore,
\begin{align*}
d_{TV}(\omega, \pi)
& \leq d_{TV}(\omega, \rho) + d_{TV}(\rho, \pi)
\comment{triangle inequality}\\
& \leq d_{TV}(\kappa, \lambda) + d_{TV}(\mu, \nu).
\end{align*}
\qed \end{proof}

For $\mu \in \Sreal(S)$, the \emph{$S^2_{\Delta}$-closed coupling} $\omega_\mu \in \Creal(\mu, \mu)$ of $\mu$ is defined as $\omega_\mu(s, s) = \mu(s)$ for all $s \in S$.  Note that $\support(\omega_\mu)$ $\subseteq S^2_{\Delta}$.

\begin{proposition}
For all $\mu$, $\nu \in \Dreal(X)$, we have $d_{TV}(\omega_\mu, \omega_\nu) \leq d_{TV}(\mu, \nu)$, where $\omega_\mu \in \Creal(\mu, \mu)$ and $\omega_\nu \in \Creal(\nu, \nu)$ are the $S^2_{\Delta}$-closed couplings of $\mu$ and $\nu$.
\label{proposition:closed-coupling-set-nonexpansive}
\end{proposition}
\begin{proof}
Let $\mu$, $\nu \in \Dreal(X)$.  Let $\omega_\mu \in \Creal(\mu, \mu)$ and $\omega_\nu \in \Creal(\nu, \nu)$ be the $S^2_{\Delta}$-closed couplings of $\mu$ and $\nu$, respectively. Let $x$, $y \in X$.  It suffices to show that $| \omega_\mu(x, y) - \omega_\nu(x, y) | \leq d_{TV}(\mu, \nu)$.  We distinguish two cases.
\begin{itemize}
\item
Assume that $x = y$.  Then $\omega_\mu(x, y) = \mu(x)$ and $\omega_\nu(x, y) = \nu(x)$.  Hence, $| \omega_\mu(x, y) - \omega_\nu(x, y) | = | \mu(x) - \nu(x) | \leq d_{TV}(\mu, \nu)$.
\item
Assume that $x \neq y$.  Since $(x, y) \not\in S^2_{\Delta}$, we have $\omega_\mu(x, y) = \omega_\nu(x, y) = 0$.  Thus, $| \omega_\mu(x, y) - \omega_\nu(x, y) | = 0 \leq d_{TV}(\mu, \nu)$.
\end{itemize}
\qed \end{proof}

\section{Policies}

We say that $P \in \mathcal{P}$ is an \emph{$S^2_\Delta$-closed policy} if $\forall s \in S : \support(P(s,s)) \subseteq S^2_\Delta$.

\begin{proposition}
\label{proposition:policy-Lipschitz}
For all $\sigma$, $\tau : S \to \Dreal(S)$, and $S^2_\Delta$-closed policies $P \in \mathcal{P}_\sigma$, there exists an $S^2_\Delta$-closed policy $Q \in \mathcal{P}_\tau$ such that $d_F(P, Q) \leq 2\, d_F(\sigma, \tau)$.
\end{proposition}
\begin{proof}
Let $\sigma$, $\tau : S \to \Dreal(S)$, and $P \in \mathcal{P}_\sigma$.  For each $(s, t) \in S^2_{0?}$, $P(s, t) \in  \Creal(\sigma(s), \sigma(t))$ and by Corollary~\ref{corollary:coupling-set-nonexpansive} there exists $\omega_{st} \in \Creal(\tau(s), \tau(t))$ such that 
\begin{equation}
\label{equation:matching-coupling}
d_{TV}(P(s, t), \omega_{st}) 
\leq d_{TV}(\sigma(s), \tau(s)) + d_{TV}(\sigma(t), \tau(t))
\leq 2\, d_F(\sigma, \tau).
\end{equation}
For each $(s, s) \in S^2_\Delta$, $P(s, s) = \omega_{\sigma(s)}$, where $\omega_{\sigma(s)} \in \Creal(\sigma(s), \sigma(s))$ is the $S^2_{\Delta}$-closed coupling of $\sigma(s)$.  Let $\omega_{\tau(s)} \in \Creal(\tau(s), \tau(s))$ be the $S^2_{\Delta}$-closed coupling of $\tau(s)$.  By Proposition~\ref{proposition:closed-coupling-set-nonexpansive}, we have 
\begin{equation}
\label{equation:matching-coupling-closed}
d_{TV}(P(s, s), \omega_{\tau(s)}) 
\leq d_{TV}(\sigma(s), \tau(s))
\leq 2\, d_F(\sigma, \tau).
\end{equation}
We define $Q$ by
\[
Q(s, t) = \left \{
\begin{array}{ll}
P(s, t)
& \hspace{1cm} \mbox{if $(s, t) \in S^2_1$}\\
\omega_{\tau(s)}
& \hspace{1cm} \mbox{if $(s, t) \in S^2_\Delta$}\\
\omega_{st}
& \hspace{1cm} \mbox{otherwise.}
\end{array}
\right .
\]
We leave it to the reader to verify that $Q \in \mathcal{P}_\tau$.  It suffices to show that for all $s$, $t \in S$, $d_{TV}(P(s, t), Q(s, t)) \leq 2\, d_F(\sigma, \tau)$.  If $(s, t) \in S^2_1$ then this is vacuously true.  Otherwise, it follows from (\ref{equation:matching-coupling}) and (\ref{equation:matching-coupling-closed}).
\qed \end{proof}

\section{Value Function}

\begin{definition}
The function $\Gamma : (S \times S \to \Dreal(S \times S)) \to (S \times S \to [0, 1]) \to (S \times S \to [0, 1])$ is defined by
\[
\Gamma_P(d)(s, t) = \left \{
\begin{array}{ll}
1
& \hspace{1cm} \mbox{if $(s, t) \in S^2_1$}\\
P(s, t) \cdot d
& \hspace{1cm} \mbox{otherwise,}	
\end{array}
\right .
\]
where $P(s, t) \cdot d = \sum_{u, v \in S} P(s, t)(u, v) \; d(u, v)$.
\end{definition}

For each $P : S \times S \to \Dreal(S \times S)$, $\Gamma_P$ is a monotone function from the complete lattice $S \times S \to [0, 1]$ to itself (see, for example, \cite[Proposition~6.1.3]{T18}).  According to the Knaster-Tarski fixed point theorem, $\Gamma_P$ has a least fixed point, which we denote by $\gamma_P$.  Note that $<S \times S, P>$ is a Markov chain.

Recall that $\mathcal{P}_\tau$ is the set of policies for $\tau$ and that the subscript $\tau$ is omitted when clear from the context.

\begin{theorem}[{\cite[Theorem~10.15]{BK08}}]
\label{theorem:reach-s1}
For all $P \in \mathcal{P}$ and $s$, $t \in S$, $\gamma_P(s, t)$ is the probability of reaching $S^2_1$ from $(s, t)$ in $<S \times S, P>$.
\end{theorem}

\begin{theorem}[{\cite[Theorem~8]{CBW12}}]
\label{proposition:minimal-coupling}
$\displaystyle
\delta_\tau = \min_{P \in \mathcal{P}} \gamma_P.
$
\end{theorem}

The above theorem is proved by showing that $\delta_\tau \sqsubseteq \gamma_P$ for all $P \in \mathcal{P}$ and that there exists $P \in \mathcal{P}$ such that $\delta_\tau = \gamma_P$.

\begin{proof}[of Proposition~\ref{proposition:reach-s1}]
Follows immediately from Theorems~\ref{theorem:reach-s1} and \ref{proposition:minimal-coupling}.
\qed \end{proof}

\begin{definition}
\label{definition:optimal-policy}
Let $s$, $t \in S$.  
\begin{itemize}
\item 
A policy $P \in \mathcal{P}$ is \emph{optimal for} $(s, t)$ if $\gamma_P(s, t) = \delta_\tau(s, t)$.
\item
A policy $P \in \mathcal{P}$ is \emph{optimal} if for all $s$, $t \in S$, $P$ is optimal for $(s, t)$.
\end{itemize}
\end{definition}

Note that from Theorem~\ref{proposition:minimal-coupling} we can conclude that optimal policies exist.

\begin{proposition}
\label{proposition:optimal-characterization}
For all $P \in \mathcal{P}$, the following are equivalent.
\begin{enumerate}
\item 
$P$ is optimal
\item 
$\Gamma_P(\delta_\tau) = \delta_\tau$
\item 
$\Gamma_P(\delta_\tau) \sqsubseteq \delta_\tau$
\end{enumerate}
\end{proposition}
\begin{proof}
Let $P \in \mathcal{P}$.  We prove three implications.
\begin{itemize}[leftmargin=15mm]
\item[1.\ $\Rightarrow$ 2.]
Assume that $P$ is optimal.  Then $\gamma_P = \delta_\tau$.  Therefore,
\[
\Gamma_P(\delta_\tau)
= \Gamma_P(\gamma_P)
= \gamma_P
= \delta_\tau.
\]
\item[2.\ $\Rightarrow$ 3.]
Immediate.
\item[3.\ $\Rightarrow$ 1.]
Assume that $\Gamma_P(\delta_\tau) \sqsubseteq \delta_\tau$, that is, $\delta_\tau$ is a pre-fixed point of $\Gamma_P$.  By the Knaster-Tarski fixed point theorem (see, for example, \cite[Theorem~2.35]{DP02}), $\gamma_P$ is the least pre-fixed point of $\Gamma_P$.  Hence,  $\gamma_P \sqsubseteq \delta_\tau$.  By Theorem~\ref{proposition:minimal-coupling}, $\delta_\tau \sqsubseteq \gamma_P$.  Therefore, $P$ is optimal.
\end{itemize}
\qed \end{proof}

Recall that states of a Markov chain \emph{communicate} with each other if both are reachable from one another by a (possibly empty) sequence of transitions that have positive probability. This is an equivalence relation which yields a set of \emph{communication classes}.  A communication class is \emph{closed} if the probability of leaving the class is zero.

\begin{proposition}
\label{proposition:closed-communication-class}
Let $P \in \mathcal{P}$ be an $S^2_\Delta$-closed policy.  If $C$ is a closed communication class of $<S \times S, P>$ then
\begin{enumerate}
\item 
$C = \{ (s, t) \}$ for some $(s, t) \in S^2_1$, or
\item 
$C \subseteq S^2_\Delta$, or
\item 
$C \subseteq S^2_{0,\tau}$.
\end{enumerate}
\end{proposition}
\begin{proof}
Let $P \in \mathcal{P}$ be an $S^2_\Delta$-closed policy and $C$ be a closed communication class of $<S \times S, P>$.  Let $s$, $t \in S$ and $(s, t) \in C$.  We distinguish the following cases.
\begin{itemize}
\item[a.] Suppose $(s, t) \in S^2_1$.  Then, it follows immediately from the definition of $\mathcal{P}$ that $C = \{ (s, t) \}$.
\item[b.] Suppose $(s, t) \in S^2_\Delta$.  Let $(u, v) \in C$.  Then $(u, v)$ is reachable from $(s, t)$.  It follows from the definition of an $S^2_\Delta$-closed policy that $(u, v) \in S^2_\Delta$.  Therefore, $C \subseteq S^2_\Delta$.
\item[c.] Suppose $(s, t) \in S^2_{0?}$.  Let $(u, v) \in C$.  By a.\ and b., $(u, v) \in S^2_{0?}$.  Hence $C \cap S^2_1 = \varnothing$ and, by Theorem~\ref{theorem:reach-s1}, we have $\delta_\tau(u, v) = 0$.  Therefore, $(u, v) \in S^2_{0,\tau}$.  Thus, $C \subseteq S^2_{0,\tau}$.
\end{itemize}
\qed \end{proof}

\begin{proposition}
\label{proposition:delta-closed-policy}
Let $s$, $t \in S$.  If there exists a policy $P \in \mathcal{P}$ such that $(s, t)$ reaches $S^2_\Delta$ with probability $p$ in $<S \times S, P>$, then there exists an $S^2_\Delta$-closed policy $Q \in \mathcal{P}$ such that $(s, t)$ reaches $S^2_\Delta$ with probability $p$ in $<S \times S, Q>$.
\end{proposition}
\begin{proof}
Let $s$, $t \in S$ and $P \in \mathcal{P}$.  Let $p$ be the probability with which $(s, t)$ reaches $S^2_\Delta$ in $<S \times S, P>$.
Observe that for all $s \in S$, there exists $\omega_{\tau(s)} \in \Omega(\tau(s), \tau(s))$ with $\support(\omega_{\tau(s)}) \subseteq S^2_\Delta$.
We define $Q$ by
\[
Q(s, t) = \left \{
\begin{array}{ll}
\omega_{\tau(s)}
& \hspace{1cm} \mbox{if $(s, t) \in S^2_\Delta$}\\
P(s, t)
& \hspace{1cm} \mbox{otherwise.}
\end{array}
\right .
\]
We leave it to the reader to verify that $(s, t)$ reaches $S^2_\Delta$ with probability $p$ in $<S \times S, Q>$.
\qed \end{proof}

\section{Linear Algebra}

We denote the infinity norm by $\| \cdot \|$.  Recall that for an $n$-vector $x$, we have that $\| x \| = \max_{0 \leq i \ls n} |x_i|$ and for an $m \times n$-matrix $A$, we have that $\| A \| = \max_{0 \leq i \ls m} \sum_{0 \leq j \ls n} | A_{ij} |$.  Given $n$-vectors $x$ and $y$, we write $x \lneqq y$ if $x_i \leq y_i$ for all $0 \leq i \ls n$ and $x_j  \ls y_j$ for some $0 \leq j \ls n$.  We denote constant vectors and matrices simply by their value.  A matrix $A$ is \emph{strictly substochastic} if $A 1 \lneqq 1$.  

The definition of an irreducible matrix from \cite{BP94} is the following, however, we will rely only on the characterisation of irreducibility in Theorem~\ref{theorem:irreducible}.
An $n \times n$ matrix $A$ is \emph{cogredient} to a matrix $E$ if for some permutation matrix $P$, $PAP^t = E$. $A$ is \emph{reducible} if it is cogredient to $E = \big[\begin{smallmatrix}
B & 0 \\
C & D
\end{smallmatrix}\big]$, where $B$ and $D$ are square matrices, or if $n = 1$ and $A = 0$. Otherwise, A is \emph{irreducible}.

\begin{theorem}[{\cite[Theorem~2.2.1]{BP94}}]
\label{theorem:irreducible}
A nonnegative $n \times n$-matrix $A$ is irreducible if and only if for every  $0 \leq i, j \ls n$ there exists $m \gr 0$ such that $A_{ij}^m \gr 0$.
\end{theorem}

\begin{proposition}
\label{proposition:irreducible-invertible}
Let $A$ be an irreducible and strictly substochastic $n \times n$-matrix.  Then $I - A$ is invertible.
\end{proposition}
\begin{proof}
Let $A$ be an irreducible and strictly substochastic $n \times n$-matrix.  Then $A$ is an irreducible nonnegative square matrix.  Since $A$ strictly substochastic, $A 1 \lneqq 1$.  \cite[Theorem~2.1.11]{BP94} states that for an irreducible nonnegative square matrix $A$, if $A x \lneqq x$ for some $x \gneqq 0$ then $\rho(A) \ls 1$, where $\rho(A)$ is the spectral radius of $A$.  Since $A 1 \lneqq 1$, we have thus that $\rho(A) \ls 1$.  Towards a contradiction, assume that $I - A$ is not invertible.  Then there exists $x \not= 0$ with $(I-A) x = 0$.  That is, $A x = x$. Thus, one is an eigenvalue of $A$, and so $\rho(A) \geq 1$.
\qed \end{proof}

\section{Probabilistic Bisimilarity Distances}

The variable $\red$ evaluates to one in the red states and zero in the blue states, that is, $\red = \set{h_2 \mapsto 0, h_3 \mapsto 0, t_4 \mapsto 1, t_5 \mapsto 1}$.  Let $\varphi_1 = \mu\mvar.\, \qnext(\red \vee \mvar) \ominus \frac{1}{2}$, then the computation of the quantitative $\mu$-calculus formula of Example~\ref{example:continuous} is as follows,
\begin{align*}
  \sem{\varphi_1}
  &= \textstyle \sem{\mu\mvar.\, \qnext(\red \vee \mvar) \ominus \frac{1}{2}} \\
  &= \textstyle {\inf}\set{\valu \in \valus \mid 
     \valu=\sem{\qnext(\red \vee \mvar) \ominus \frac{1}{2}}} \\
  &= \textstyle {\inf}\set{\valu \in \valus \mid 
     \valu=\sem{\qnext(\red \vee \mvar)} \ominus \frac{1}{2}} \\
  &= \textstyle {\inf}\set{\valu \in \valus \mid 
     \valu=\pnext(\sem{\red \vee \mvar}) \ominus \frac{1}{2}} \\
  &= \textstyle {\inf}\set{\valu \in \valus \mid 
     \valu=\pnext(\sem{\red} \imax \sem{\mvar}) \ominus \frac{1}{2}} \\
  &= \textstyle {\inf}\set{\valu \in \valus \mid 
     \valu=\pnext([\red] \imax \valu) \ominus \frac{1}{2}}
\end{align*}
\begin{align*}
  \sem{\varphi_1} (h_0)
  &= \textstyle \pnext([\red] \imax \sem{\varphi_1})(h_0) \ominus \frac{1}{2}\\
  &= \textstyle \frac{1}{2}([\red] \imax \sem{\varphi_1})(h_0) + \frac{1}{2}([\red] \imax \sem{\varphi_1})(t) \ominus \frac{1}{2}\\
  &= \textstyle \frac{1}{2}(\sem{\varphi_1}(h_0)) + \frac{1}{2}([\red] (t)) \ominus \frac{1}{2}\\
  &= \textstyle \frac{1}{2}(\sem{\varphi_1}(h_0))\\
  &= 0
\end{align*}
\begin{align*}
  \sem{\varphi_1} (h_1)
  &= \textstyle \pnext([\red] \imax \sem{\varphi_1})(h_1) \ominus \frac{1}{2}\\
  &= \textstyle (\frac{1}{2} - \varepsilon)([\red] \imax \sem{\varphi_1})(h_1) + (\frac{1}{2} + \varepsilon)([\red] \imax \sem{\varphi_1})(t) \ominus \frac{1}{2}\\
  &= \textstyle (\frac{1}{2} - \varepsilon)(\sem{\varphi_1}(h_1)) + (\frac{1}{2} + \varepsilon)([\red] (t)) \ominus \frac{1}{2}\\
  &= \textstyle (\frac{1}{2} - \varepsilon)(\sem{\varphi_1}(h_1)) + \varepsilon\\
  &= \textstyle \frac{\varepsilon}{0.5 + \varepsilon}
\end{align*}

Let $\varphi_2 = \mu\mvar.\, \qnext(\red \vee \mvar)$, then the computation of the quantitative $\mu$-calculus formula of Example~\ref{example:discontinuous} is as follows,
\begin{align*}
  \sem{\varphi_2} 
  &= \textstyle \sem{\mu\mvar.\, \qnext(\red \vee \mvar)} \\
  &= \textstyle {\inf}\set{\valu \in \valus \mid 
     \valu=\sem{\qnext(\red \vee \mvar)}} \\
  &= \textstyle {\inf}\set{\valu \in \valus \mid 
     \valu=\pnext(\sem{\red \vee \mvar})} \\
  &= \textstyle {\inf}\set{\valu \in \valus \mid 
     \valu=\pnext(\sem{\red} \imax \sem{\mvar})} \\
  &= \textstyle {\inf}\set{\valu \in \valus \mid 
     \valu=\pnext([\red] \imax \valu)}
\end{align*}
\begin{align*}
  \textstyle \sem{\varphi_2}(h_2)
  &= \textstyle \pnext([\red] \imax \sem{\varphi_2})(h_2)\\
  &= \textstyle 1([\red] \imax \sem{\varphi_2})(h_2)\\
  &= \textstyle \sem{\varphi_2}(h_2)\\
  &= 0
\end{align*}
\begin{align*}
  \textstyle \sem{\varphi_2}(h_3)
  &= \textstyle \pnext([\red] \imax \sem{\varphi_2})(h_3)\\
  &= \textstyle (1 - \varepsilon)([\red] \imax \sem{\varphi_2})(h_3) + \varepsilon([\red] \imax \sem{\varphi_2})(t_5)\\
  &= \textstyle (1 - \varepsilon)(\sem{\varphi_2}(h_3)) + \varepsilon([\red](t_5))\\
  &= \textstyle (1 - \varepsilon)(\sem{\varphi_2}(h_3)) + \varepsilon\\
  &= 1
\end{align*}

\section{Continuity}

\begin{proof}[of Proposition~\ref{proposition:lower-semi-continuous}]
Let $(\tau_n)_n$ be a sequence converging to $\tau$.  Below, we prove that $\lim \inf_n \delta_{\tau_n}$ is a pre-fixed point of $\Delta_\tau$.  Since $\delta_\tau$ is the least pre-fixed point of $\Delta_\tau$ by the Knaster-Tarski fixed point theorem, we can conclude that $\delta_\tau \sqsubseteq \lim \inf_n \delta_{\tau_{n}}$.

Since the set $S$ is finite, we can conclude from Propositions~\ref{proposition:unit-interval-compact} and \ref{proposition:finite-product-compact} that $S \times S \to [0, 1]$ is compact.
Hence, there exists a subsequence $(\tau_{f(n)})_n$ of $(\tau_n)_n$, which also converges to $\tau$, such that $\lim_n \delta_{\tau_{f(n)}} = \liminf_n \delta_{\tau_{n}}$.
By Theorem~\ref{proposition:minimal-coupling}, for each $\tau_{f(n)}$ there exists an optimal policy $P_{f(n)} \in \mathcal{P}_{\tau_{f(n)}}$.
Moreover, $\Dreal(S \times S)^{S \times S}$ is also compact.  Thus, the sequence $(P_{f(n)})_n$ has a converging subsequence $(P_{g(n)})_n$.  Define $P \in \mathcal{P}_\tau$ by $P = \lim_{n} P_{g(n)}$.

To conclude that $\Delta_\tau(\lim \inf_n \delta_{\tau_{n}}) \sqsubseteq \lim \inf_n \delta_{\tau_{n}}$, it suffices to show that for all $s$, $t \in S$, $\Delta_\tau(\lim_n \delta_{\tau_{g(n)}})(s, t) \leq \lim_n \delta_{\tau_{g(n)}}(s, t)$.
Let $s$, $t \in S$.  We distinguish the following cases.
\begin{itemize}
\item 
If $(s, t) \in S^2_1$ then
\[
\Delta_\tau(\lim_n \delta_{\tau_{g(n)}})(s, t)
= 1
= \lim_n \Delta_{\tau_{g(n)}}(\delta_{\tau_{g(n)}})(s, t)
= \lim_n \delta_{\tau_{g(n)}}(s, t).
\]
\item 
Assume that $(s, t) \in S^2_\Delta \cup S^2_{0?}$.  Then,
\begin{align*}
\Delta_\tau(\lim_n \delta_{\tau_{g(n)}})(s, t)
& = \inf_{\omega \in \Creal(\tau(s), \tau(t))} \omega \cdot \lim_n \delta_{\tau_{g(n)}}\\
& \leq (\lim_n P_{g(n)})(s, t) \cdot \lim_n \delta_{\tau_{g(n)}}\\
& \quad \comment{$(\lim_n P_{g(n)})(s, t) = P(s, t) \in \Creal(\tau(s), \tau(t))$}\\
& = \lim_n \left( P_{g(n)}(s, t) \cdot \delta_{\tau_{g(n)}} \right)\\ 
& = \lim_n \Gamma_{P_{g(n)}}(\delta_{\tau_{g(n)}})(s, t)\\
& = \lim_n \delta_{\tau_{g(n)}}(s, t)
\comment{$P_{g(n)}$ is optimal, Proposition~\ref{proposition:optimal-characterization}}.
\end{align*}
\end{itemize}
\qed \end{proof}

\begin{lemma}
\label{lemma:continuous}
Let $s$, $t \in S$.  If there exists a policy $P \in \mathcal{P}_\tau$ such that 
\begin{equation}
\label{equation:reach-ccc}
(s, t) \mbox{ reaches } S^2_\Delta \mbox{ with probability } 1 \mbox{ in } <S \times S, P>
\end{equation}
then the function $\delta_{\_}(s, t) : (S \to \Dreal(S)) \to [0, 1]$ is upper semi-continuous at $\tau$. 
\end{lemma}
\begin{proof}
Let $s$, $t \in S$.   Assume that $(\tau_n)_n$ is a sequence in $S \to \Dreal(S)$ that converges to $\tau$.  It suffices to show that $\lim \sup_n \delta_{\tau_n}(s, t) \leq \delta_{\tau}(s, t)$.

Assume that $Q \in \mathcal{P}_\tau$ is a policy such that $(s, t)$ reaches $S^2_\Delta$ with probability $1$ in $<S \times S, Q>$.  By Proposition~\ref{proposition:delta-closed-policy} there exists an $S^2_\Delta$-closed policy $P \in \mathcal{P}_\tau$ and (\ref{equation:reach-ccc}).  It follows from Theorem~\ref{theorem:reach-s1} that $\gamma_P(s, t) = 0$.  Thus, by Theorem~\ref{proposition:minimal-coupling}, $\delta_\tau(s, t) = 0$.  Hence, $P$ is optimal for $(s, t)$ due to Definition~\ref{definition:optimal-policy}.  According to Proposition~\ref{proposition:policy-Lipschitz}, for each $n \in \nat$, there exists an $S^2_\Delta$-closed policy $P_n \in \mathcal{P}_{\tau_n}$ such that $d_F(P_n, P) \leq 2\, d_F(\tau_n, \tau)$.  Hence, $(P_n)_n$ converges to $P$.

Consider the directed graph consisting of the communication classes of $<S \times S, P>$ reachable from $(s, t)$ as vertices.  There is an edge from communication class $C$ to communication class $D$ if there exist $(u, v) \in C$ and $(w, x) \in D$ such that $P(u, v)(w, x) \gr 0$.  This graph is acyclic.  We first prove that for all communication classes $C$ of $<S \times S, P>$ that are reachable from $(s, t)$ and for all $(u, v) \in C$,
\begin{equation}
\label{equation:lim-gamma-is-gamma}
\lim_n \gamma_{P_n}(u, v) = \gamma_P(u, v)
\end{equation} 
by induction on the length of a longest path from $C$ in the communication classes graph.

In the base case, we consider the \emph{closed} communication classes $C$, from which the length of a longest path in the communication classes graph is one.  By (\ref{equation:reach-ccc}) and Proposition~\ref{proposition:closed-communication-class}, we only need to consider closed communication classes that are subsets of $S^2_\Delta$.
Let $C \subseteq S^2_\Delta$ and $(u, v) \in C$.  According to Theorem~\ref{theorem:reach-s1}, for all $n \in \nat$, $\gamma_{P_n}(u, v) = 0$ and $\gamma_P(u, v) = 0$.  Therefore, (\ref{equation:lim-gamma-is-gamma}).

Next, we consider the inductive case.  Let $C$ be a communication class of $<S \times S, P>$ reachable from $(s, t)$.  Let $B$ be the set of state pairs of all communication classes that can be reached from $C$ in the communication classes graph via a path of length greater than $1$.  By induction, for all $(u, v) \in B$, (\ref{equation:lim-gamma-is-gamma}) holds.  Let $A$ be the set of all other state pairs, that is, $A = (S \times S) \setminus (B \cup C)$.
\begin{center}
\begin{tikzpicture}[xscale=2]
\node[state] (st) at (1, 3.5) {$s, t$};

\node[ellipse, draw, fit={(0.5, 1.5) (1, 2) (1, 1) (1.5,1.5)}, fill=cBlue!40] {};

\node[state] (uv) at (1, 1.8) {$u, v$};

\node at (2, 1.5) {$C$};

\draw[->,decorate,decoration={snake, pre length=1pt, post length=2pt, amplitude=1.3pt}] (st) to (uv);

\node[ellipse, draw, fit={(-0.5, -1) (0, -1.5) (0.5, -1) (0,-0.5)}, fill=cRed!40] {};
\node[ellipse, draw, fit={(1.5, -1) (2, -1.5) (2.5, -1) (2,-0.5)}, fill=cRed!40] {};
\node[ellipse, draw, fit={(0.5, -2.5) (1, -2) (1, -3) (1.5,-2.5)}, fill=cRed!40] {};

\node[draw, fit={(-1, -3.5) (-1, 0) (3, 0) (3,-3.5)}] {};

\node at (3.25, -1.75) {$B$};

\draw[->]  (0.5, 1.5) to (0, -0.5);
\draw[->]  (1.5, 1.5) to (2, -0.5);
\draw[->]  (0, -1.5) to (0.5, -2.5);
\draw[->]  (2, -1.5) to (1.5, -2.5);
\end{tikzpicture}
\end{center}

For $X \subseteq S \times S$ and $n \in \nat$, consider the vectors
\begin{align*}
\gamma_{P_n}^X & = (\gamma_{P_n}(u, v))_{(u, v) \in X}\\
\gamma_{P}^X & = (\gamma_{P}(u, v))_{(u, v) \in X}
\end{align*}
and the matrices
\begin{align*}
P_n^X & = (P_n(u, v)(w, x))_{(u, v) \in C, (w, x) \in X}\\
P^X & = (P(u, v)(w, x))_{(u, v) \in C, (w, x) \in X}
\end{align*}
For all $n \in \nat$, $\gamma_{P_n} = \Gamma_{P_n}(\gamma_{P_n})$ and, hence,
\[
\gamma_{P_n}^C = P_n^C \gamma_{P_n}^C + P_n^B \gamma_{P_n}^B + P_n^A \gamma_{P_n}^A.
\]
Since $\gamma_P = \Gamma_P(\gamma_P) $, we also have
\[
\gamma_P^C = P^C \gamma_P^C + P^B \gamma_P^B + P^A \gamma_P^A.
\]
From the communication classes graph we can infer that for all $(u, v) \in C$, $\support(P(u, v)) \subseteq B \cup C$.  Hence, $P^A = 0$.  Since $\lim_n P_n = P$, we have that
\begin{align}
\lim_n P_n^A & = P^A = 0 \nonumber\\
\lim_n P_n^B & = P^B
\label{equation:limit-matrix}\\
\lim_n P_n^C & = P^C \nonumber
\end{align}

Next, we prove that the inverse $(I - P^C)^{-1}$ exists.  We distinguish the following cases.
\begin{itemize}
\item 
If $P^C = 0$ then $I - P^C = I$, which has an inverse.
\item
Otherwise, $P^C \gneqq 0$.  Because $C$ is a communication class, for every $(u, v)$, $(w, x) \in C$ there exists $m$ such that $(P^C)^m_{(u, v), (w, x)} \gr 0$.  By Theorem~\ref{theorem:irreducible}, $P^C$ is irreducible.  Since the communication class $C$ is not closed, $P^C$ is strictly substochastic.  Hence, by Proposition~\ref{proposition:irreducible-invertible}, the inverse $(I - P^C)^{-1}$ exists.
\end{itemize}

Therefore,
\begin{align*}
\gamma_{P_n}^C - \gamma_P^C
& = (P_n^C \gamma_{P_n}^C + P_n^B \gamma_{P_n}^B + P_n^A \gamma_{P_n}^A) - (P^C \gamma_P^C + P^B \gamma_P^B + P^A \gamma_P^A)\\
& = (P_n^C \gamma_{P_n}^C + P_n^B \gamma_{P_n}^B + P_n^A \gamma_{P_n}^A) - (P^C \gamma_P^C + P^B \gamma_P^B)
\comment{$P^A = 0$}\\
& = P^C (\gamma_{P_n}^C - \gamma_P^C) + (P_n^C - P^C) \gamma_{P_n}^C + P^B (\gamma_{P_n}^B - \gamma_P^B)\\
&\qquad + (P_n^B - P^B) \gamma_{P_n}^B + P_n^A \gamma_{P_n}^A
\end{align*}
Hence,
\[
(I - P^C) \, (\gamma_{P_n}^C - \gamma_P^C)
= (P_n^C - P^C) \gamma_{P_n}^C + P^B (\gamma_{P_n}^B - \gamma_P^B) + (P_n^B - P^B) \gamma_{P_n}^B + P_n^A \gamma_{P_n}^A.
\]
As a consequence,
\[
\gamma_{P_n}^C - \gamma_P^C
= (I - P^C)^{-1} \, ( (P_n^C - P^C) \gamma_{P_n}^C + P^B (\gamma_{P_n}^B - \gamma_P^B) + (P_n^B - P^B) \gamma_{P_n}^B + P_n^A \gamma_{P_n}^A)
\]
Hence,
\begin{align*}
& \| \gamma_{P_n}^C - \gamma_P^C \|\\
& = \| (I - P^C)^{-1} \, ( (P_n^C - P^C) \gamma_{P_n}^C + P^B (\gamma_{P_n}^B - \gamma_P^B) + (P_n^B - P^B) \gamma_{P_n}^B + P_n^A \gamma_{P_n}^A) \|\\
& \leq \| (I - P^C)^{-1} \| \, (\| P_n^C - P^C \| \, \|\gamma_{P_n}^C\| + \| P^B \| \, \| \gamma_{P_n}^B - \gamma_P^B \| + \| P_n^B - P^B \| \, \| \gamma_{P_n}^B \|\\
& \qquad + \| P_n^A \| \, \| \gamma_{P_n}^A \|)\\
& \leq \| (I - P^C)^{-1} \| \, (\| P_n^C - P^C \| + \| P^B \| \, \| \gamma_{P_n}^B - \gamma_P^B \| + \| P_n^B - P^B \| + \| P_n^A \|)\\
& \quad \comment{$\| \gamma_{P_n}^X \| \leq 1$}\\
& \leq \| (I - P^C)^{-1} \| \, (\| P_n^C - P^C \| + |S|^2 \, \| \gamma_{P_n}^B - \gamma_P^B \| + \| P_n^B - P^B \| + \| P_n^A \|)\\
& \quad \comment{$\| P^B \| \leq |S|^2$}
\end{align*}

We need to prove that $\lim_n \gamma_{P_n}^C = \gamma_P^C$ and that this limit exists.  It is sufficient to show that $\limsup_n \|  \gamma_{P_n}^C - \gamma_P^C \| = 0$.  From the above we can conclude that this holds, as
\begin{align*}
& \limsup_n \| \gamma_{P_n}^C - \gamma_P^C \|\\
& \leq \limsup_n \| (I - P^C)^{-1} \| \, (\| P_n^C - P^C \| + |S|^2 \, \| \gamma_{P_n}^B - \gamma_P^B \| + \| P_n^B - P^B \|\\
& \qquad + \| P_n^A \|)\\
& = \| (I - P^C)^{-1} \| \, (\limsup_n \| P_n^C - P^C \| + |S|^2 \, \limsup_n \| \gamma_{P_n}^B - \gamma_P^B \|\\
& \qquad + \limsup_n \| P_n^B - P^B \| + \limsup_n \| P_n^A \|)\\
& = \| (I - P^C)^{-1} \| \, (|S|^2 \, \limsup_n \| \gamma_{P_n}^B - \gamma_P^B \|)
\comment{(\ref{equation:limit-matrix})}\\
& = 0
\comment{$\limsup_n \| \gamma_{P_n}^B - \gamma_P^B \| = 0$ by induction}
\end{align*}
This proves (\ref{equation:lim-gamma-is-gamma}).

Assume that $(s, t)$ belongs to communication class $C$.  Then
\begin{align*}
\limsup_n \delta_{\tau_n}(s, t)
& \leq \limsup_n \gamma_{P_n}(s, t)
\comment{$\delta_{\tau_n} \sqsubseteq \gamma_{P_n}$ by Theorem~\ref{proposition:minimal-coupling}}\\
& = \limsup_n \gamma_{P_n}^C(s, t)
\comment{$(s, t) \in C$}\\
& = \gamma_P^C(s, t)
\comment{(\ref{equation:lim-gamma-is-gamma})}\\
& = \gamma_P(s, t)
\comment{$(s, t) \in C$}\\
& = \delta_\tau(s, t)
\comment{$P \in \mathcal{P}_\tau$ is optimal for $(s, t)$}
\end{align*}
Hence, the function $\delta_{\_}(s, t)$ is upper semi-continuous at $\tau$.
\qed \end{proof}

\begin{proof}[of Theorem~\ref{theorem:continuous}]
Follows directly from Proposition~\ref{proposition:lower-semi-continuous} and Lemma~\ref{lemma:continuous}.
\qed \end{proof}

\section{Robust Probabilistic Bisimilarity}

\begin{theorem}[{\cite[Theorem 2.1.30]{T18}}]
\label{theorem:pseudometric}
$\delta_\tau$ is a pseudometric.
\end{theorem}

\begin{proposition}
\label{proposition:bisim-equivalence}
$\sim$ is an equivalence relation.
\end{proposition}
\begin{proof}
Let $s \in S$.  By Theorem~\ref{theorem:pseudometric}, $\delta_\tau(s, s) = 0$ and, hence, $s \sim s$ by Theorem~\ref{theorem:distance-zero}.  Let $s$, $t \in S$.  Then
\begin{align*}
s \sim t
& \mbox{ iff } \delta_\tau(s, t) = 0
\comment{Theorem~\ref{theorem:distance-zero}}\\
& \mbox{ iff } \delta_\tau(t, s) = 0
\comment{$\delta_\tau(s, t)  = \delta_\tau(t, s)$ by Theorem~\ref{theorem:pseudometric}}\\
& \mbox{ iff } t \sim s
\comment{Theorem~\ref{theorem:distance-zero}}
\end{align*}
Let $s$, $t$, $u \in S$.  Then
\begin{align*}
s \sim t \mbox{ and } t \sim u
& \mbox{ iff } \delta_\tau(s, t) = 0 \mbox{ and } \delta_\tau(t, u) = 0
\comment{Theorem~\ref{theorem:distance-zero}}\\
& \mbox{ iff } \delta_\tau(s, u) = 0
\comment{$\delta_\tau(s, u)  \leq \delta_\tau(s, t) + \delta_\tau(t, u)$ by Theorem~\ref{theorem:pseudometric}}\\
& \mbox{ iff } s \sim u
\comment{Theorem~\ref{theorem:distance-zero}}
\end{align*}
Therefore, $\sim$ is an equivalence relation.
\qed \end{proof}

\begin{proof}[of Lemma~\ref{lemma:rpb-is-pb}]
Let $s$, $t \in S$ such that $s \simeq t$.  By Definition~\ref{definition:probabilistic-bisimilary} we need to show that $\ell(s) = \ell(t)$, there exists $\omega \in \Creal(\tau(s), \tau(t))$ such that $\support(\omega) \subseteq \mathord{\simeq}$, and that $\robust$ is an equivalence relation.

Let $P_{st} \in \mathcal{P}$ be an $S^2_\Delta$-closed policy such that $(s, t)$ reaches $S^2_\Delta$ with probability $1$ in $<S \times S, P_{st}>$.  Such a policy exists according to Proposition~\ref{proposition:delta-closed-policy}.  Since, according to Proposition~\ref{proposition:closed-communication-class}, all pairs of states in $S_1^2$ are closed communication classes, we know that $(s, t) \not\in S_1^2$ and $\ell(s) = \ell(t)$.

Let $\omega = P_{st}(s, t)$, $u$, $v \in S$ and $(u, v) \in \support(\omega)$.  Hence, $\omega(u, v) \gr 0$ and $(u, v)$ is reachable from $(s, t)$.  Therefore, $(u, v)$ must reach $S^2_\Delta$ with probability $1$ in $<S \times S, P_{st}>$.  Consequently, $u \simeq v$. As a result, $\support(\omega) \subseteq \mathord{\simeq}$.

It remains to prove that $\robust$ is an equivalence relation. Clearly $S^2_\Delta \subseteq \robust$, thus, $\robust$ is reflexive.
We can construct $P_{ts}$ such that for all $w$, $x$, $y$, $z \in S$, $P_{ts}(x, w)(z, y) = P_{st}(w, x)(y, z)$. Since $(t, s)$ reaches $S^2_\Delta$ with probability $1$ in $<S \times S, P_{ts}>$, we have $t \simeq s$. Thus, $\robust$ is symmetric.

Let $u \in S$ such that $t \simeq u$.  Then, by Proposition~\ref{proposition:delta-closed-policy}, there exists an $S^2_\Delta$-closed policy $P_{tu} \in \mathcal{P}$ such that $(t, u)$ reaches $S^2_\Delta$ with probability $1$ in $<S \times S, P_{tu}>$.
To show that $\robust$ is transitive, it suffices to show that $s \simeq u$. 
We define the following sets,
\begin{align*}
R_{st} =~ &\{\, (a, b) \in S \times S \mid (a, b) \text{ is reachable from } (s, t) \text{ in } <S \times S, P_{st}> \,\}\\
R_{tu} =~ & \{\, (a, b) \in S \times S \mid (a, b) \text{ is reachable from } (t, u) \text{ in } <S \times S, P_{tu}> \,\}\\
R =~ & \{\, (a, c) \in S \times S \mid b_{(a,c)} \neq \varnothing \,\} \text{, where}\\
b_{(a,c)} =~ & \{\, b \in S \mid (a, b) \in R_{st} \text{ and } (b, c) \in R_{tu} \,\}.
\end{align*}
Let $T$, $U \in S \times S$.  We define $T \bowtie U$ as the set $\{\, (s_1, s_3) \in S \times S \mid \exists s_2 \in S$ such that $(s_1, s_2) \in T$ and $(s_2, s_3) \in U \,\}$. With this notation, $R = R_{st} \bowtie R_{tu}$.

\begin{claim-num}
\label{claim:s-delta}
For all $(a, b) \in R_{st}$, $(a, b)$ has a path to $S^2_\Delta$ in $<S \times S, P_{st}>$ and for all $(a, b) \in R_{tu}$, $(a, b)$ has a path to $S^2_\Delta$ in $<S \times S, P_{tu}>$.
\end{claim-num}
\begin{proof}[of Claim~\ref{claim:s-delta}]
Note that, from the definition of $R_{st}$, for all $(a, b) \in R_{st}$, it holds that $(a, b)$ reaches $S^2_\Delta$ with probability $1$ in $<S \times S, P_{st}>$.  Similarly, for all $(a, b) \in R_{tu}$, $(a, b)$ reaches $S^2_\Delta$ with probability $1$ in $<S \times S, P_{tu}>$. This proves Claim~\ref{claim:s-delta}.
\end{proof}

Let $(a, c) \in R$, we construct $\omega_{(a,c)}$ as follows. For all $(x, z) \in S \times S$, let
\[ \omega_{(a,c)}(x, z) = \frac{\sum_{b \in b_{(a,c)}} \sum_{y \in S} \frac{P_{st}(a, b)(x, y) P_{tu}(b, c)(y, z)}{\tau(b)(y)}}{|b_{(a,c)}|}. \]
Then, for all $x \in S$ we have
\begin{align*}
\omega_{(a,c)}(x,S) & = \sum_{z \in S} \frac{\sum_{b \in b_{(a,c)}} \sum_{y \in S} \frac{P_{st}(a, b)(x, y) P_{tu}(b, c)(y, z)}{\tau(b)(y)}}{|b_{(a,c)}|}\\
& = \frac{\sum_{b \in b_{(a,c)}} \sum_{y \in S} \frac{P_{st}(a, b)(x, y) \tau(b)(y)}{\tau(b)(y)}}{|b_{(a,c)}|} \comment{$P_{tu}(b,c) \in \Omega(\tau(b), \tau(c))$}\\
& = \frac{\sum_{b \in b_{(a,c)}} \tau(a)(x)}{|b_{(a,c)}|} \comment{$P_{st}(a,b) \in \Omega(\tau(a), \tau(b))$}\\
& = \tau(a)(x).
\end{align*}
One can show similarly that $\omega_{(a,c)}(S, z) = \tau(c)(z)$ holds for all $z \in S$. Hence $\omega_{(a,c)} \in \Omega(\tau(a), \tau(c))$.
\begin{claim-num}
\label{claim:coupling}
For all $(a, c) \in R$, we have
\[
\support(\omega_{(a,c)}) = \bigcup_{b \in b_{(a, c)}} \big( \support(P_{st}(a, b)) \bowtie \support(P_{tu}(b, c)) \big) \subseteq R.
\]
\end{claim-num}
\begin{proof}[of Claim~\ref{claim:coupling}]
Assume that $(a, c) \in R$.  First, we show $\support(\omega_{(a,c)}) \supseteq \bigcup_{b \in b_{(a, c)}} \support(P_{st}(a, b)) \bowtie \support(P_{tu}(b, c))$. Let $(x, z) \in \support(P_{st}(a, b))$ $\bowtie \support(P_{tu}(b, c))$ for some $b \in b_{(a, c)}$.  Then there exists $y \in S$ such that $(x, y) \in \support(P_{st}(a, b))$ and $(y, z) \in \support(P_{tu}(b, c))$.  By the definition of $\omega_{(a,c)}$, we have $(x, z) \in \support(\omega_{(a,c)})$.
To show the other inclusions, let $(x, z) \in \support(\omega_{(a,c)})$.  Then there exist $b \in b_{(a, c)}$ and $y \in S$ such that  $(x, y) \in \support(P_{st}(a, b))$ and $(y, z) \in \support(P_{tu}(b, c))$. Therefore, $(x, z) \in \support(P_{st}(a, b)) \bowtie \support(P_{tu}(b, c))$.
Moreover, since $b \in b_{(a, c)}$, we have $(a, b) \in R_{st}$ and $(b, c) \in R_{tu}$. It follows that $(x, y) \in R_{st}$ and $(y, z) \in R_{tu}$.  Hence, $(x, z) \in R$. Thus, $\support(\omega_{(a,c)}) \subseteq R$. This proves Claim~\ref{claim:coupling}.
\end{proof}


Let $P \in \mathcal{P}$ be a policy such that $P(a,c) = \omega_{(a,c)}$ for all $(a, c) \in R$.
We have $(s, u) \in R_{st} \bowtie R_{tu} = R$.  By Claim~\ref{claim:coupling}, for all $(x, z)$ reachable from $(s, u)$ in $<S \times S, P>$, we have $(x, z) \in R$.  Let $C$ be any closed communication class that $(s, u)$ can reach in $<S \times S, P>$.  To show that $s \simeq u$ it suffices to show that $C \cap S^2_\Delta \neq \varnothing$.

Let $(x_1, z_1) \in C$, then $(x_1, z_1) \in R$ and there exists $y_1 \in b_{(x_1, z_1)}$.  
By Claim~\ref{claim:s-delta}, $(x_1, y_1) \in R_{st}$ has a path to $S^2_\Delta$ in $<S \times S, P_{st}>$.
Let $(x_1, y_1), \ldots, (x_n, y_n)$ be this path in $<S \times S, P_{st}>$, where $x_n = y_n$.
Since $y_1, \ldots, y_n$ is a path in the original Markov chain $<S, \tau>$, there is also a path $z_1, \ldots, z_n$ in $<S, \tau>$ such that $(y_1, z_1), \ldots, (y_n, z_n)$ is a path in $<S \times S, P_{tu}>$ and $(y_i, z_i) \in R_{tu}$ for all $1 \leq i \leq n$.  Then, by Claim~\ref{claim:coupling}, we have $(x_i, z_i) \in \support(P(x_{i-1}, z_{i-1}))$ for all $2 \leq i \leq n$.  Hence, there exists a path $(x_1, z_1), \ldots, (x_n, z_n)$ in $<S \times S, P>$.

Similarly, by Claim~\ref{claim:s-delta}, $(x_n, z_n) = (y_n, z_n) \in R_{tu}$ has a path to $S^2_\Delta$ in $<S \times S, P_{tu}>$.  Let $(y_n, z_n), \ldots, (y_m, z_m)$ be this path in $<S \times S, P_{tu}>$, where $y_m = z_m$.  Furthermore, since $y_n, \ldots, y_m$ is a path in the original Markov chain $<S, \tau>$ and $P_{st}$ is an $S^2_\Delta$-closed policy, there is also a path $x_n, \ldots, x_m$ in $<S, \tau>$ such that $(x_n, y_n), \ldots, (x_m, y_m)$ is a path in $<S \times S, P_{st}>$ and $(x_i, y_i) \in R_{st} \cap S^2_\Delta$ for all $n \leq i \leq m$.  Then, by Claim~\ref{claim:coupling}, we have $(x_i, z_i) \in \support(P(x_{i-1}, z_{i-1}))$ for all $n+1 \leq i \leq m$.  Hence, there exists a path $(x_n, z_n), \ldots, (x_m, z_m)$ in $<S \times S, P>$.  See Figure~\ref{figure:equivalence}.
Thus, we have $(x_m, z_m) \in C$ and $x_m = z_m$, that is, $(x_m, z_m) \in C \cap S^2_\Delta$, as required.
\qed \end{proof}

\begin{figure}[ht]
\begin{center}
\begin{tikzpicture}[decoration={snake, pre length=1pt, post length=2pt, amplitude=1.3pt}, yscale=0.6, xscale=0.7]
    \node (0) at (2, 3) {$\simeq$};
    \node (0) at (4, 3) {$\simeq$};
    \node (0) at (2, 5) {$\simeq$};
    \node (0) at (4, 5) {$\simeq$};
    \node[smallstate] (1a) at (1,5) {$s$};
    \node[smallstate] (2a) at (3,5) {$t$};
    \node[smallstate] (3a) at (5,5) {$u$};
    \node[smallstate] (1) at (1,3) {$x_1$};
    \node[smallstate] (2) at (3,3) {$y_1$};
    \node[smallstate] (3) at (5,3) {$z_1$};
    \node[smallstate] (4) at (2,1) {$x_n$};
    \node (0) at (3, 1) {$\simeq$};
    \node[smallstate] (5) at (4,1) {$z_n$};
    \node[smallstate] (6) at (3,-1) {$x_m$};

    \draw[->,decorate] (1a) to (1);
    \draw[->,decorate] (2a) to (2);
    \draw[->,decorate] (3a) to (3);
    \draw[->,decorate] (1) to (4);
    \draw[->,decorate] (2) to (4);
    \draw[->,decorate] (3) to (5);
    \draw[->,decorate] (4) to (6);
    \draw[->,decorate] (5) to (6);
\end{tikzpicture}
\end{center}
\caption{Illustration of the proof of Lemma~\ref{lemma:rpb-is-pb}}
\label{figure:equivalence}
\end{figure}

\begin{proof}[of Proposition~\ref{proposition:robust-probabilistic-bisimulation}]
Let $s$, $t \in S$.  Assume that $s \simeq t$.  According to Lemma~\ref{lemma:rpb-is-pb}, $\robust$ is a bisimulation.  By Definition~\ref{definition:robust-probabilistic-bisimulation} we need to show that there exists a policy $P \in \mathcal{P}$ such that $\mathord{\simeq}$ supports a path from $(s, t)$ to $S^2_\Delta$ in $<S \times S, P>$.

Let $P \in \mathcal{P}$ be the policy such that $(s, t)$ reaches $S^2_\Delta$ with probability $1$ in $<S \times S, P>$.  Write $s_1 = s$ and $t_1 = t$.  Since $(s, t)$ reaches $S^2_\Delta$ with probability $1$ in $<S \times S, P>$, there is a path $(s_1, t_1), \ldots, (s_n, t_n)$ in $<S \times S, P>$, with $s_n = t_n$ and for all $i \ls n$, $s_i \neq t_i$.  For all $i \leq n$, $(s_i, t_i)$ is reachable from $(s, t)$.  Therefore, $(s_i, t_i)$ must reach $S^2_\Delta$ with probability $1$ in $<S \times S, P>$.  Consequently, $s_i \simeq t_i$.  Similarly, for each $(u, v) \in \support(P(s_i, t_i))$, $(u, v)$ must reach $S^2_\Delta$ with probability $1$ in $<S \times S, P>$.  Hence, $u \simeq v$ and, as a result, $\support(P(s_i, t_i)) \subseteq \mathord{\simeq}$.  Thus, $\mathord{\simeq}$ supports a path from $(s, t)$ to $S^2_\Delta$ in $<S \times S, P>$.
\qed \end{proof}

\begin{proposition}
\label{proposition:nwc-equivalence-relation}
For all $\mu$, $\nu \in \Sreal(X)$, given an equivalence relation $R \subseteq X \times X$ such that for all $R$-equivalence classes $A$, $\mu(A) = \nu(A)$, one can compute in $\mathcal{O}(|R|)$ time a coupling $\omega \in \Creal(\mu, \nu)$ such that $\support(\omega) \subseteq R$.
\end{proposition}
\begin{proof}
Let $\mu$, $\nu \in \Sreal(X)$ and $R \subseteq X \times X$ be an equivalence relation such that for all $R$-equivalence classes $A$, $\mu(A) = \nu(A)$.  Let $A \subseteq X$ be an $R$-equivalence class and $\mu_A$, $\nu_A \in \Sreal(A)$ such that
\begin{align*}
\mu_A(x) & = \left \{
\begin{array}{ll}
\mu(x) & \hspace{1cm} \mbox{if $x \in A$}\\
0 & \hspace{1cm} \mbox{otherwise}
\end{array}
\right .\\
\nu_A(x) & = \left \{
\begin{array}{ll}
\nu(x) & \hspace{1cm} \mbox{if $x \in A$}\\
0 & \hspace{1cm} \mbox{otherwise.}
\end{array}
\right .
\end{align*}
Since $\mu(A) = \nu(A)$ and, thus, $\mu_A(X) = \nu_A(X)$, we know that there exists a coupling of $\mu_A$ and $\nu_A$ \cite[Lemma~1]{F56}.  The North-West corner method \cite{H41} constructs a coupling $\omega_A \in \Creal(\mu_A, \nu_A)$ in $\mathcal{O}(|A|)$ time.  Note that $\support(\omega_A) \subseteq A \times A \subseteq R$.

Let $\omega$ be the sum of $\omega_A$ over all $R$-equivalence classes $A$.  Let $x \in X$ and $B \subseteq X$ be the $R$-equivalence class such that $x \in B$.  Then for all $y \in X$, we have $\omega(x, y) = \omega_B(x, y)$.  Therefore, $\omega \in \Creal(\mu, \nu)$ such that $\support(\omega) \subseteq R$.
\qed \end{proof}




For $\mu$, $\nu \in \Sreal(S)$ and an equivalence relation $R \subseteq S \times S$ such that for all $R$-equivalence classes $A$, $\mu(A) = \nu(A)$, we say that $\omega \in \Creal(\mu, \nu)$ is a \emph{maximal $R$-support coupling} if $\support(\omega) = (\support(\mu) \times \support(\nu)) \cap R$.  Moreover, we say that $P \in \mathcal{P}$ is a \emph{maximal $R$-support policy} if for all $(s, t) \in R \cap S^2_{0?}$, the coupling $P(s, t)$ is a maximal $R$-support coupling, that is, $\support(P(s, t)) = \mathrm{Post}((s, t)) \cap R$.

\begin{proposition}
\label{proposition:maximal-support-coupling}
For all $\mu$, $\nu \in \Sreal(X)$, given an equivalence relation $R \subseteq X \times X$ such that for all $R$-equivalence classes $A$, $\mu(A) = \nu(A)$, there exists a maximal $R$-support coupling $\omega \in \Creal(\mu, \nu)$.
\end{proposition}
\begin{proof}
Let $\mu$, $\nu \in \Sreal(X)$ and $R \subseteq X \times X$ be an equivalence relation such that for all $R$-equivalence classes $A$, $\mu(A) = \nu(A)$.  For each $x \in X$, let $L_{x}^{(1)}$ be the set $\{ s \in X \mid (x, s) \in (\support(\mu) \times \support(\nu)) \cap R \}$ and $L_{x}^{(2)}$ be the set $\{ s \in X \mid (s, x) \in (\support(\mu) \times \support(\nu)) \cap R \}$.  We assign $\omega_{1}$, $\mu'$ and $\nu'$ as follows:
\begin{lstlisting}[numberstyle=\tiny,numbers=left,xleftmargin=1.5em]
$\mu' \leftarrow \mu$
$\nu' \leftarrow \nu$
for each $(u, v) \in (\support(\mu) \times \support(\nu)) \cap R$
  $\displaystyle p \leftarrow \min \left( \frac{\mu(u)}{|L_{u}^{(1)}|}, \frac{\nu(v)}{|L_{v}^{(2)}|} \right)$
  $\omega_{1}(u, v) \leftarrow p$
  $\mu'(u) \leftarrow \mu'(u) - p$
  $\nu'(v) \leftarrow \nu'(v) - p$
\end{lstlisting}
Initially, for all $R$-equivalence classes $A$, $\mu'(A) = \nu'(A)$.  In each iteration of the loop above, $(u, v) \in R$, and therefore lines~6 and 7 preserve this property.  At the end we have $\support(\omega_{1}) = (\support(\mu) \times \support(\nu)) \cap R$.  We can then construct a coupling $\omega_{2} \in \Creal(\mu', \nu')$ with $\support(\omega_{2}) \subseteq R$ as described in Proposition~\ref{proposition:nwc-equivalence-relation}.  Define $\omega = \omega_{1} + \omega_{2}$.  Observe that $\omega \in \Creal(\mu, \nu)$ and $\support(\omega) = (\support(\mu) \times \support(\nu)) \cap R$.
\qed \end{proof}

\begin{proposition}
\label{proposition:maximal-support-policy}
For any bisimulation $R \subseteq S \times S$, a maximal $R$-support policy $P \in \mathcal{P}$ exists.
\end{proposition}
\begin{proof}
Let $R \subseteq S \times S$ be a bisimulation. Define $P \in \mathcal{P}$ to be a policy such that for all $(s, t) \in R \cap S^2_{0?}$, $P(s, t) \in \Creal(\tau(s), \tau(t))$ is a maximal $R$-support coupling.  Such a policy $P$ exists by Proposition~\ref{proposition:maximal-support-coupling}.  It is a maximal $R$-support policy.
\qed \end{proof}

\begin{proof}[of Proposition~\ref{proposition:robust-bisimilarity}]
Let $R \subseteq S \times S$ be a robust bisimulation and $P \in \mathcal{P}$ be an $S^2_\Delta$-closed maximal $R$-support policy.  Let $(s, t) \in R$.  Then there exists a policy $P_{st} \in \mathcal{P}$ such that $R$ supports a path from $(s, t)$ to $S^2_\Delta$ in $<S \times S, P_{st}>$.  Thus, $R$ supports the same path in $<S \times S, P>$ as well.  Therefore, for every $(s, t) \in R$, $R$ supports a path from $(s, t)$ to $S^2_\Delta$ in $<S \times S, P>$.

Let $s$, $t \in S$.  Assume that $(s, t) \in R$.  According to, for example, \cite[Theorem~10.27]{BK08}, $(s, t)$ almost surely reaches a closed communication class in $<S \times S, P>$.  By Proposition~\ref{proposition:closed-communication-class}, a closed communication class, say $C$, is a subset of $S^2_{1}$, $S^2_{0,\tau}$, or $S^2_{\Delta}$.  Since $\support(P(u, v)) \subseteq R$, for all $(u, v) \in R$, $(s, t)$ cannot leave $R$. By the definition of robust bisimilarity, we know that $S^2_{1} \cap \robust = \varnothing$, thus, $S^2_{1} \cap R = \varnothing$.  Observe that for all closed communication classes $C$ of $<S \times S, P>$ with $C \subseteq S^2_{0,\tau}$, we have $C \cap R = \varnothing$, as each $(s, t) \in R$ reaches $S^2_{\Delta}$ in $<S \times S, P>$.  Therefore, we can conclude that $(s, t)$ reaches $S^2_{\Delta}$ with probability $1$ in $<S \times S, P>$.
\qed \end{proof}

\section{Algorithm}

\begin{proof}[of Proposition~\ref{proposition:monotonicity}]
Clearly, $\mathrm{Bisim}$ is monotone with respect to $\subseteq$.

Let $A$, $B \subseteq S \times S$, with $A \subseteq B$. Then for all $(s, t) \in \mathrm{Filter}(A)$ there exists a policy $P \in \mathcal{P}$ such that $A$ supports a path from $(s, t)$ to $S^2_{\Delta}$ in $<S \times S, P>$. Since $A \subseteq B$, $B$ supports the same path from $(s, t)$ to $S^2_{\Delta}$ in $<S \times S, P>$. Thus, $(s, t) \in \mathrm{Filter}(B)$.

Let $s$, $t$, $u \in S$. Let $A = \{(s, s), (t, t), (u, u), (s, t), (t, s)\}$ and $B = \{(s, s), (t, t),$ $(u, u), (s, t), (t, s), (t, u),$ $(u, t)\}$.  $A$ and $B$ are symmetric and reflexive and, thus, can be visualized as an undirected graph as shown in Figure~\ref{figure:monotone}.  Observe that $A \subseteq B$, however, $\mathrm{Prune}(A) = \{(s, s), (t, t), (u, u), (s, t), (t, s)\} \not\subseteq \mathrm{Prune}(B) = \{(s, s), (t, t), (u, u)\}$.  Thus, $\mathrm{Prune}$ is not monotone.
\qed \end{proof}

\begin{proof}[of Proposition~\ref{proposition:fp-refine-robust}]
Let $R \subseteq S \times S$.  We prove the two implications.

Assume that $R$ is a robust bisimulation.  It is sufficient to show that $R$ is a fixed point of $\mathrm{Filter}$, $\mathrm{Prune}$ and $\mathrm{Bisim}$.  By the definitions of the functions, $\mathrm{Filter}(R) \subseteq R$, $\mathrm{Prune}(R) \subseteq R$ and $\mathrm{Bisim}(R) \subseteq R$.  Since $R$ is a robust bisimulation, $R \subseteq \mathrm{Filter}(R)$, $R \subseteq \mathrm{Prune}(R)$ and $R \subseteq \mathrm{Bisim}(R)$.

Assume that $R$ is a fixed point of $\mathrm{Refine}$, thus $R = \mathrm{Refine}(R)$.  It follows that $R = \mathrm{Bisim}(R)$ and $R = \mathrm{Filter}(R)$.  Then, $R$ is a bisimulation and for every $(s, t) \in R$ there exists a policy $P \in \mathcal{P}$ such that $R$ supports a path from $(s, t)$ to $S^2_\Delta$ in $<S \times S, P>$.  Therefore, $R$ is a robust bisimulation.
\qed \end{proof}

\begin{proposition}
\label{proposition:filter-max-policy}
For all $R \in [S^2_{\Delta}, \mathord{\sim}]_\mathcal{B}$ and for all maximal $R$-support policies $P \in \mathcal{P}$, we have $\mathrm{Filter}(R) = \{\, (s, t) \in R \mid \exists$ path from $(s, t)$ to $S^2_{\Delta}$ in $<S \times S, P> \,\}$.
\end{proposition}
\begin{proof}
Let $R \in [S^2_{\Delta}, \mathord{\sim}]_\mathcal{B}$ and $(s_1, t_1) \in R$.  Let $P \in \mathcal{P}$ be a maximal $R$-support policy.  We prove the two inclusions.

Assume that $(s_1, t_1) \in \mathrm{Filter}(R)$.  Then there exists a policy $P_{1} \in \mathcal{P}$ such that $R$ supports a path $(s_1, t_1), \ldots, (s_n, t_n)$ in $<S \times S, P_{1}>$, where $s_n = t_n$.  Hence, for all $1 \leq i \leq n$ we have $(s_i, t_i) \in R$.  Thus, the same path $(s_1, t_1), \ldots, (s_n, t_n)$ exists in $<S \times S, P>$.

Assume that there exists a path from $(s_1, t_1)$ to $S^2_{\Delta}$ in $<S \times S, P>$.  Below we prove, by induction, that for all $1 \leq i \leq n$ we have $(s_i, t_i) \in R$ and $\support(P(s_i, t_i)) \subseteq R$.  Therefore, $R$ supports the same path $(s_1, t_1), \ldots, (s_n, t_n)$ in $<S \times S, P>$. Thus, $(s_1, t_1) \in \mathrm{Filter}(R)$.

In the base case, we know that $(s_1, t_1) \in R$ and $\support(P(s_1, t_1)) \subseteq R$, as $P$ is a maximal $R$-support policy.  In the inductive case, assume that $(s_i, t_i) \in R$ and $\support(P(s_i, t_i)) \subseteq R$.  Since $(s_{i+1}, t_{i+1}) \in \support(P(s_i, t_i))$, we have $(s_{i+1}, t_{i+1}) \in R$.  Hence, $\support(P(s_{i+1}, t_{i+1})) \subseteq R$.
\qed \end{proof}

In the following we use implicitly the characterization of $\mathrm{Filter}$ from Proposition~\ref{proposition:filter-max-policy}.

\begin{proposition}
\label{proposition:filter-sr}
For all $R \in [S^2_{\Delta}, \mathord{\sim}]_\mathcal{B}$, 
$\mathrm{Filter}(R)$ is symmetric and reflexive.
\end{proposition}
\begin{proof}
Let $R \in [S^2_{\Delta}, \mathord{\sim}]_\mathcal{B}$.  By the definition of the function, $S^2_{\Delta} \subseteq \mathrm{Filter}(R)$, hence $\mathrm{Filter}(R)$ is reflexive.  Let $s_1$, $t_1 \in S$ and $P \in \mathcal{P}$ be a maximal $R$-support policy.  Assume that $(s_1, t_1) \in \mathrm{Filter}(R)$.  Then there exists a path $(s_1, t_1), \ldots, (s_n, t_n)$ in $<S \times S, P>$, where $s_n = t_n$.  By the definition of a maximal $R$-support policy, for all $1 \leq i \leq n$ we have $(s_i, t_i) \in R$.  Since $R$ is an equivalence relation, for all $1 \leq i \leq n$ we have $(t_i, s_i) \in R$.  Thus, there exists a path $(t_1, s_1), \ldots, (t_n, s_n)$ in $<S \times S, P>$, where $t_n = s_n$.  Since $P$ is a maximal $R$-support policy, it follows from Proposition~\ref{proposition:filter-max-policy} that $(t, s) \in \mathrm{Filter}(R)$. Hence, $\mathrm{Filter}(R)$ is symmetric.
\qed \end{proof}

\begin{proposition}
\label{proposition:prune-equivalence}
For all $R \in [S^2_{\Delta}, \mathord{\sim}]$ such that $R$ is symmetric and reflexive, $\mathrm{Prune}(R)$ is an equivalence relation.
\end{proposition}
\begin{proof}
Let $R \in [S^2_{\Delta}, \mathord{\sim}]$ be symmetric and reflexive.  It is sufficient to show that $\mathrm{Prune}(R)$ is reflexive, symmetric and transitive.

By the definition of the function, $S^2_{\Delta} \subseteq \mathrm{Prune}(R)$, hence $\mathrm{Prune}(R)$ is reflexive.

Let $s$, $t$, $u \in S$.  Assume that $(s, t) \in \mathrm{Prune}(R)$, then $\forall (t, u) \in R : (s, u) \in R$ and $\forall (u, s) \in R : (u, t) \in R$.  Since $R$ is symmetric, $\forall (u, t) \in R : (u, s) \in R$ and $\forall (s, u) \in R : (t, u) \in R$.  Thus, $(t, s) \in \mathrm{Prune}(R)$ and $\mathrm{Prune}(R)$ is symmetric.

Lastly, assume that $(s, t)$, $(t, u) \in \mathrm{Prune}(R)$.  Then $(s, t) \in R$, $(t, u) \in R$ and we have that $\forall (t, x) \in R : (s, x) \in R$, $\forall (x, s) \in R : (x, t) \in R$, $\forall (u, x) \in R : (t, x) \in R$, and $\forall (x, t) \in R : (x, u) \in R$.  Therefore, $(s, u) \in R$, $\forall (u, x) \in R : (s, x) \in R$, and  $\forall (x, s) \in R : (x, u) \in R$.  Thus, $(s, u) \in \mathrm{Prune}(R)$.  Hence, $\mathrm{Prune}(R)$ is transitive.
\qed \end{proof}

\begin{proposition}
\label{proposition:stabilise}
Given an equivalence relation $E \in [S^2_{\Delta}, \mathord{\sim}]$, $\mathrm{Bisim}(E)$ can be computed in polynomial time.
\end{proposition}
\begin{proof}
Let $E \in [S^2_{\Delta}, \mathord{\sim}]$ be an equivalence relation.  The largest bisimulation $E' \subseteq E$ exists, since the transitive closure of the union of all bisimulations $E' \subseteq E$, is also $\subseteq E$.  $\mathrm{Bisim}(E)$ can be computed in polynomial time, for example, by using the partition refinement algorithm by Derisavi et al. \cite{DHS03}. 
The algorithm takes as input an equivalence relation $E$ and returns the largest equivalence relation $E' \subseteq E$ such that for all $(s, t) \in E'$ and for all $E'$-equivalence classes $C$, we have $\tau(s)(C) = \tau(t)(C)$.  This is done by selecting a single equivalence class $X$ from the current partition at each iteration and then refining the partition by comparing $\tau(s)(X)$ for each $s \in S$, until a fixed point is reached.  Since for all $(s, t) \in E$, $\ell(s) = \ell(t)$, it follows that $E'$ is the largest bisimulation $\subseteq E$.
\qed \end{proof}

\begin{proposition}
\label{proposition:filter}
Algorithm~\ref{algorithm:filter} computes $\mathrm{Filter}$.
\end{proposition}
\begin{proof}
Let $R \in [S^2_{\Delta}, \mathord{\sim}]_\mathcal{B}$ and $P \in \mathcal{P}$ be a maximal $R$-support policy.  We show the following loop invariant of Algorithm~\ref{algorithm:filter}: $Q = \{\, (s, t) \in R \mid \exists$ path of length $\leq n$ from $(s, t)$ to $S^2_{\Delta}$ in $<S \times S, P> \,\}$.

$Q$ is initialized to $S^2_{\Delta}$.  For all $(s, t) \in S^2_{\Delta}$, $(s, t)$ can reach $S^2_{\Delta}$ in $0$ steps in $<S \times S, P>$.  Hence, the loop invariant holds before the loop.

Assume that the loop invariant holds before an iteration of the loop, that is, $Q = Q_{\mathrm{old}} = \{\, (s, t) \in R \mid \exists$ path of length $\leq n$ from $(s, t)$ to $S^2_{\Delta}$ in $<S \times S, P> \,\}$.  Let $s$, $t \in S$ and $(s, t) \in R$.  We need to show that $(s, t)$ is added to $Q$ on line~6 if and only if there exists a shortest path of length $n + 1$ from $(s, t)$ to $S^2_{\Delta}$ in $<S \times S, P>$.  We prove the two implications.

Assume that $(s, t)$ is added to $Q$.  Then $(s, t) \in (R \cap \mathrm{Pre}(Q_{\mathrm{old}})) \setminus Q_{\mathrm{old}}$.  Thus, there is $(u, v) \in \mathrm{Post}((s, t)) \cap Q_{\mathrm{old}}$.  Since $P$ is a maximal $R$-support policy and $Q_{\mathrm{old}} \subseteq R$, we have $(u, v) \in \support(P(s, t)) \cap Q_{\mathrm{old}}$.  By the induction hypothesis and since $(s, t) \not\in Q_{\mathrm{old}}$, there exists a shortest path of length $n$ from $(u, v)$ to $S^2_{\Delta}$ in $<S \times S, P>$.  Therefore, there exists a shortest path of length $n + 1$ from $(s, t)$ to $S^2_{\Delta}$ in $<S \times S, P>$.

To prove the other implication, assume that there exists a shortest path of length $n + 1$ from $(s, t)$ to $S^2_{\Delta}$ in $<S \times S, P>$.  Let $u$, $v \in S$ and $(u, v) \in R$ such that the path is $(s, t), (u, v), \hdots ,S^2_{\Delta}$.  Then there exists a shortest path of length $n$ from $(u, v)$ to $S^2_{\Delta}$ in $<S \times S, P>$.  Hence, by the induction hypothesis, $(u, v) \in Q_{\mathrm{old}}$ and $(s, t) \in (R \cap \mathrm{Pre}(Q_{\mathrm{old}})) \setminus Q_{\mathrm{old}}$.  Therefore, $(s, t)$ is added to $Q$.

Hence, $Q = \{\, (s, t) \in R \mid \exists$ path of length $\leq n + 1$ from $(s, t)$ to $S^2_{\Delta}$ in $<S \times S, P> \,\}$ and, thus, the loop invariant is maintained in each iteration of the loop.

The loop terminates when a fixed point is reached, therefore, by the loop invariant we know that there are no pairs of states $(s, t) \in R \setminus Q$ such that $(s, t)$ can reach $S^2_{\Delta}$ in $<S \times S, P>$ with a shortest path of length $n$.  It follows that there are no pairs of states $(s, t) \in R \setminus Q$ such that $(s, t)$ can reach $S^2_{\Delta}$ in $<S \times S, P>$ with a shortest path of length $\geq n$.  Hence, $Q = \{\, (s, t) \in R \mid \exists$ path from $(s, t)$ to $S^2_{\Delta}$ in $<S \times S, P> \,\}$.
\qed \end{proof}

\begin{proposition}
\label{proposition:prune}
Algorithm~\ref{algorithm:prune} computes $\mathrm{Prune}$.
\end{proposition}
\begin{proof}
Let $Q \in [S^2_{\Delta}, \mathord{\sim}]$.  Let $s$, $t \in S$ and $(s, t) \in Q$.  Initially, $E = Q$.  Observe that $(s, t)$ is removed from $E$ on line~5 if and only if there exists $u \in S$ such that $(t, u) \in Q \wedge (s, u) \not\in Q$ or $(u, s) \in Q \wedge (u, t) \not\in Q$.  Therefore, $E = \mathrm{Prune}(Q)$.
\qed \end{proof}

\begin{proof}[of Proposition~\ref{proposition:transitive}]
Let $R \in [\simeq, \mathord{\sim}]_\mathcal{B}$ and $s$, $t$, $u \in S$ with $(s, t)$, $(t, u) \in \mathrm{Filter}(R)$.  We show that if $t \simeq u$ then $(s, u) \in \mathrm{Filter}(R)$.  The case $s \simeq t$ is similar.

Since $\mathrm{Filter}(R) \subseteq R$, we have $(s, t)$, $(t, u) \in R$.  Since $R$ is an equivalence relation, $(s, u) \in R$.  Write $s_1 = s$ and $t_1 = t$ and $u_1 = u$.  Let $P \in \mathcal{P}$ be a maximal $R$-support policy.  Since $(s, t) \in \mathrm{Filter}(R)$, there exists a path $(s_1, t_1), \ldots, (s_n, t_n)$ in $<S \times S, P>$, where $s_n = t_n$.

Assume that $(t, u) \in \mathord{\simeq}$.  By Proposition~\ref{proposition:robust-probabilistic-bisimulation}, $\simeq$ is a bisimulation.  Since $t_1, \ldots, t_n$ is a path in the original Markov chain $<S, \tau>$, there is also a path $u_1, \ldots, u_n$ in $<S, \tau>$ such that $(t_i, u_i) \in \mathord{\simeq}$ for all $1 \leq i \leq n$.  In particular, $(t_n, u_n) \in \mathord{\simeq}$.  Since $\mathord{\simeq} \subseteq R$, there exists a path $(t_1, u_1), \ldots, (t_n, u_n)$ in $<S \times S, P>$.  Note that $(s_i, u_i) \in R$ for all $1 \leq i \leq n$.  Hence, there exists a path $(s_1,u_1), \ldots, (s_n,u_n) = (t_n,u_n)$ in $<S \times S, P>$.  See Figure~\ref{figure:transitive}.

Since $(t_n, u_n) \in \mathord{\simeq}$, there exists a policy $P' \in \mathcal{P}$ such that $(t_n, u_n)$ reaches $S^2_{\Delta}$ with probability $1$.  Therefore, there is a path $(t_n, u_n), \ldots, (t_m, u_m)$ in $<S \times S, P'>$, with $t_m = u_m$ and $(t_i, u_i) \in \mathord{\simeq}$ for all $n \leq i \leq m$.  Since $\mathord{\simeq} \subseteq R$, the same path $(t_n, u_n), \ldots, (t_m, u_m)$ exists in $<S \times S, P>$, with $t_m = u_m$.  Thus, there exists paths $(s_1,u_1), \ldots, (s_n,u_n)$ and $(t_n,u_n), \ldots, (t_m,u_m)$, with $s_n = t_n$ and $t_m=u_m$, in $<S \times S, P>$.  Hence, $(s,u) \in \mathrm{Filter}(R)$.
\qed \end{proof}

\begin{proof}[of Proposition~\ref{proposition:robust-subset-R}]
$R$ is initialized to $\mathord{\sim}$, hence, by Proposition~\ref{proposition:bisim-equivalence} and the definition of $\mathord{\sim}$, the loop invariant holds before the loop.

Assume that the loop invariant holds before an iteration of the loop, that is $R \in [\mathord{\simeq}, \mathord{\sim}]_\mathcal{B}$.  Since $\mathord{\simeq} \subseteq R$ and, by Proposition~\ref{proposition:monotonicity}, $\mathrm{Filter}$ is monotone, we have $\mathrm{Filter}(\mathord{\simeq}) \subseteq \mathrm{Filter}(R)$. According to Proposition~\ref{proposition:fp-refine-robust}, $\mathord{\simeq}$ is a fixed point of $\mathrm{Refine}$.  It follows that $\mathord{\simeq} = \mathrm{Filter}(\mathord{\simeq}) \subseteq \mathrm{Filter}(R)$.  Next we show that $\mathord{\simeq} \subseteq \mathrm{Prune}(\mathrm{Filter}(R))$.  Let $s$, $t \in S$ and $s \simeq t$.  Thus, $(s, t) \in \mathrm{Filter}(R)$.  Then, by Proposition~\ref{proposition:transitive}, for all $(t, u) \in \mathrm{Filter}(R)$ we have $(s, u) \in \mathrm{Filter}(R)$ and for all $(u, s) \in \mathrm{Filter}(R)$ we have $(u, t) \in \mathrm{Filter}(R)$.  Hence, $(s, t) \in \mathrm{Prune}(\mathrm{Filter}(R))$, and we have shown $\mathord{\simeq} \subseteq \mathrm{Prune}(\mathrm{Filter}(R))$.  Since, by Proposition~\ref{proposition:monotonicity}, $\mathrm{Bisim}$ is monotone, we have $\mathord{\simeq} = \mathrm{Bisim}(\mathord{\simeq}) \subseteq \mathrm{Bisim}(\mathrm{Prune}(\mathrm{Filter}(R))) = \mathrm{Refine}(R)$.  By the definition of $\mathrm{Bisim}$, $\mathrm{Refine}(R)$ is a bisimulation, that is, $\mathrm{Refine}(R) \in [\mathord{\simeq}, \mathord{\sim}]_\mathcal{B}$.  Thus, the loop invariant is maintained in each iteration of the loop.
\qed \end{proof}

\begin{proof}[of Theorem~\ref{theorem:refine}]
The loop on lines 1-5 in Algorithm~\ref{algorithm:robust-bisimilarity} can be rewritten as follows:
\begin{lstlisting}[numberstyle=\tiny,numbers=left,xleftmargin=1.5em]
$R \leftarrow \mathord{\sim}$
while $\mathrm{Refine}(R) \subsetneq R$
  $R \leftarrow \mathrm{Refine}(R)$
\end{lstlisting}
It is immediate from the definitions of $\mathrm{Bisim}$, $\mathrm{Filter}$ and $\mathrm{Prune}$ that $\mathrm{Refine}(R) \subseteq R$ holds for all $R \subseteq S \times S$.  Therefore, Algorithm~\ref{algorithm:robust-bisimilarity} is a standard fixed point iteration.  By Proposition~\ref{proposition:robust-subset-R}, $\mathord{\simeq} \subseteq R$, thus, it computes a fixed point of $\mathrm{Refine}$ greater than or equal to $\mathord{\simeq}$.  Since $\mathord{\simeq}$ is the greatest fixed point of $\mathrm{Refine}$, we can conclude that Algorithm~\ref{algorithm:robust-bisimilarity} computes $\mathord{\simeq}$.
\qed \end{proof}

\begin{proposition}
\label{proposition-runtime}
Algorithm~\ref{algorithm:robust-bisimilarity} runs in $\mathcal{O}(n^6)$ time, where $n = |S|$.
\end{proposition}
\begin{proof}
Let $|S| = n$.  $\mathrm{Refine}$ begins with $\sim$, containing at most $n^2$ pairs of states.  Since at least one pair is removed in each iteration, $\mathrm{Refine}$ requires at most $n^2$ iterations.

$\mathrm{Filter}$ checks at most $|R| \leq n^2$ pairs of states per iteration and adds at least one pair of states to $Q$.  Thus, there are at most $n^2$ iterations, with a total runtime of $\mathcal{O}(n^4)$.

$\mathrm{Prune}$ has an outer loop over $|Q| \leq n^2$ pairs of states and an inner loop over at most $|S| = n$ pairs of states. Hence, $\mathrm{Prune}$ runs in $\mathcal{O}(n^3)$ time.

$\mathrm{Bisim}$ runs in $\mathcal{O}(m \log n) = \mathcal{O}(n^2 \log n)$ time \cite{DHS03}.

Therefore, the overall runtime of $\mathrm{Refine}$ is $\mathcal{O}(n^6)$.
\qed \end{proof}

\section{Experiments}
\label{appendix:experiments}

Below is a description of jpf-probabilistic's randomized algorithms utilized in our experiments.
\begin{itemize}
    \item Erd{\"o}s-R\'enyi Model: a model for generating a random (directed or undirected) graph. A graph with a given number of vertices $v$ is constructed by placing an edge between each pair of vertices with a given probability $p$, independent from every other edge.  We check the probability that the generated graph is connected (for every pair of nodes, there is a path). \cite{ER59}
    \item Fair Biased Coin: makes a fair coin from a biased coin, where $p$ denotes the probability by which the biased coin tosses heads.  We check the probability that the coin toss results in heads. \cite{vN51}
    \item Has Majority Element: a Monte Carlo algorithm that determines whether an integer array has a majority element (appears more than half of the time in the array).  The parameter $s$ denotes the size of the given array, $t$ denotes the number of trials, and $m$ denotes the amount of times that the majority element occurs in the array.  We check the probability that the algorithm erroneously reports that the array does not have a majority element. \cite{MR95}
    \item Pollards Integer Factorization: finds a factor of an integer $i$. We check the probability that the algorithm returns $i$, when $i$ is not prime. \cite{P75}
    \item Queens: attempts to place a queen on each row of an $n \times n$ chess board such that no queen can attack another.  We check the probability of success. \cite{B10}
    \item Set Isolation: finds a sample of the universe $U$ that is disjoint from the subset $S$ but not disjoint from the subset $T$.  Let $u$ denote the size of the universe and $st$ denote the size of $S$ and $T$.  We check the probability that the randomly selected sample is good, that is, disjoint from $S$ and intersects $T$. \cite{KM94}
\end{itemize}

}

\clearpage
\bibliographystyle{splncs04}
\bibliography{main}

@book{DP02,
    author    = "Brian Davey and Hilary Priestley",
    title     = "Introduction to lattices and order",
    publisher = "Cambridge University Press",
    year      = "2002",
    address   = "Cambridge, United Kingdom"}

@book{E89,
    author    = "Ryszard Engelking",
    title     = "General topology",
    publisher = "Heldermann Verlag",
    year      = "1989",
    address   = "Berlin, Germany"}

@inproceedings{CBW12,
    author       = "Di Chen and {van Breugel}, Franck and James Worrell",
    title        = "On the complexity of computing probabilistic bisimilarity",
    booktitle    = "Proceedings of the 15th International Conference on Foundations of Software Science and Computational Structures",
    year         = "2012",
    editor       = "Lars Birkedal",
    volume       = "7213",
    series       = "Lecture Notes in Computer Science",
    pages        = "437-451",
    address      = "Tallinn, Estonia",
    month        = mar # "/" # apr,
    publisher    = "Springer-Verlag"}

@phdthesis{T18,
    author  = "Qiyi Tang",
    title   = "Computing probabilistic bisimilarity distances",
    school  = "York University",
    year    = "2018",
    address = "Toronto, Canada",
    month   = aug}

@article{DGJP04,
    author  = "Jos{\'{e}}e Desharnais and Vineet Gupta and Radha Jagadeesan and Prakash Panangaden",
    title   = "Metrics for labelled {M}arkov processes",
    journal = "Theoretical Computer Science",
    year    = "2004",
    volume  = "318",
    number  = "3",
    pages   = "323--354",
    month   = jun}

@book{BK08,
    author    = "Christel Baier and Joost-Pieter Katoen",
    title     = "Principles of model checking",
    publisher = "The MIT Press",
    year      = "2008",
    address   = "Cambridge, MA, USA"}

@book{BP94,
    author    = "Abraham Berman and Robert Plemmons",
    title     = "Nonnegative matrices in the mathematical sciences",
    publisher = "SIAM",
    year      = "1994"}

@book{KS60,
    author    = "John G. Kemeny and J. Laurie Snell",
    title     = "Finite {M}arkov chains",
    publisher = "Springer-Verlag",
    year      = "1960",
    address   = "Heidelberg, Germany"}

@inproceedings{LS89,
    author       = "Kim Larsen and Arne Skou",
    title        = "Bisimulation through probabilistic testing",
    booktitle    = "Proceedings of the 16th Annual ACM Symposium on Principles of Programming Languages",
    year         = "1989",
    pages        = "344--352",
    address      = "Austin, TX, USA",
    month        = jan,
    publisher    = "ACM"}

@inproceedings{JL91,
    author       = "Bengt Jonsson and Kim Larsen",
    title        = "Specification and Refinement of Probabilistic Processes",
    booktitle    = "Proceedings of the 6th Annual Symposium on Logic in Computer Science",
    year         = "1991",
    pages        = "266--277",
    address      = "Amsterdam, The Netherlands",
    month        = jul,
    publisher    = "IEEE"}

@inproceedings{GJS90,
    author       = {Alessandro Giacalone and Chi{-}Chang Jou and Scott A. Smolka},
    editor       = {Manfred Broy and Cliff B. Jones},
    title        = {Algebraic Reasoning for Probabilistic Concurrent Systems},
    booktitle    = {Proceedings of the Working Conference on Programming Concepts and Methods},
    pages        = {443--458},
    publisher    = {North-Holland},
    year         = {1990},
    address      = {Sea of Galilee, Israel},
    month        = apr}

@inproceedings{DGJP99,
    author       = {Jos{\'{e}}e Desharnais and Vineet Gupta and Radha Jagadeesan and Prakash Panangaden},
    editor       = {Jos C. M. Baeten and Sjouke Mauw},
    title        = {Metrics for Labeled {M}arkov Systems},
    booktitle    = {Proceedings of the 10th International Conference on Concurrency Theory},
    series       = {Lecture Notes in Computer Science},
    volume       = {1664},
    pages        = {258--273},
    publisher    = {Springer-Verlag},
    year         = {1999},
    address      = {Eindhoven, The Netherlands},
    month        = aug}

@inproceedings{JMLM14,
    author       = {Manfred Jaeger and Hua Mao and Kim Guldstrand Larsen and Radu Mardare},
    editor       = {Gethin Norman and William H. Sanders},
    title        = {Continuity Properties of Distances for {M}arkov Processes},
    booktitle    = {Proceedings of the 11th International Conference on Quantitative Evaluation of Systems},
    series       = {Lecture Notes in Computer Science},
    volume       = {8657},
    pages        = {297--312},
    publisher    = {Springer-Verlag},
    year         = {2014},
    address      = {Florence, Italy},
    month        = sep}

@inproceedings{TB17,
    author       = {Qiyi Tang and Franck van Breugel},
    editor       = {Roland Meyer and Uwe Nestmann},
    title        = {Algorithms to Compute Probabilistic Bisimilarity Distances for Labelled {M}arkov Chains},
    booktitle    = {Proceedings of the 28th International Conference on Concurrency Theory},
    series       = {LIPIcs},
    volume       = {85},
    pages        = {27:1--27:16},
    publisher    = {Schloss Dagstuhl - Leibniz-Zentrum f{\"{u}}r Informatik},
    year         = {2017},
    address      = {Berlin, Germany},
    month        = sep}

@book{E11,
    author    = {{\c{C}}{\i}nlar, Erhan},
    address   = {New York, NY, US},
    publisher = {Springer-Verlag},
    series    = {Graduate Texts in Mathematics},
    title     = {Probability and stochastics},
    volume    = {261},
    year      = {2011}}

@book{S64,
    author    = {Frank Spitzer},
    title     = {Principles of random walk},
    publisher = {Springer-Verlag},
    year      = {1964},
    address   = {New York, NY, USA},
    series    = {Graduate Texts in Mathematics}}

@book{F56,
    title     = "Hitchcock transportation problem",
    author    = "Delbert Ray Fulkerson",
    year      = "1956",
    publisher = "Rand Corporation"}

@article{H41,
    author  = "Frank Hitchcock",
    title   = {The Distribution of a Product from Several Sources to Numerous Localities},
    journal = "Studies in Applied Mathematics",
    year    = "1941",
    volume  = "20",
    number  = "1/4",
    pages   = "224-230",
    month   = apr}

@inproceedings{D07,
    author    = {Salem Derisavi},
    title     = {Signature-based Symbolic Algorithm for Optimal {Markov} Chain Lumping},
    booktitle = {Proceedings of the 4th International Conference on the Quantitative Evaluation of Systems},
    pages     = {141--150},
    publisher = {{IEEE} Computer Society},
    year      = {2007},
    month     = sep,
    address   = {Edinburgh, Scotland, {UK}}}

@article{DHS03,
    author       = {Salem Derisavi and Holger Hermanns and William H. Sanders},
    title        = {Optimal state-space lumping in {M}arkov chains},
    journal      = {Information Processing Letters},
    volume       = {87},
    number       = {6},
    pages        = {309--315},
    year         = {2003}}

@inproceedings{KNP11,
    author       = {Marta Kwiatkowska and Gethin Norman and David Parker},
    editor       = {Ganesh Gopalakrishnan and Shaz Qadeer},
    title        = {{PRISM} 4.0: Verification of Probabilistic Real-Time Systems},
    booktitle    = {Proceedings of the 23rd International Conference on Computer Aided Verification},
    series       = {Lecture Notes in Computer Science},
    volume       = {6806},
    pages        = {585--591},
    publisher    = {Springer-Verlag},
    year         = {2011},
    address      = {Snowbird, Utah, USA},
    month        = jul}

@article{CAMR10,
    author       = {Krishnendu Chatterjee and Luca de Alfaro and Rupak Majumdar and Vishwanath Raman},
    title        = {Algorithms for Game Metrics (Full Version)},
    journal      = {Logical Methods in Computer Science},
    volume       = {6},
    number       = {3},
    year         = {2010}}

@article{VHBPL03,
    author     = "Visser, Willem and Havelund, Klaus and Brat, Guillaume and Park, Seungjoon and Lerda, Flavio",
    title      = "Model Checking Programs",
    journal    = "Automated Software Engineering",
    volume     = "10",
    number     = "2",
    pages      = "203-232",
    month      = apr,
    year       = "2003"}

@inproceedings{FCDWTB21,
    author       = {Syyeda Zainab Fatmi and Xiang Chen and Yash Dhamija and Maeve Wildes and Qiyi Tang and Franck van Breugel},
    editor       = {Alfons Laarman and Ana Sokolova},
    title        = {Probabilistic Model Checking of Randomized {J}ava Code},
    booktitle    = {Proceedings of the 27th International Symposium on Model Checking Software, {SPIN}},
    series       = {Lecture Notes in Computer Science},
    volume       = {12864},
    pages        = {157--174},
    publisher    = {Springer},
    year         = {2021},
    month        = jul}

@inproceedings{vN51,
    title     = {Various Techniques Used in Connection with Random Digits},
    author    = {von Neumann, John},
    booktitle = {Monte Carlo Method},
    editor    = {Householder, A.~S. and Forsythe, G.~E. and Germond, H.~H.},
    series    = {National Bureau of Standards Applied Mathematics Series},
    volume    = {12},
    chapter   = {13},
    pages     = {36--38},
    year      = {1951},
    publisher = {US Government Printing Office},
    address   = {Washington, DC}}

@inproceedings{KM94,
    author    = {Karger, David R. and Motwani, Rajeev},
    title     = {Derandomization Through Approximation: An {NC} Algorithm for Minimum Cuts},
    booktitle = {Proceedings of the 26th Annual ACM Symposium on Theory of Computing},
    year      = {1994},
    location  = {Montreal, Quebec, Canada},
    pages     = {497--506},
    numpages  = {10},
    publisher = {ACM},
    address   = {New York, NY, USA},
    month     = may}

@unpublished{B10,
    author       = {Barringer, Howard},
    institution  = {University of Manchester},
    year         = {2010},
    title        = {Randomized algorithms - a brief introduction},
    note         = {Lecture at the University of Manchester}}

@book{MR95,
    author    = {Motwani, Rajeev and Raghavan, Prabhakar},
    title     = {Randomized algorithms},
    year      = {1995},
    publisher = {Cambridge University Press},
    address   = {New York, NY, USA}}

@Article{P75,
    author  = "Pollard, John M.",
    title   = "A {Monte Carlo} method for factorization",
    journal = "BIT Numerical Mathematics",
    year    = "1975",
    month   = sep,
    volume  = "15",
    number  = "3",
    pages   = "331--334"}

@article{ER59,
   author    = "Paul Erd{\"o}s and Alfr\'ed R\'enyi",
   title     = "On Random Graphs {I}",
   journal   = " Publicationes Mathematicae",
   volume    = "6",
   pages     = "290-297",
   year      = "1959"}

@inproceedings{HKPQR19,
    author       = {Arnd Hartmanns and Michaela Klauck and David Parker and Tim Quatmann and Enno Ruijters},
    editor       = {Tom{\'{a}}s Vojnar and Lijun Zhang},
    title        = {The Quantitative Verification Benchmark Set},
    booktitle    = {Proceedings of the 25th International Conference on Tools and Algorithms for the Construction and Analysis of Systems},
    series       = {Lecture Notes in Computer Science},
    volume       = {11427},
    pages        = {344--350},
    publisher    = {Springer},
    year         = {2019},
    month        = apr,
    address      = {Prague, Czech Republic}}

@article{EH20,
    author  = {Eastman, J. Ronald and He, Jiena},
    title   = {A Regression-Based Procedure for {M}arkov Transition Probability Estimation in Land Change Modeling},
    journal = {Land},
    volume  = {9},
    year    = {2020},
    number  = {11},
    month   = oct,
    article-number = {407}}

@article{OCHTC17,
    author  = "Olariu, Elena and Cadwell, Kevin K. and Hancock, Elizabeth and Trueman, David and Chevrou-Severac, Helene",
    title   = "Current recommendations on the estimation of transition probabilities in {Markov} cohort models for use in health care decision-making: a targeted literature review",
    journal = "ClinicoEconomics and Outcomes Research",
    year    = 2017,
    volume  = "9",
    pages   = "537-546",
    month   = sep}

@article{SLT21,
    author  = "Tushar Srivastava and Nicholas R. Latimer and Paul Tappenden",
    title   = "Estimation of Transition Probabilities for State-Transition Models: A Review of {NICE} Appraisals",
    journal = "PharmacoEconomics",
    year    = 2021,
    volume  = "39",
    number  = "8",
    pages   = "869-878",
    month   = aug}

@article{MLAK17,
    author  = {Mizutani, Daijiro and Lethanh, Nam and Adey, Bryan T. and Kaito, Kiyoyuki},
    year    = {2017},
    month   = oct,
    pages   = {58},
    title   = {Improving the Estimation of {M}arkov Transition Probabilities Using Mechanistic-Empirical Models},
    volume  = {3},
    journal = {Frontiers in Built Environment}}

@inproceedings{BreugelW01,
    author       = {Franck van Breugel and James Worrell},
    editor       = {Fernando Orejas and Paul G. Spirakis and Jan van Leeuwen},
    title        = {Towards Quantitative Verification of Probabilistic Transition Systems},
    booktitle    = {Proceedings of the 28th International Colloquium on Automata, Languages and Programming},
    series       = {Lecture Notes in Computer Science},
    volume       = {2076},
    pages        = {421--432},
    publisher    = {Springer},
    year         = {2001},
    month        = jul,
    address      = {Crete, Greece}}

@article{ThorsleyKlavins2010,
    author = {D. Thorsley  and E. Klavins},
    title = {Approximating stochastic biochemical processes with {W}asserstein pseudometrics},
    journal = {IET Systems Biology},
    volume = {4},
    issue = {3},
    pages = {193-211},
    year = {2010}}

@inproceedings{ComaniciP11,
    author       = {Gheorghe Comanici and Doina Precup},
    editor       = {Wolfram Burgard and Dan Roth},
    title        = {Basis Function Discovery Using Spectral Clustering and Bisimulation Metrics},
    booktitle    = {Proceedings of the 25th {AAAI} Conference on Artificial Intelligence},
    pages        = {325--330},
    publisher    = {{AAAI} Press},
    year         = {2011},
    address      = {San Francisco, California, USA},
    month        = aug}

@inproceedings{CaiG09,
    author       = {Xiaojuan Cai and Yonggen Gu},
    editor       = {Feng Bao and Hui Li and Guilin Wang},
    title        = {Measuring Anonymity},
    booktitle    = {Proceedings of the 5th International Conference on Information Security Practice and Experience},
    series       = {Lecture Notes in Computer Science},
    volume       = {5451},
    pages        = {183--194},
    publisher    = {Springer},
    year         = {2009},
    month        = apr,
    address      = {Xi'an, China}}

@inproceedings{K83,
    author       = {Dexter Kozen},
    editor       = {David S. Johnson and Ronald Fagin and Michael L. Fredman and David Harel and Richard M. Karp and Nancy A. Lynch and Christos H. Papadimitriou and Ronald L. Rivest and Walter L. Ruzzo and Joel I. Seiferas},
    title        = {A Probabilistic {PDL}},
    booktitle    = {Proceedings of the 15th Annual Symposium on Theory of Computing},
    pages        = {291--297},
    publisher    = {{ACM}},
    year         = {1983},
    month        = apr,
    address      = {Boston, Massachusetts, {USA}}}

@inproceedings{AM01,
    author       = {Luca de Alfaro and Rupak Majumdar},
    editor       = {Jeffrey Scott Vitter and Paul G. Spirakis and Mihalis Yannakakis},
    title        = {Quantitative solution of omega-regular games},
    booktitle    = {Proceedings of the 33rd Annual Symposium on Theory of Computing},
    pages        = {675--683},
    publisher    = {{ACM}},
    year         = {2001},
    month        = jul,
    address      = {Heraklion, Crete, Greece}}

@book{MM05,
    author       = {Annabelle McIver and Carroll Morgan},
    title        = {Abstraction, Refinement and Proof for Probabilistic Systems},
    series       = {Monographs in Computer Science},
    publisher    = {Springer},
    year         = {2004}}

@inproceedings{KKZJ07,
	author       = {Joost{-}Pieter Katoen and Tim Kemna and Ivan S. Zapreev and David N. Jansen},
	title        = {Bisimulation Minimisation Mostly Speeds Up Probabilistic Model Checking},
	booktitle    = {Proceedings of the 13th International Conference on Tools and Algorithms for the Construction and Analysis of Systems},
	series       = {Lecture Notes in Computer Science},
	volume       = {4424},
	pages        = {87--101},
	publisher    = {Springer},
	year         = {2007}}

@article{HJK+22,
	author       = {Christian Hensel and Sebastian Junges and Joost{-}Pieter Katoen and Tim Quatmann and Matthias Volk},
	title        = {The probabilistic model checker Storm},
	journal      = {International Journal on Software Tools for Technology Transfer},
	volume       = {24},
	number       = {4},
	pages        = {589--610},
	year         = {2022}}

@inproceedings{DLT08,
  author       = {Jos{\'{e}}e Desharnais and Fran{\c{c}}ois Laviolette and Mathieu Tracol},
  title        = {Approximate Analysis of Probabilistic Processes: Logic, Simulation and Games},
  booktitle    = {Proceedings of the 5th International Conference on the Quantitative Evaluation of Systems},
  pages        = {264--273},
  publisher    = {{IEEE} Computer Society},
  year         = {2008},
  address      = {Saint-Malo, France},
  month        = sep}

@inproceedings{full,
    author       = {Syyeda Zainab Fatmi and Stefan Kiefer and David Parker and Franck van Breugel},
    editor       = {Ruzica Piskac and Zvonimir Rakamaric},
    title        = {Robust Probabilistic Bisimilarity for Labelled {M}arkov Chains},
    booktitle    = {Proceedings of the 37th International Conference on Computer Aided Verification},
    series       = {Lecture Notes in Computer Science},
    volume       = {15932},
    pages        = {254--275},
    publisher    = {Springer-Verlag},
    year         = {2025},
    address      = {Zagreb, Croatia},
    month        = jul}

@article{arxiv,
    author       = {Syyeda Zainab Fatmi and Stefan Kiefer and David Parker and Franck van Breugel},
    title        = {Robust Probabilistic Bisimilarity for Labelled {M}arkov Chains (Full Version)},
    journal      = {CoRR},
    year         = {2025},
    month        = may,
    doi          = {10.48550/arXiv.2505.15290},
    eprinttype   = {arXiv},
    eprint       = {2505.15290}}

\end{document}